\pdfoutput=1
\documentclass[final]{siamltex}
\usepackage{graphicx}
\usepackage{color}
\usepackage{epsfig}
\usepackage{url}
\usepackage{amsfonts}
\usepackage{amsmath}
\usepackage{textcomp}
\usepackage{paralist}
\usepackage[square,numbers,sort]{natbib}
\usepackage[colorlinks=true,linkcolor=black,citecolor=black,urlcolor=black,
breaklinks=true, pdfstartview=Fit,pdfpagemode=UseOutlines,plainpages=false]{hyperref}

\graphicspath{{figs/wct/examples/},%
    {figs/wct/misc/energyGraphs/},%
    {figs/wct/meshes/},%
    {figs/wct/old/workingNotes/},%
    {figs/wct/apps/}}

\newcommand{\Real}{\ensuremath{\mathbb{R}}} 
\newcommand{\meshM}{\ensuremath{\mathcal{M}}} 
\newcommand{\meshV}{\ensuremath{\mathcal{V}}} 
\newcommand{\meshT}{\ensuremath{\mathcal{T}}} 

\DeclareMathOperator{\aff}{aff}
\DeclareMathOperator{\conv}{conv}

\definecolor{darkgrn}{rgb}{0, 0.75, 0}

\title{Well-Centered Triangulation\footnotemark[5]}
\author{
Evan VanderZee\thanks{Department of Mathematics, 1409 W. Green Street,
University of Illinois at Urbana-Champaign, Urbana, IL 61801
(vanderze@illinois.edu). Research supported by CSE Fellowship from the 
Computational Science and Engineering Program and Applied Mathematics 
Program, University of Illinois and by 
NSF CAREER Award, Grant No. DMS-0645604.}
\and Anil N. Hirani\thanks{Author for Correspondence,  
Department of Computer Science,
201 N. Goodwin Avenue, 
University of Illinois at Urbana-Champaign, Urbana, IL 61801
(hirani@illinois.edu). Research supported by NSF CAREER Award, 
Grant No. DMS-0645604.}
\and Damrong Guoy\thanks{Computational Science and Engineering
Program, Center for Simulation of Advanced Rockets,
University of Illinois at Urbana-Champaign now at
Synopsys Inc., Mountain View, California
(Damrong.Guoy@synopsys.com)}
\and Edgar A. Ramos\thanks{Escuela de Matem\'aticas,
Universidad Nacional de Colombia,
Medell\'in, Colombia
\mbox{(earamosn@unalmed.edu.co)}}
}

\begin{document}

\renewcommand{\thefootnote}{\fnsymbol{footnote}}
\footnotetext[5]{Preliminary results for the 2-dimensional problem
of well-centered planar triangulations appeared previously
in the Proceedings of the 16th International
Meshing Roundtable, Seattle, WA, October 14-17, 2007
\cite{VaHiGuRa2007}.}
\renewcommand{\thefootnote}{\arabic{footnote}}
\maketitle

\begin{abstract}
  Meshes composed of well-centered simplices have nice orthogonal dual
  meshes (the dual Voronoi diagram). This is useful for certain
  numerical algorithms that prefer such primal-dual mesh pairs. We
  prove that well-centered meshes also have optimality properties and
  relationships to Delaunay and minmax angle triangulations. We
  present an iterative algorithm that seeks to transform a given
  triangulation in two or three dimensions into a well-centered one by
  minimizing a cost function and moving the interior vertices while
  keeping the mesh connectivity and boundary vertices fixed. The cost
  function is a direct result of a new characterization of
  well-centeredness in arbitrary dimensions that we present. Ours is
  the first optimization-based heuristic for well-centeredness, and
  the first one that applies in both two and three dimensions. We show
  the results of applying our algorithm to small and large
  two-dimensional meshes, some with a complex boundary, and obtain a
  well-centered tetrahedralization of the cube.  We also show
  numerical evidence that our algorithm preserves gradation and that
  it improves the maximum and minimum angles of acute triangulations
  created by the best known previous method.
\end{abstract}

\begin{keywords}
  well-centered, meshing, mesh optimization, acute, triangulation,
  discrete exterior calculus
\end{keywords}

\begin{AMS}
65N50,
%
65M50,
%
65D18,
%
%
51M04
%
%

\end{AMS}

\pagestyle{myheadings}
\thispagestyle{plain}
\markboth{Evan VanderZee, Anil N. Hirani, Damrong Guoy and Edgar Ramos}
{Well-Centered Triangulation}

\section{Introduction}

A \emph{completely well-centered} mesh is a simplicial mesh in which
each simplex contains its circumcenter in its interior. A
3-dimensional example is a tetrahedral mesh in which the circumcenter
of each tetrahedron lies inside it and the circumcenter of each
triangle face lies inside it. Weaker notions of well-centeredness
  require that simplices of specific dimensions contain their
  circumcenters. In two dimensions, a completely well-centered
triangulation is the same thing as an acute triangulation.

Typical meshing algorithms do not guarantee well-centeredness. For
example, a Delaunay triangulation is not necessarily well-centered. In
this paper we discuss well-centered triangulations, with particular
application to triangle and tetrahedral meshes.  We present an
iterative energy minimization approach in which a given mesh, after
possible preprocessing, may be made well-centered by moving the
internal vertices while keeping the boundary vertices and connectivity
fixed.

A well-centered (primal) mesh has a corresponding dual mesh assembled
from a circumcentric subdivision \cite{Hirani2003}. For an
$n$-dimensional primal mesh, a $k$-simplex in the primal corresponds
to an $(n-k)$-cell in the dual. For example, in a well-centered planar
triangle mesh, the dual of a primal interior vertex is a convex
polygon with boundary edges that are orthogonal and dual to primal
edges. This orthogonality makes it possible to discretize the Hodge
star operator of exterior calculus \cite{AbMaRa1988} as a diagonal
matrix, simplifying certain computational methods for solving partial
differential equations and for topological calculations. Some
numerical methods that mention well-centered meshes in this context
are the covolume method \cite{Nicolaides1992} and Discrete Exterior
Calculus \cite{Hirani2003,DeHiLeMa2005}.

Well-centered meshes are not strictly required for these or other
related methods; however, some computations would be easier if such
meshes were available. For example, a stable mixed method for Darcy
flow has recently been derived using Discrete Exterior Calculus
\cite{HiNaCh2008} and applied to well-centered meshes generated by our
code and to Delaunay meshes. That numerical method passes patch tests
in 2 and 3 dimensions for both homogeneous and heterogeneous problems.
Figure~\ref{fig:Darcy_flow} (reproduced from \cite{HiNaCh2008} by
permission of the authors) shows the velocities from a solution to the
Darcy flow problem in a layered medium. The solution was computed with
that numerical method and a well-centered mesh. \begin{figure}[ht]
  \centering
  \includegraphics[width=260pt, trim=304pt 49pt 307pt 55pt, clip]
  {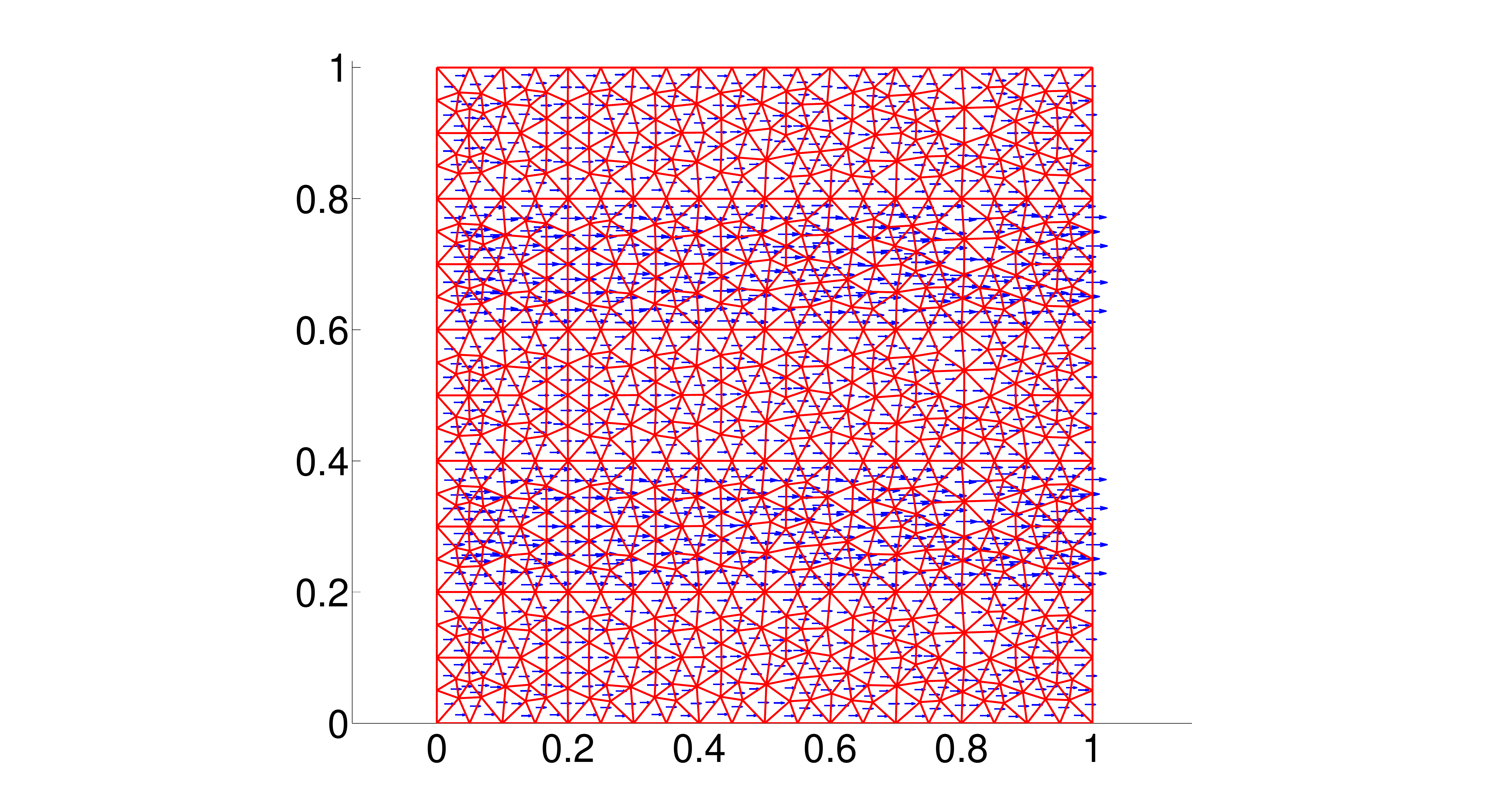}
  \caption{Darcy flow in a medium with 5 layers, computed on a
    well-centered mesh. The odd layers have a permeability of 5 and
    even layers have permeability of 10.  The velocities in the odd
    and even layers should be different and should have no vertical
    component, as shown. The mesh was created using our code.  Figure
    taken from \cite{HiNaCh2008}, used by permission from authors.}
  \label{fig:Darcy_flow}
\end{figure}

In the case of covolume methods applied to Maxwell's equations, a
justification for well-centered triangulation is given
in~\cite{SaHaMoWe2006,SaWaHaMoWe2006,SaHaMoWe2006a,SaHaMoWe2007}.

Another example from scientific computing is space-time meshing.  When
tent-pitching methods for space-time meshing were first introduced,
the initial spatial mesh was required to be acute, which for
two-dimensional meshes is the same thing as being well-centered
\cite{UnSh2002}. More recently this requirement has been avoided,
although at the expense of some optimality in the construction
\cite{ErGuSuUn2002}.

In two dimensions, well-centered meshes achieve optimality in
\emph{two} objectives that are important in some applications.  If a
planar point set has a well-centered triangulation,
that triangulation both
\emph{minimizes} the \emph{maximum} angle and \emph{maximizes} the
\emph{minimum} angle. We don't know any generalizations of this double
optimality to higher dimensions, but it is known that in any dimension
if the convex hull of a point set has a well-centered triangulation,
then that triangulation is unique and it is the Delaunay triangulation
\cite{Rajan1994}.

\section{Our Results} We \emph{characterize} well-centered
triangulations in arbitrary dimensions,
\emph{prove optimality} results for two-dimensional
well-centered triangulations, and give many experimental results.

The new characterization of well-centeredness that we give here is a
useful theoretical tool that allows us to relate well-centeredness and
Delaunay triangulation in arbitrary dimensions. In addition, it is
also a practical tool since it presents, for the first time, a path to
the creation of higher-dimensional well-centered triangulations.  Even
the formulation of an optimization approach for higher-dimensional
well-centeredness would be difficult without such a characterization.
Indeed, ours is the first algorithm
to even consider using an optimization approach
to seek well-centeredness. This approach allows us
both to improve existing triangulations in $\Real^{2}$
and to create well-centered triangulations in $\Real^{3}$.
We also prove optimality results about our cost function
and optimality results that relate well-cen\-tered\-ness to well-known
triangulation schemes. The specific results are enumerated below.

\begin{inparaenum}[(a)]
\item We introduce a new characterization of well-centeredness in
  arbitrary dimensions (Thm.~\ref{thm:characterize}). \item As a
  simple corollary (Cor.~\ref{cor:Delaunay}) we show that for any
  dimension $n$, an $n$-well-centered triangulation of a convex subset
  of $\Real^n$ is Delaunay, which is a new proof of a result in
  \cite{Rajan1994}.  \item Using the characterization of
  Thm.~\ref{thm:characterize} we define a family of cost functions
  $E_p$ (equation~\ref{eq:genE_p}) suitable for creating well-centered
  triangulations in arbitrary dimensions. \item With these we design
  an algorithm that optimizes meshes with the goal of producing
  well-centered meshes.  The
  algorithm generalizes our previous angle-based optimization in two
  dimensions, described in \cite{VaHiGuRa2007}.  \emph{Ours is the
    first known strategy
    for well-centeredness that generalizes to
    higher dimensions}.
\item Using the algorithm we produce a
  well-centered triangulation of a cube
  (Fig.~\ref{fig:cube_430}).
\item We show several two dimensional examples,
  including one with more than 60000 triangles
  (middle of Fig.~\ref{fig:geomeshes}).
\item  In two dimensions, every algorithm
  proven to generate acute
  triangulations may produce angles arbitrarily close to $\pi/2$.
  Moreover, in all cases we have tried, {\emph{our optimization
  algorithm can improve the quality of planar acute-angled
  triangulations produced by other heuristics for creating
  acute triangulations}}.
    A challenging example is shown in
  Fig.~\ref{fig:superior_1388}.
\item We also demonstrate numerically
  that \emph{graded triangulations maintain their gradation while
    being processed by our algorithm} (Fig.~\ref{fig:titan2D},
  \ref{fig:grdsquare_966}, \ref{fig:superior_1388}). This is useful
  since producing provably acute graded triangulations is an open
  problem.
\item For planar triangulations, we show that the \emph{minmax
    triangulation \cite{EdTaWa1992} is the optimal triangulation with
    respect to our energy} $E_{\infty}$ (Cor.~\ref{cor:minmaxisbest}).
\item 
We give a different proof for the acute angle case of a result
from~\cite{BeEp1995}; we show that if a planar point set admits
a~$2$-well-centered triangulation, then that triangulation is
the unique Delaunay triangulation and
the unique minmax triangulation of the point set
(Thm.~\ref{thm:uniquewct}).
\end{inparaenum}

Our \emph{experimental} results in three dimensions are rudimentary,
although even these
were not available before our work.  The difficulty in three
dimensions lies further upstream, in a step that precedes the
application of our optimization algorithm. In the planar case, an
interior vertex with four neighbors must be incident to an obtuse
triangle, but some simple connectivity preprocessing can fix this
problem \cite{VaHiGuRa2007}. Similarly, a tetrahedral mesh may have
topological obstructions to well-centeredness.  The topological
obstructions in this case, however, are not yet fully understood.
Some progress has been made in our other work \cite{VaHiGuRaZh2008} by
studying the link of (topological sphere around) a vertex, but much
remains to be done.  The techniques used to study such topological
obstructions are interesting, but they are transversal to this paper.

\section{Previous Results}
\label{sec:prevWork}
We are concerned with triangulations for which the domain is specified
by a polygonal or polyhedral boundary.  Our main objective is
obtaining well-centered triangle and tetrahedral meshes.  Relevant
work can be divided into constructive and iterative approaches.

Constructive approaches start with specified input constraints and
generate additional points, called Steiner points, and/or a
corresponding triangulation. Normally a point is committed to a
position and never moved afterwards. An algorithm for nonobtuse planar
triangulations based on circle packings is described in
\cite{BeMiRu1994}.  More recent works describe improved constructions
for nonobtuse triangulations while also describing how to derive an
acute triangulation from a nonobtuse one \cite{Maehara2002,
  Yuan2005}. There are two major difficulties with such
  methods. The first is that these algorithms aim to achieve a
  triangulation of size linear in the input size. As a result, the
  \emph{largest and smallest angles can be arbitrarily close to
    $\pi/2$ and 0 respectively}. The second major difficulty with
  these algorithms is that they \emph{do not offer a clear path
    towards a higher-dimensional generalization.} Moreover, we are
not aware of any existing implementations of these algorithms, which
seem to be primarily of theoretical interest.  As recently as 2007,
Erten and \"Ung\"or \cite{ErUn2007} proposed a variant of the Delaunay
refinement algorithm for generating acute triangulations of planar
domains.  This
heuristic, which relocates Steiner points after they
are added, has been implemented and appears to work quite well.
Experiments suggest, however, that the maximum angle in the output is
often near $\pi/2$, and our method is able to improve their meshes.
See, for example, the mesh of Lake Superior in
Section~\ref{sec:results}.

There is also a constructive algorithm that achieves a well-centered
quality triangulation of a point set \cite{BeEpGi1994} (with no
polygonal boundary specified), and an algorithm for constructing
nonobtuse quality triangulations \cite{MeSo1992}.  Also relevant is an
algorithm that, given a constraint set of both points and segments in
the plane, finds a triangulation that minimizes the maximum angle
\cite{EdTaWa1992}, without adding points.  If an acute triangulation
exists for the input constraints, the algorithm will find one.  The
most promising of the constructive algorithms is probably
\cite{ErUn2007} mentioned above. But for this algorithm, as well as
for the others mentioned in this paragraph, we are not aware of
higher-dimensional generalizations.

Yet another approach is the mesh stitching approach
in~\cite{SaHaMoWe2006,SaHaMoWe2007,SaWaHaMoWe2006}. In this scheme, the
region near the boundary and the interior far from boundary are meshed
seperately and these two regions are stitched with a special
technique. However, in three dimensions, the method is unable to
generate a well-centered triangulation in their
examples~\cite{SaHaMoWe2006}.

On the other hand, there are iterative or optimization approaches
which allow an initial triangulation (possibly the canonical Delaunay)
and then move the points while possibly changing the connectivity.
These algorithms often apply in three dimensions as well as two.
Moreover, there are many well-known existing meshing algorithms, some
of which generate quality triangulations \cite{Ruppert1995,
  Edelsbrunner2001} and have reliable implementations. An iterative
approach can start from an existing high-quality mesh and seek to make
it well-centered while retaining its high quality.

In the class of
iterative approaches there are optimization methods like centroidal
Voronoi tessellations \citep{DuFaGu1999, DuGuJu2003, DuWa2005},
variational tetrahedral meshing \cite{AlCoYvDe2005}.
Each of these methods has a global cost function that
it attempts to minimize
through an iterative procedure that alternates between updating the
location of the mesh vertices and the triangulation of those vertices.
Our algorithm has some similarities to these methods, but uses a cost
function explicitly designed to seek well-centered simplices, in
contrast to the cost functions optimized in \cite{DuFaGu1999} and
\cite{AlCoYvDe2005}.

There are also many iterative optimization methods
that, like our method, relocate vertices without changing the
initial mesh connectivity.  Traditional Laplacian
smoothing\cite{Winslow1964} is one such method.  Such methods
improve meshes according to some criteria, but do not
typically produce well-centered meshes.  (See, for example,
our comparisons with Laplacian smoothing in
Sections~\ref{subsec:lake} and~\ref{subsec:3dresult}.)

In addition to optimization approaches that work directly with a mesh,
there are several algorithms that generate circle packings or circle
patterns by optimizing the radii of the circles.  In particular, the
algorithms for creating circle patterns that were proposed in
\cite{CoSt2003} and \cite{BoSp2004} can be adapted to create
triangulations.  These algorithms produce circle patterns that have
specified combinatorics, but they do not permit a complete
specification of the domain boundary, so they are not appropriate to
our purpose.

The problem of generating a well-centered tetrahedralization in
$\Real^{3}$ is considerably harder than the two-dimensional
analogue. A complete characterization of the topological obstructions
to well-centeredness in three dimensions is still an open problem,
although a start has been made in our work
elsewhere~\cite{VaHiGuRaZh2008}. Similarly, the problem of generating
a three-dimensional acute triangulation---a~tetrahedralization in
which all the dihedral angles are acute---is more difficult than
generating a two-dimensional acute triangulation. For tetrahedra, it
is no longer true that well-centeredness and acuteness are equivalent
\cite[Section 2]{VaHiGu2008}.
In addition, acute tetrahedralizations
are known for only restricted domains.  For example,
until recently it was not
known whether the cube has an acute triangulation.  The
construction that showed the cube does have an acute
triangulation made use of the well-centered optimization
discussed in this paper~\cite{VaHiZhGu2009}.

\section{Characterization of Well-Centeredness}

We begin with a new characterization of well-centeredness in arbitrary
dimension. This characterization allows us to create 
an algorithm, described in Section~\ref{sec:iterative},
that uses optimization to seek well-centeredness.
It also serves, later in the current
section, as a theoretical tool in relating arbitrary-dimensional
well-centeredness to Delaunay triangulations.

Consider an $n$-dimensional simplex $\sigma^{n}$ embedded in Euclidean
space $\Real^{m}$, $m \ge n$.  The affine hull of $\sigma^{n}$,
$\aff(\sigma^{n})$, is the smallest affine subspace of $\Real^{m}$
that contains $\sigma^{n}$.  In this case, $\aff(\sigma^{n})$ is a
copy of $\Real^{n}$ embedded in $\Real^{m}$.  The circumcenter of
$\sigma^{n}$, which we denote $c(\sigma^{n})$, is the unique point in
$\aff(\sigma^{n})$ that is equidistant from every vertex of
$\sigma^{n}$.

For an $n$-simplex $\sigma^{n}$ with $n \ge 3$, it is possible for
$\sigma^{n}$ to contain its circumcenter $c(\sigma^{n})$ while some
proper face $\sigma^{p} \prec \sigma^{n}$ does not contain its
circumcenter $c(\sigma^{p})$.  It is also possible that for all $1 \le
p < n$ and all $\sigma^{p} \prec \sigma^{n}$, $c(\sigma^{p})$ lies in
the interior of $\sigma^{p}$, but $\sigma^{n}$ does not contain its
circumcenter.  (See \cite{VaHiGu2008} for examples with $n = 3$.)
Thus we say that an $n$-simplex $\sigma^{n}$ is a
{\emph{$(p_1,\ldots,p_k)$-well-centered simplex}} if for $p_i$, $i =
1,\ldots,k$, all faces of $\sigma^{n}$ of dimension $p_i \le n$
properly contain their circumcenters.  The parentheses are suppressed
when referring to only one dimension.  A simplex $\sigma^{n}$ is
{\emph{completely well-centered}} if it is
$(1,2,\ldots,n-1,n)$-well-centered.


In this section we give an alternate characterization for an
$n$-simplex $\sigma^{n}$ that is $n$-well-centered.
The characterization, which shows how the $n$-well-centered
$n$-simplex generalizes the acute triangle to higher dimensions,
%
%
uses the concept of an equatorial ball,
which we now define.

Let $\sigma^{n}$ be a simplex embedded in a
hyperplane $P^{m}$ with $m > n$.  The {\emph{equatorial ball}} of
$\sigma^{n}$ in $P^{m}$ is the closed ball $\{x \in P^{m}:\lvert x -
c(\sigma^{n})\rvert \le R(\sigma^{n})\}$, where $c(\sigma^{n})$ is the
circumcenter of $\sigma^{n}$, $R(\sigma^{n})$ its circumradius, and
$\lvert \cdot \rvert$ the standard Euclidean norm.  In this paper we
use the notation $B(\sigma^{n})$ for the equatorial ball of
$\sigma^{n}$.  The notation is used in the context of $\sigma^{n}
\prec \sigma^{n+1}$, and the hyperplane $P^{m}$
is understood to be $\aff(\sigma^{n+1})$.
The equatorial ball is an extension of the circumball into higher
dimensions; it is assumed throughout this paper that the
{\emph{circumball}} and {\emph{circumsphere}} of a simplex
$\sigma^{n}$ are embedded in $\aff({\sigma^{n}})$.  Note that here and
throughout the paper we have implicitly assumed that an $n$-simplex is
fully $n$-dimensional, though when a simplicial mesh is represented on
a computer it may be the case that some of the simplices are
degenerate.

\bigskip

\begin{theorem}
\label{thm:characterize}
The $n$-simplex $\sigma^{n} = v_{0}v_{1}\ldots v_{n}$ is
$n$-well-centered if and only if for each $i = 0,1,\ldots,n$, vertex
$v_{i}$ lies strictly outside $B^{n}_{i} :=
B(v_{0}v_{1}\ldots{v}_{i-1}v_{i+1}\ldots v_{n})$.
\end{theorem}
\begin{proof}
%
  Figure~\ref{fig:eqballsproof} illustrates this proof in dimension
  $n = 2$.  It may help the reader understand the notation used in
  the proof and give some intuition for what the proof looks like in
  higher dimensions.

  First we suppose that $\sigma^{n}$ is $n$-well-centered.  Let
  $S^{n-1} = S^{n-1}(\sigma^{n})$ be the circumsphere of $\sigma^{n}$.
  Now $\aff(\sigma^{n})$ is a copy of $\Real^{n}$, and within that
  copy of $\Real^{n}$, $\sigma^{n}$ is an intersection of half-spaces.
  Considering some particular vertex $v_{i}$ of $\sigma^{n}$, we know
  that one
  of the bounding hyperplanes of $\sigma^{n}$ is the hyperplane
  $P^{n-1}_{i}$ that contains the simplex $\sigma^{n-1}_{i} =
  v_{0}v_{1}\ldots{v}_{i-1}v_{i+1}\ldots v_{n}.$

  Hyperplane $P^{n-1}_{i}$ partitions our copy of $\Real^{n}$ into two
  half-spaces --- an open half-space $H^{n}_{i}$ that contains the
  interior of $\sigma^{n}$ and vertex $v_{i}$, and a closed half-space
  that contains $\sigma^{n-1}_{i}$ (on its boundary).

  \begin{figure}
    \centering
    \scalebox{0.8}{%
      \includegraphics[width=250pt, trim = 126pt 237pt 108pt 225pt, clip]
      {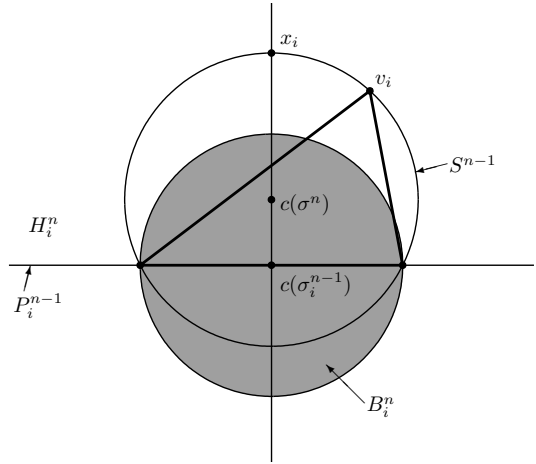}%
      \begin{picture}(0, 0)
        \put(-121, 82){$c(\sigma^{n-1}_{i})$}
        \put(-121, 120){$c(\sigma^{n})$}
        \put(-121, 199){$x_{i}$}
        \put(-76, 180){$v_{i}$}
        \put(-247, 72){$P^{n-1}_{i}$}
        \put(-242, 82){\vector(1,4){3}}
        \put(-240, 110){$H^{n}_{i}$}
        \put(-41, 138){$S^{n-1}$}
        \put(-42, 142){\vector(-4,-1){15}}
        \put(-80, 25){$B^{n}_{i}$}
        \put(-81, 29){\vector(-1,1){18}}
      \end{picture}}
    \caption{An illustration of the proof of
      Theorem~\ref{thm:characterize} in two dimensions.  In an
      $n$-well-centered simplex $\sigma^{n}$, vertex $v_{i}$ and
      circumcenter $c(\sigma^{n})$ lie in the same open half-space
      $H^{n}_{i}$, the region where circumsphere $S^{n-1}$ lies
      outside equatorial ball $B^{n}_{i}$.}
    \label{fig:eqballsproof}
  \end{figure}

  Because $\sigma^{n}$ is well-centered, $c(\sigma^{n})$ lies in its
  interior.  Thus $c(\sigma^{n})$ lies in $H^{n}_{i}$, the open
  half-space that contains $v_{i}$.  Consider, then, the line through
  $c(\sigma^{n})$ and $c(\sigma^{n-1}_{i})$.
  Within $H^{n}_{i}$, this line intersects $S^{n-1}$ at a point
  $x_{i}$ with $\lvert x_{i} - c(\sigma^{n})\rvert = R(\sigma^{n}).$
  Moreover, $\lvert x_{i} - c(\sigma^{n-1}_{i})\rvert > R(\sigma^{n})
  > R(\sigma^{n-1}_{i})$.  We see that $x_{i}$ lies outside
  $B^{n}_{i}$ and conclude that $S^{n-1} \cap H^{n}_{i}$ lies outside
  $B^{n}_{i}$.  In particular, since $v_{i} \in S^{n-1} \cap
  H^{n}_{i}$, we know that $v_{i}$ lies outside $B^{n}_{i}$.  Since
  $v_{i}$ was chosen arbitrarily, we conclude that $v_{i}$ lies
  outside $B^{n}_{i}$ for each $i = 0,1,\ldots,n$, and necessity is
  proved.

  For sufficiency we consider an $n$-simplex $\sigma^{n}$ such that
  $v_{i}$ lies outside $B^{n}_{i}$ for each $i=0,1,\ldots,n$.  We will
  show that the circumcenter $c(\sigma^{n})$ lies in the interior of
  $\sigma^{n}$ by demonstrating that for each vertex $v_{i}$,
  $c(\sigma^{n})$ lies in $H^{n}_{i}$.  We know that $P^{n-1}_{i}$
  cuts $S^{n-1}$ into a part inside $B^{n}_{i}$ and a part outside
  $B^{n}_{i}$, and we have just established that whichever of the
  (open) half-spaces contains $c(\sigma^{n})$ is the half-space where
  $S^{n-1}$ lies outside $B^{n}_{i}$.  Since we are given that $v_{i}
  \in S^{n-1}$ lies outside $B^{n}_{i}$, we know that $v_{i}$ and
  $c(\sigma^{n})$ must lie in the same open half-space $H^{n}_{i}$.
  This holds for every $v_{i}$, so $c(\sigma^{n})$ is in the interior
  of $\sigma^{n}$, and $\sigma^{n}$ is, by definition,
  $n$-well-centered.
\end{proof}

\bigskip

Figure~\ref{fig:eqballsxmpl} shows how Thm.~\ref{thm:characterize} can
be applied to a tetrahedron.  In Fig.~\ref{fig:eqballsxmpl} we see
that for each vertex $v_{i}$ of the tetrahedron, $v_{i}$ lies outside
of equatorial ball $B^{n}_{i}$.  By Thm.~\ref{thm:characterize} we can
conclude that the tetrahedron is $3$-well-centered, even though we
have not precisely located its circumcenter.  This clearly generalizes
the acute triangle; the angle at vertex $v_{i}$ of a triangle is acute
if and only if $v_{i}$ lies outside $B^{n}_{i}$, and a triangle is
$2$-well-centered if and only if each of its angles is acute.

\begin{figure}
  \centering
  \scalebox{0.95}{%
    \begin{minipage}[c]{93pt}
      \vspace{35pt}
      \includegraphics[width=93pt, trim=169pt 38pt 146pt 144pt, clip]
      {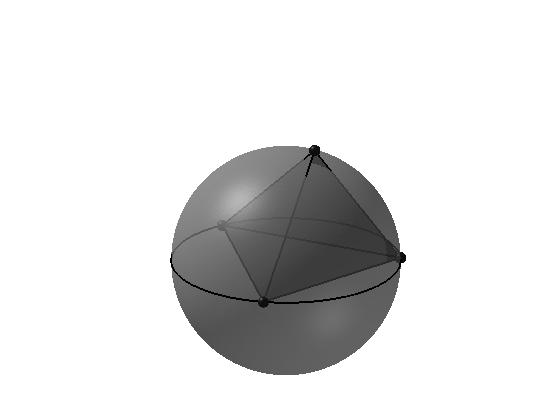}%
    \end{minipage}%
    \hspace{10pt}%
    \begin{minipage}[c]{81pt}
      \vspace{24pt}
      \includegraphics[width=81pt, trim=222pt 81pt 141pt 149pt, clip]
      {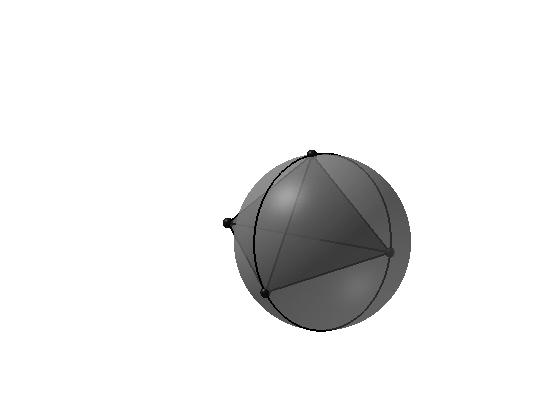}%
    \end{minipage}%
    \hspace{10pt}%
    \begin{minipage}[c]{101pt}
      \vspace{14.5pt}
      \includegraphics[width=101pt, trim=106pt 26pt 111pt 98pt, clip]
      {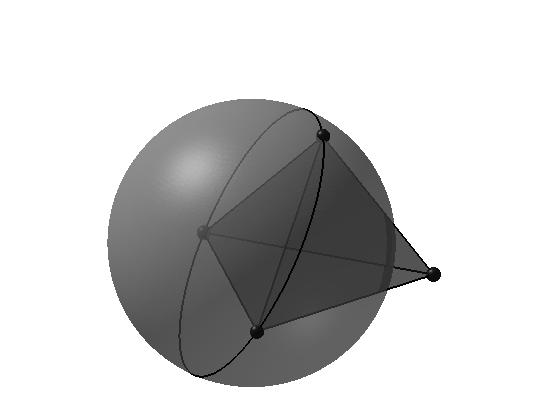}%
    \end{minipage}%
    \hspace{10pt}%
    \begin{minipage}[c]{85pt}
      \includegraphics[width=85pt, trim=175pt 65pt 113pt 85pt, clip]
      {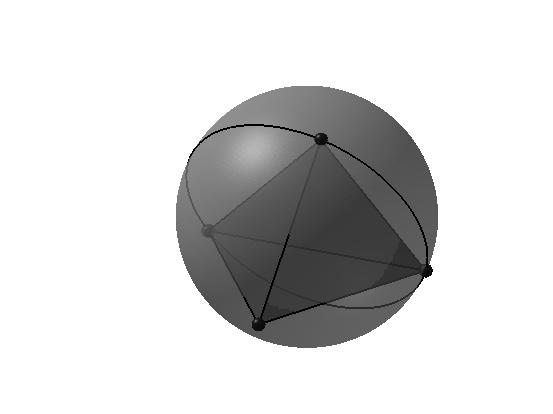}%
    \end{minipage}%
  }
  \caption{One characterization of $n$-well-centeredness of an
    $n$-simplex $\sigma^{n}$ is that for each vertex~$v_{i}$
    of~$\sigma^{n}$, $v_{i}$ lies outside of the equatorial ball
    $B^{n}_{i}$ of the facet $\sigma^{n}_{i}$ opposite $v_{i}$.}
\label{fig:eqballsxmpl}
\end{figure}

When we say that a mesh is a {\emph{$(p_1,\ldots,p_k)$-well-centered
    mesh}}, we mean that every element of the mesh is a
$(p_{1},\ldots,p_{k})$-well-centered simplex.  In the proof of
Thm.~\ref{thm:characterize} we showed that for each face
$\sigma^{n-1}_{i}$ of an $n$-well-centered $n$-simplex $\sigma^{n}$,
the hyperplane $\aff(\sigma^{n-1}_{i})$ cuts the circumball of
$\sigma^{n}$ into two pieces, one piece contained in $B^{n}_{i}$ and
the other piece lying on the same side of $\aff(\sigma^{n-1}_{i})$ as
the interior of $\sigma^{n}$.  It follows that the circumball of
$\sigma^{n}$ is contained in $\left(\bigcup_{i} B^{n}_{i}\right) \cup
\sigma^{n}$.  (It can be shown, in fact, that $\sigma^{n} \subset
\bigcup_{i} B^{n}_{i}$, but we do not need that result here.)
Moreover, if we consider some other $n$-well-centered
  $n$-simplex $\tau^{n}$ such that $\sigma^{n-1}_{i} = \tau^{n}
  \cap \sigma^{n}$, and if vertex $u$ is the vertex of
$\tau^{n}$ opposite $\sigma^{n-1}_{i}$, then
Thm.~\ref{thm:characterize} implies that $u$ is outside $B^{n}_{i}$.
Thus $u$ also lies outside the circumball of $\sigma^{n}$.  If the
underlying space of the mesh is a convex subset of $\Real^{n}$, we can
conclude that the mesh is locally Delaunay.  Since in any dimension a
locally Delaunay mesh is globally Delaunay~\cite{EdSh1996}, we obtain
a new proof of the following result, which was originally proved by
Rajan \cite{Rajan1994}.

\bigskip

\begin{corollary}
  \label{cor:Delaunay}
  If a simpicial mesh of a convex subset of $\Real^{n}$ is
  $n$-well-centered, then the mesh is a Delaunay triangulation of its
  vertices.
\end{corollary}

\bigskip

The converse, of course, is not true.  Section~\ref{sec:globaltop}
gives more details for the planar case.

\section{Iterative Energy Minimization}
\label{sec:iterative}
Given a simplicial mesh, we seek to make the mesh well-centered by
minimizing a cost function defined over the mesh. We'll refer to the
cost function as \emph{energy}. Our method is somewhat similar to the
methods of \cite{AlCoYvDe2005} and \cite{DuFaGu1999} in that it uses
an iterative procedure to minimize an energy defined on the mesh, but
for reasons discussed in Section~\ref{sec:globaltop}, it differs in that
the mesh connectivity and boundary vertices remain fixed as the energy
is minimized.  Also, in contrast to the methods of \cite{AlCoYvDe2005}
and \cite{DuFaGu1999}, the cost function we minimize is explicitly
designed to achieve the aim of well-centeredness.  This section
describes the energy we minimize, which is the main component of our
method.

Before describing the energy we note that at times the mesh
connectivity or boundary vertices of an initial mesh are defined in
such a way that no well-centered mesh exists.  For such cases one can
apply a preprocessing algorithm to update the mesh connectivity.
Section~\ref{sec:globaltop} discusses this problem in more detail.

In the proof of Thm.~\ref{thm:characterize} we see that in order for a
simplex $\sigma^{n}$ to be $n$-well-centered, the circumcenter
$c(\sigma^{n})$ must lie on the same side of facet $\sigma^{n-1}_{i}$
as vertex $v_{i}$.  To convert this discrete variable into something
quantitative we introduce the function $h(v_{i}, \sigma^{n})$, the
signed distance from $c(\sigma^{n})$ to $\aff(\sigma^{n-1}_{i})$ with
the convention that $h(v_{i}, \sigma^{n}) > 0$ when $c(\sigma^{n})$
and $v_{i}$ are on the same side of $\aff(\sigma^{n-1}_{i})$.  The
magnitude of $h(v_{i}, \sigma^{n})$ can be computed as the distance
between $c(\sigma^{n})$ and $c(\sigma^{n-1}_{i})$, and its sign can be
computed by testing whether $c(\sigma^{n})$ and $v_{i}$ have the same
orientation with respect to $\aff(\sigma^{n-1}_{i})$.
A mesh is $n$-well-centered if and only
if $h(v_{i}, \sigma^{n}) > 0$ for every vertex
$v_{i}$ of every $n$-simplex $\sigma^{n}$ of the mesh.

We divide the quantity $h(v_{i}, \sigma^{n})$ by the circumradius
$R(\sigma^{n})$ to get a quantity that does not depend on the size of
the simplex $\sigma^{n}$.  We expect a cost function based on
$h(v_{i}, \sigma^{n})/R(\sigma^{n})$ to do a better job than the basic
$h(v_{i}, \sigma^{n})$ at preserving properties of the initial mesh.
In particular, the grading (relative sizes of the elements) of the
initial mesh should be preserved better with $h/R$ than with $h$.
Sazonov et al. have also noticed that cost functions based on the
quantity $h/R$ may be helpful in quantifying
well-centeredness~\cite{SaHaMoWe2006}.


%

Note that $-1 < h(v_{i}, \sigma^{n})/R(\sigma^{n}) < 1$ for finite
$\sigma^{n},$ because $R(\sigma^{n})^2 = h(v_{i}, \sigma^{n})^2 +
R(\sigma_{i}^{n-1})^2$.  Instead of using the quantity $h/R$ directly,
we consider the function
\[
f_{n}(\sigma^{n}) = \max_{\mathrm{vertices~}v \in \sigma^{n}}
    \left\lvert \frac{h(v, \sigma^{n})}{R(\sigma^{n})} - k_{n}\right\rvert,
\]
where $0 < k_{n} \le 1$ is a constant that may depend on the dimension
$n$ of the simplex.  The advantage of minimizing $f_{n}$ as opposed to
maximizing $h/R$ is that if $k_{n}$ is chosen properly, the measure
penalizes simplex vertices where $h/R$ approaches $1$ (e.g., small
angles of triangles and sharp points of needle tetrahedra) as well as
vertices where $h/R \le 0$.

We want to choose $k_{n}$ so that $f_{n}(\sigma^{n})$ is minimized
when $\sigma^{n}$ is the regular $n$-simplex.  Taking $k_{n} = 1/n$
may seem like a good choice because it is clear that the regular
simplex minimizes $f_{n}$.  (When $k_{n} =1/n$, $f_{n}(\sigma^{n}) =
0$ for the regular $n$-simplex $\sigma^{n}$).  We show in
Lemma~\ref{lemma:regsmplxmin}, however, that the regular simplex
minimizes $f_{n}$ for any $1 \ge k_{n} \ge 1/n$.

\bigskip

\begin{lemma}
  For $k_{n} \ge 1/n$, the measure $f_{n}(\sigma^{n})$ is minimized
  when $\sigma^{n}$ is a regular simplex.
\label{lemma:regsmplxmin}
\end{lemma}
\begin{proof}
  Suppose that $k_{n} \ge 1/n$.  For the regular simplex, then,
  $f_{n}(\sigma^{n}) = k_{n} - 1/n$.  Thus it suffices to show
  that for any simplex $\sigma^{n}$ there exists a
  vertex $v$ such that $h(v, \sigma^{n}) \le R(\sigma^{n})/n$;
  at such a vertex we have
  \[
  \left\lvert \frac{h(v, \sigma^{n})}{R(\sigma^{n})} - k_{n}\right\rvert
  = k_{n} - \frac{h(v, \sigma^{n})}{R(\sigma^{n})}
  \ge k_{n} - \frac{1}{n}\,.
  \]
  
  We have seen that for a simplex that is not
  $n$-well-centered, there exists a vertex $v$ with $h(v,\sigma^{n})
  \le 0$, so it remains to prove this for simplices that are
  $n$-well-centered.

  Suppose $\sigma^{n}$ is $n$-well-centered.  Let $h := \min_{i}
  h(v_{i},\sigma^{n}).$ Consider a sphere $S^{n-1} \subset
  \aff(\sigma^{n})$ with center $c(\sigma^{n})$ and radius $h$.  We
  claim that $\sigma^{n}$ contains the sphere $S^{n-1}$.  Indeed, for
  each facet $\sigma^{n-1}_{i}$ of $\sigma^{n}$, since the radius of
  $S^{n-1}$ is $h \le h(v_{i}, \sigma^{n})$ we have that the sphere
  $S^{n-1}$ is contained in the same half space as $c(\sigma^{n})$ and
  $v_{i}$.  Thus the sphere is contained in the intersection of half
  spaces that defines the simplex, i.e., is contained in the simplex.

  It follows, then, that $h \le r(\sigma^{n})$ where $r(\sigma^{n})$
  is the inradius of $\sigma^{n}$.  We know that $h/R \le r/R \le 1/n$
  and that equality is achieved for only the regular simplex.  (The
  inequality $r/R \le 1/n$ is proved in \cite{KlTs1979}, among
  others.)
\end{proof}

\bigskip

%
%

In light of Lemma~\ref{lemma:regsmplxmin}, taking $k_{n} = 1/2$,
independent of $n$, is a good strategy, because for $k_{n} = 1/2$ the
cost function $f_{n}$ will prefer any $n$-well-centered simplex to any
simplex that is not $n$-well-centered, and among all $n$-well-centered
simplices, $f_{n}$ will prefer the regular simplex over all others.
We use $k_{n} = 1/2$ for all of the results discussed in
Section~\ref{sec:results}.

For $k_{n} > 0$ the objective of $n$-well-centeredness is achieved
when $\lvert h/R - k_{n}\rvert < k_{n}$ at every vertex of every
simplex $\sigma^{n}$.  (Note that this is not a necessary condition if
$k_{n} < 1/2$.)
Our goal, then, is to minimize $\left\lvert h/R - k_{n}\right\rvert$
over all vertices and all simplices, driving it below $k_{n}$ at every
vertex of every simplex.  It could be effective to work directly with
\begin{equation}
  \label{eq:genE_infty}
  E_{\infty}\left(\meshM\right) = E_{\infty}\left(\meshV, \meshT\right)
  = \max_{\substack{\text{simplices~}\sigma^{n} \in \meshT\\
      \text{vertices~}v_{i} \in \sigma^{n} \cap \meshV}}
  \left\lvert \frac{h(v_{i}, \sigma^{n})}{R(\sigma^{n})}
    - \frac{1}{2}\right\rvert\, ,
\end{equation}
but we choose instead to minimize an approximation to $2 E_{\infty}$
given by
\begin{equation}
  \label{eq:genE_p}
  E_{p}\left(\meshM\right) = E_{p}\left(\meshV, \meshT\right) =
  \sum_{\substack{\sigma^{n} \in \meshT\\
      v_{i} \in \sigma^{n} \cap \meshV}}
  \left\lvert \frac{2 h(v_{i}, \sigma^{n})}{R(\sigma^{n})}
    - 1\right\rvert^{p},
\end{equation}
where $p$ is a parameter.  $\meshM$ here stands for a mesh consisting
of vertices $\meshV$ with particular coordinates and a connectivity
table $\meshT$ that describes which groups of vertices form simplices.
Note that $\lim_{p\to\infty}
\left(E_{p}\left(\meshM\right)\right)^{1/p} = 2
E_{\infty}\left(\meshM\right)$, so $E_{p}(\meshM)$ is indeed an
approximation to $2 E_{\infty}(\meshM)$.  The factor of $2$ is
included for numerical robustness.  The parameter $p$ influences the
relative importance of the worst vertex-simplex pair compared to the
other vertex-simplex pairs in computing the quality of the mesh as a
whole.  It is convenient to choose $p$ as a positive even integer,
since the absolute value need not be taken explicitly in those cases.

As stated, the measure $E_{p}(\meshM)$ leaves some ambiguity in the
case of a degenerate simplex, which may occur in a computational
setting.  For several reasons, including a desire to maintain upper
semicontinuity of the cost function, we
use the convention that any degenerate
simplex, even one with coincident vertices, has its circumcenter at
infinity and $h/R = -1$.

\begin{figure}
  \centering
  \resizebox{1.8 in}{!}{
    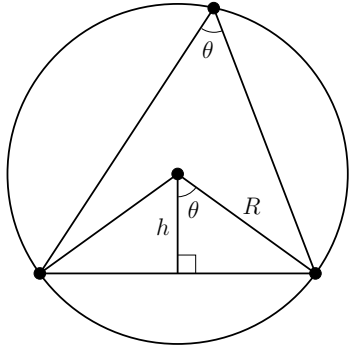}
  \caption{For a triangle, $h/R = \cos(\theta).$}
  \label{fig:genE_p}
\end{figure}
 
Figure~\ref{fig:genE_p} shows the quantities $h$ and $R$ in a sample
triangle.
We see in the figure that $\cos(\theta) = h/R$.
Thus~\eqref{eq:genE_p} is a generalization of the energy
\begin{equation}
  \label{eq:Ep}
  E_p(\meshM) = E_p\left(\meshV, \meshT\right) = \sum_{\theta \in \meshM}
  \left\lvert2\cos(\theta) - 1\right\rvert^p,
\end{equation}
which is a constant multiple of the energy the authors proposed
earlier for achieving well-centeredness of planar triangle meshes
\cite{VaHiGuRa2007}.  In three dimensions the quantity $h/R$ is
related to the cosine of the tetrahedron vertex angle, as discussed
in~\cite{SaHaMoWe2006}.

The cost functions $E_{p}$ and $E_{\infty}$ are not convex.  When
designing a cost function for mesh optimization, one might hope to
develop a function that is convex, or, if not convex, at least one
that has a unique minimum.  It is, however, not possible to define an
energy that accurately reflects the goals of well-centered meshing and
also has a unique minimum.  Consider the mesh shown on the left in
Fig.~\ref{fig:notCnvx}.  We suppose that the boundary vertices are
fixed, but the interior vertex is free to move.  We want to decide
where to move the interior vertex in order to obtain a well-centered
mesh.  The right side of Fig.~\ref{fig:notCnvx} shows
where the free vertex can be placed to produce a well-centered mesh.
The light gray regions are not allowed because placing the free vertex
in those regions would make some boundary angle nonacute.  (The dotted
lines indicate how the four most important boundary angles influence
the definition of this region.)  The darker gray regions, shown
overlaying the light gray region, are not permitted because placing
the interior vertex in those regions would make some angle at the
interior vertex nonacute.

\begin{figure}
  \centering
  \includegraphics[width=2.1in, trim=3.1in 4.7in 3.0in 4.7in, clip]
  {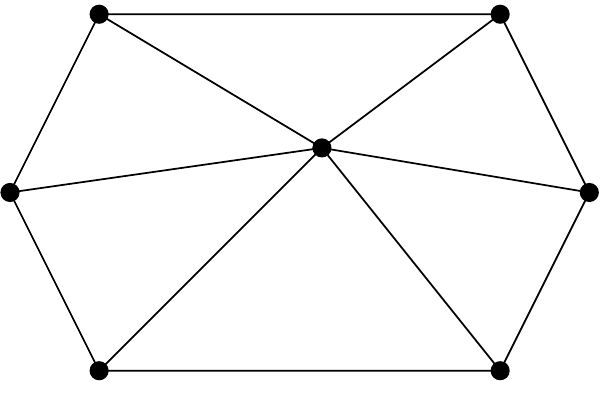}
  \quad\qquad
  \includegraphics[width=2.1in, trim=3.1in 4.7in 3.0in 4.7in, clip]
  {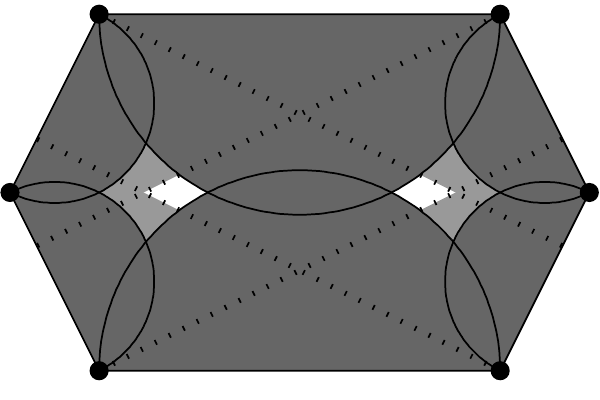}
  \caption{A cost function that accurately reflects the goal of
    well-centeredness cannot have a unique minimum, because the set of
    points that make the mesh well-centered may be a symmetric
    disconnected set.}
  \label{fig:notCnvx}
\end{figure}

If the interior vertex is placed in either of the two small white
regions that remain, the mesh will be well-centered.  We see that the
points permitted for well-centeredness form a disconnected set
in~$\Real^2$.  Moreover, the mesh is radially symmetric, so there is
no way to create an energy that prefers one white region over the
other unless we violate the desired property that the energy be
insensitive to a rotation of the entire mesh.  Any symmetric energy
that has minima in only the white regions must have at least two
distinct global minima.

In most planar triangle
meshes there is an interior vertex~$v$ that has exactly
six neighbors, all of which are interior vertices.  If all interior
vertices are free to move, as we assume in the method we propose, then
the six neighbors could be moved into the relative positions that the
boundary vertices have in the mesh in Fig.~\ref{fig:notCnvx}.  Moving
$v$ around when its neighbors have such positions should exhibit
nonconvexity in whatever cost function we might define.

\section{The Optimal Planar Triangulation}
\label{sec:globaltop}

A variety of our experimental results appears in
Section~\ref{sec:results} below.  The results support the claim that
$E_{p}$ is an appropriate cost function for quantifying the
$2$-well-centeredness of a planar mesh.  In some cases, though, the
mesh connectivity, the fixed boundary vertices, or a combination of
the two are specified in such a way that no well-centered mesh exists
with the given mesh connectivity and boundary vertices.  The simplest
example of this is a planar mesh with an interior vertex $v$ that has
fewer than five neighbors.  Since the angles around $v$ sum to $2\pi$,
$v$ has some adjacent angle of at least $\pi/2$.  The triangle
containing that angle is not $2$-well-centered.  Similarly, a boundary
vertex with a boundary angle measuring at least $\pi/2$ must have
enough interior neighbors to divide the boundary angle into pieces
strictly smaller than $\pi/2$.  We will refer to a vertex that does
not have enough neighbors as a \emph{lonely} vertex.  (In three
dimensions, a vertex must have at least 7 incident edges to permit
a 3-well-centered mesh, though having 7 neighbors is not sufficient to
guarantee that a 3-well-centered neighborhood exists.)

One way to approach problems with mesh connectivity, such as the
problem of lonely vertices, is a global mesh connectivity update,
i.e., to change the mesh connectivity over the entire mesh.  The
methods that use Voronoi diagrams \citep{DuFaGu1999} and variational
triangulations \cite{AlCoYvDe2005} both employ this approach, updating
to a Delaunay mesh each time the vertices are relocated.  In this
section we show that the optimal triangulation of a planar point set
with respect to the energy $E_{\infty}$ is a minmax triangulation,
i.e. a triangulation that minimizes the maximum angle.  Note that in
general a minmax triangulation is not a Delaunay triangulation.  (A
Delaunay triangulation is, rather, a maxmin triangulation of a planar
point set \cite{Sibson1978}).

There is an $O(n^2 \log n)$ time algorithm for computing the minmax
angle triangulation of a fixed set of points in the plane
\cite{EdTaWa1992}, so in the plane it might be feasible to
recompute the optimal triangulation at every step of our iterative
algorithm.  It is not clear, however, whether the algorithm of
\cite{EdTaWa1992} can be generalized into higher dimensions.
At the end of this section we discuss some other reasons to avoid
recomputing the optimal triangulation after each step of
energy minimization.

In the rest of this section we restrict our attention to a given set
of vertices $\meshV$ in $\Real^{2}$, fixed at their initial locations.
Given $\meshV$ we seek the mesh connectivity $\meshT$ that minimizes
$E_{\infty}(\meshV, \meshT)$. Throughout this section, where we refer
to mesh connectivity or triangulation it is assumed (often implicitly)
that we mean an admissible triangulation, i.e., a triangulation of
$\meshV$ that covers the convex hull of $\meshV$, $\conv(\meshV)$, and
has no inverted or overlapping triangles.  Many of the results would
apply when considering a different set of admissible triangulations,
but some might need small modifications, depending on the particular
set of triangulations admitted.

Since we are working in the plane, the discussion is based on planar
angles $\theta$ and the cost function defined in~\eqref{eq:Ep} in
terms of $\cos(\theta)$.  In particular we consider the cost functions
\begin{align*}
E_{cos}\left(\meshV, \meshT\right) &= 
    \max_{\theta \in \meshM} \big\{\left\lvert 2 \cos(\theta)
      - 1\right\rvert\big\}  = \lim_{p\rightarrow\infty}
      \left(\sum_{\theta \in \meshM} \left\lvert 2 \cos\left(\theta\right)
      - 1\right\rvert^p\right)^{1/p} \\
E_{min}\left(\meshV, \meshT\right) &= \min_{\theta \in \meshM}
  \left\{\theta\right\} \\
E_{max}\left(\meshV, \meshT\right) &= \max_{\theta \in \meshM}
  \left\{\theta\right\},
\end{align*}
where in the latter two cases we require $\theta \in [0, \pi]$.

We start by showing that when all triangulations of a planar point set
have a maximum angle that is at least $\pi/2$, a triangulation
minimizing $E_{max}$ is also a triangulation that minimizes $E_{cos}$.
This claim is readily proved as a corollary of the following
proposition.

\bigskip

\begin{proposition}
  \label{prop:trngfns} Let $f$ be a strictly increasing function of $\theta$
  and $g$ a nondecreasing function of $\theta$ for $\theta \in [0,
  \pi]$.  If $E_{f}(\meshT) = \max \{f(\theta_i)\}$ and $E_{g}(\meshT)
  = \max \{g(\theta_i)\}$, then $\arg\min E_{f} \subseteq \arg\min
  E_{g}$.
\end{proposition}

\begin{proof}
  For each triangulation $\meshT$, there exists some angle
  $\theta_{\meshT}$ such that $E_{f}(\meshT) = \max \{f(\theta_i)\} =
  f(\theta_{\meshT})$.  Thus for all other angles $\theta$ appearing
  in triangulation $\meshT$, we have that $f(\theta_{\meshT}) \ge
  f(\theta)$.

  Consider a specific triangulation $\meshT_0 \in \arg\min E_{f}$.  We
  have $E_{f}(\meshT_0) \le E_{f}(\meshT)$ for all triangulations
  $\meshT$.  Thus $f(\theta_{\meshT_0}) \le f(\theta_{\meshT})$
  Moreover, since $f$ is a strictly increasing function of $\theta$,
  we can conclude that $\theta_{\meshT_0} \le \theta_{\meshT}$ Then
  since $g$ is nondecreasing, we have $g(\theta_{\meshT_0}) \le
  g(\theta_{\meshT})$ for all triangulations $\meshT$.

  Now we claim that for arbitrary triangulation $\meshT$ we have
  $g(\theta_{\meshT}) \ge g(\theta)$ for all angles~$\theta$ appearing
  in triangulation~$\meshT$.  If this were not the case, then there
  would exist some angle $\hat\theta$ in $\meshT$ with $g(\hat\theta)
  > g(\theta_{\meshT})$.  Since $g$ is nondecreasing, it would follow
  that $\hat\theta > \theta_{\meshT}$, and since $f$ is strictly
  increasing, we would have $f(\hat\theta) > f(\theta_{\meshT})$.
  This, however, contradicts our definition of~$\theta_\meshT$, which
  states that $f(\theta_{\meshT}) = \max \{f(\theta_i)\} \ge
  f(\hat\theta)$.  We conclude that the claim is correct.

  It follows, then, that $g(\theta_{\meshT}) = \max \{g(\theta_i)\} =
  E_{g}(\meshT)$ for each triangulation $\meshT$.  In particular, the
  inequality $g(\theta_{\meshT_0}) \le g(\theta_{\meshT})$ implies
  that $E_{g}(\meshT_0) \le E_{g}(\meshT)$ for all triangulations
  $\meshT$.  By definition, $\meshT_0$ is a member of the set
  $\arg\min E_{g}$.
\end{proof}

\bigskip

\begin{corollary}
  If $f$ is a strictly increasing function of $\theta$ for $\theta \in
  [0, \pi]$, then $\arg\min E_{f} = \arg\min E_{max}$.
\end{corollary}

\begin{proof}
  The function $E_{max}$ is of the form $E_{g}$ where $g$ is the
  identity function on $[0, \pi]$.  Since $g$ is a strictly increasing
  function, we may apply Proposition~\ref{prop:trngfns} in both
  directions to show that $\arg\min E_{f} \subseteq \arg\min E_{max}$
  and that $\arg\min E_{max} \subseteq \arg\min E_{f}$.  We conclude
  that $\arg\min E_{max} = \arg\min E_{f}$.
\end{proof}

\bigskip

\begin{corollary}
\label{cor:minmaxisbest}
If all triangulations of a set of vertices $\meshV$ that cover
$\conv(\meshV)$ have maximum angle at least $\pi/2$, then a
triangulation minimizing $E_{max}$ also minimizes $E_{cos}$ and vice
versa.
\end{corollary}
\begin{proof}
  We can restate the corollary as follows.  If $E_{max} \ge \pi/2$ for
  all triangulations $\meshT$, then $\arg\min E_{cos} = \arg\min
  E_{max}$.  This follows because $E_{cos}$ is of the form $E_{f}$
  where $f = \left\lvert 2 \cos(\theta) - 1\right\rvert$ is a strictly
  increasing function on the interval $[\pi/2, \pi],$ and $f(\theta) <
  f(\pi/2)$ for $0 < \theta < \pi/2$.  For all practical purposes, we
  could redefine $f$ on $[0, \pi/2)$ to make $f$ a strictly increasing
  function on $[0, \pi]$.  The redefinition would have no effect
  because for all $\meshT$, the maximal $f(\theta_{i})$ occurs at some
  $\theta_{i} \ge \pi/2$.

  Some care should be taken if we allow meshes that have an angle
  $\theta = 0$, but we know that a triangle with an angle of $0$ has
  some angle measuring at least $\pi/2$, even if two of the triangle
  vertices coincide.  Since $f(\pi/2) = f(0)$, we may say that on a
  triangle with angle $0$, $f$ is maximized at the largest angle
  $\theta \ge \pi/2$.
\end{proof}

\bigskip

It should be clear that the proofs of Prop.~\ref{prop:trngfns} and
Cor.~\ref{cor:minmaxisbest} do not apply when a triangulation exists
with $E_{max} < \pi/2$.  In that case, $E_{cos}$ may be maximized at
some angle $\theta \approx 0$ rather than at the largest angle of the
mesh.  In the next theorem we establish that there is an important
relationship between $\arg\min E_{max}$ and $\arg\min E_{cos}$ even
when a well-centered triangulation exists.  (This theorem
is the acute angle case of a result from~\cite{BeEp1995},
presented here with a different proof.)

\bigskip

\begin{theorem}
\label{thm:uniquewct}
If a $2$-well-centered triangulation of a planar point set exists,
then that $2$-well-centered triangulation is unique and is both the
unique Delaunay triangulation of the point set and the unique minmax
triangulation of the point set.
\end{theorem}
\begin{proof}
  Recall that if the Delaunay complex of a planar point set has a cell
  that is not triangular, then this cell is a convex polygon with more
  than three vertices.  The vertices of the polygon are all
  cocircular, and the circumcircle is empty of other points.  In this
  case a (nonunique) Delaunay triangulation may be obtained by
  triangulating each such polygon arbitrarily.  Any such Delaunay
  triangulation must contain an angle with measure $\pi/2$ or larger.

  This can be argued from considering the possible triangulations of a
  Delaunay cell that is not triangular.
  An {\emph{ear}} of the triangulation of the Delaunay cell is a
  triangle bounded by one diagonal and two edges of the Delaunay cell.
  Since the Delaunay cell has four or more
  vertices, 
  at least two triangles will be ears in any triangulation of the
  cell.  Moreover, we can divide the circumdisk of the Delaunay cell
  into a pair of
  closed semidisks in such a way that at least one semidisk completely
  contains an ear.  
  In an ear contained in a semidisk, the angle along the boundary of
  the Delaunay cell is at least $\pi/2$.  We conclude that if the
  Delaunay complex of a planar point set is not a triangulation, then
  no completion of the Delaunay complex to a triangulation (i.e., a
  Delaunay triangulation) yields a $2$-well-centered triangulation.


Suppose, then, that a point set permits a $2$-well-centered
triangulation $\meshT_{0}$.  By Cor.~\ref{cor:Delaunay}, $\meshT_{0}$
is a Delaunay triangulation.  The Delaunay triangulation is unique in
this case (by the argument of the preceding paragraph).  Moreover, any
other triangulation $\meshT$ of the point set has a maximum angle that
is at least as large as $\pi/2$.  (If not, $\meshT$ would be
$2$-well-centered, and, therefore, a Delaunay triangulation,
contradicting the uniqueness of the Delaunay triangulation.)  We
conclude that the minmax triangulation in this case is $\meshT_{0}$
and is unique.
\end{proof}

\bigskip

\begin{figure}
  \centering
  \includegraphics[width=85pt, trim=226pt 314pt 217pt 309pt, clip]%
  {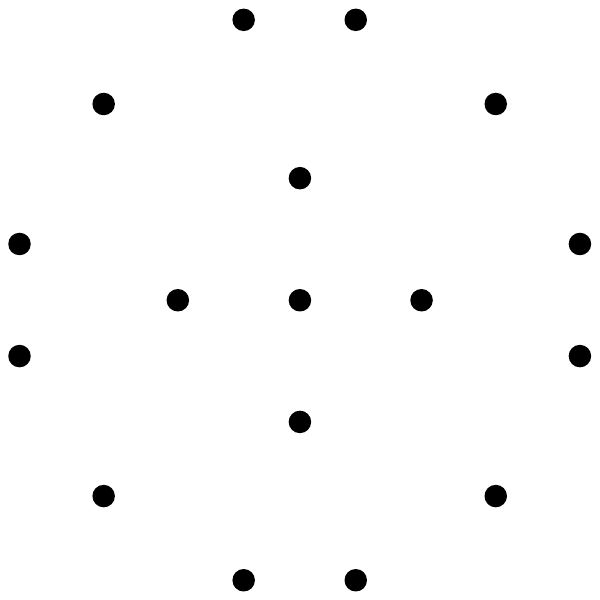}%
  \hspace{10pt}%
  \includegraphics[width=85pt, trim=229pt 317pt 220pt 312pt, clip]%
  {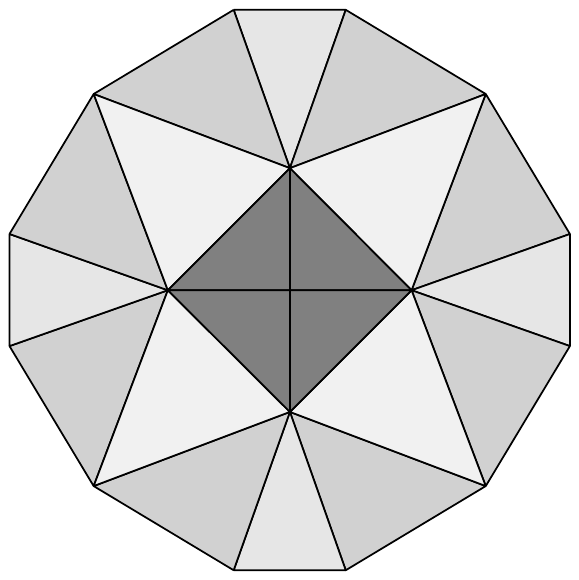}%
  \hspace{10pt}%
  \includegraphics[width=85pt, trim=229pt 317pt 220pt 312pt, clip]%
  {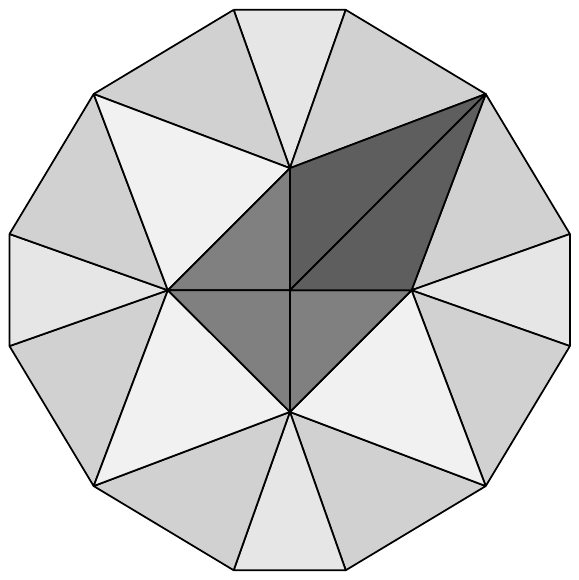}%
  \hspace{10pt}%
  \includegraphics[width=85pt, trim=229pt 317pt 220pt 312pt, clip]%
  {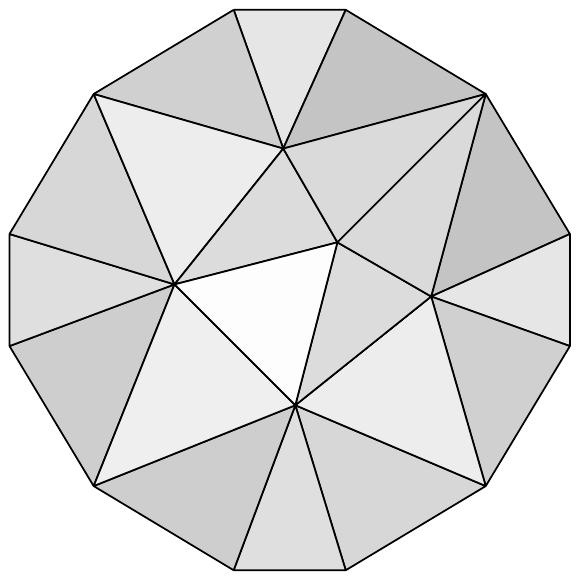}%
  \caption{The minmax triangulation may produce a triangulation in
    which interior vertices are lonely, even when there are
    triangulations with no lonely vertices.  The sequence of figures
    shows a point set, the minmax triangulation of the point set, an
    alternate triangulation of the point set with no lonely vertices,
    and a $2$-well-centered triangulation that is obtained from the
    alternate triangulation by optimizing $E_{4}$.}
  \label{fig:minmaxproblem}
\end{figure}

Combining Thm.~\ref{thm:uniquewct} with Cor.~\ref{cor:minmaxisbest} we
see that $\arg\min E_{cos} = \arg\min E_{max}$ in all cases.

Unfortunately, the minmax triangulation and the Delaunay triangulation
both have the undesirable property that they may have interior
vertices with only four neighbors, i.e., lonely vertices.
Figure~\ref{fig:minmaxproblem} shows a small point set for which the
minmax triangulation contains an interior vertex with only four
neighbors.  In this particular case, the minmax triangulation gives a
mesh for which the vertex locations optimize both $E_{\infty}$ and
$E_{4}$.  Thus optimizing $E_{\infty}$ or $E_{4}$ will not change the
mesh, even if we interleave the mesh optimization with recomputing the
optimal triangulation.

As long as we maintain the mesh connectivity given by this minmax
triangulation, we cannot make the mesh $2$-well-centered, regardless
of what function we optimize.  To address this problem we choose to
use an algorithm that preprocesses the mesh, updating the mesh
connectivity locally to eliminate lonely vertices.  The algorithm we
use for the two-dimensional case is outlined in \cite{VaHiGuRa2007}.
The preprocessing step applied to the minmax triangulation produces an
alternate triangulation of the initial vertex set.  (See
Fig.~\ref{fig:minmaxproblem}.)  For the new triangulation, optimizing
$E_{4}$ quickly finds a $2$-well-centered mesh.

A key reason that we choose to preserve the mesh connectivity
throughout the optimization process is that we want to
prevent the appearance of lonely vertices
during the optimization process.  It might
be interesting to interleave the energy optimization with a
retriangulation step that computed a triangulation that minimizes the
maximum angle among all triangulations with no lonely vertices, but we
do not know how to compute such a triangulation efficiently.  The
choice to maintain mesh connectivity during optimization also
simplifies the handling of meshes of domains with holes.

\section{Experimental Results}
\label{sec:results}

In this section we give some experimental results of applying our
energy minimization to a variety of meshes.  All of the initial meshes
shown here permit well-centered triangulations,
in many cases because the ``initial mesh'' is the output
of some preprocessing algorithm that improves the mesh
connectivity, e.g., the preprocessing algorithm described
in \cite{VaHiGuRa2007}.
The mesh optimization was
implemented using the Mesquite library developed at Sandia National
Laboratories \cite{BrDiKnLeMe2003}.  We implemented the cost function
$E_{p}$ by writing a new element-based {\texttt{QualityMetric}} with a
constructor accepting the argument $p$ and summing the energy values
on each element with the standard {\texttt{LPtoPTemplate}} objective
function (with power $1$).

We used Mesquite's implementation of the conjugate gradient method to
optimize $E_{p}$ on each mesh shown.  We did not write code for an
analytical gradient, so Mesquite numerically estimated the gradients
needed for the conjugate gradient optimization.  The optimization was
terminated with a {\texttt{TerminationCriterion}} based on the number
of iterations, so where the phrase {\emph{number of iterations}}
appears in the experimental results, it refers to the number of
iterations of the conjugate gradient method.  For the
three-dimensional meshes shown here we used the cost function $E_{p}$
for dimension $n = 3$, which is designed to find $3$-well-centered
meshes and is not sensitive to whether the facets of the tetrahedra
are acute triangles.

All of the experimental results discussed in this section were run on
a desktop machine with a dual 1.42 GHz PowerPC G4 processor and 2 GB
of memory.  
As is often the case with mesh optimization, the algorithm is quite
slow.  There are certainly opportunities for improving the efficiency
of the algorithm as well; the authors suspect that modifying the
algorithm to do optimization only in the regions where it is
necessary, instead of optimizing over the entire mesh, could improve
the efficiency significantly.

\begin{figure}
\centering
\includegraphics[width=250pt, trim=78pt 363pt 57pt 368pt, clip]
    {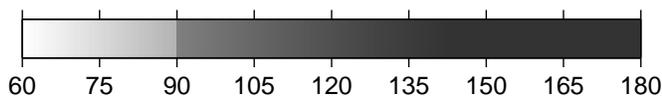}%
\caption{For two-dimensional meshes, the shade of a triangle
    indicates the measure of its largest angle.}
\label{fig:qscale}
\end{figure}

\emph{Shading scheme:} For all the two-dimensional meshes shown in
this section, we use the scale shown in Fig.~\ref{fig:qscale} to
determine the shade of each triangle.  The shade of a triangle is
determined by the measure of the largest angle of the triangle.  The
shade gets darker as the largest angle increases, with a noticeable
jump at $90$\textdegree\ so that $2$-well-centered triangles can be
distinguished from nonacute triangles.  For example, the three meshes
in Fig.~\ref{fig:minmaxproblem} use this shading scheme, and it should
be easy to identify the triangles that are not $2$-well-centered in
the first two meshes.

Along with figures of meshes, we include histograms that show
the distribution of the angles for two-dimensional meshes.  We
report near the histogram the percentage $p$ and
the number $n$ of nonacute triangles in each mesh.  The
mean of each distribution is $60$\textdegree, and the
standard deviation $\sigma$ is written near the distribution.

\subsection{Mesh of a Disk}

\begin{figure}
  \vspace*{90pt}%
  \hspace*{100pt}%
  \includegraphics[width=85pt, trim=198pt 294pt 184pt 313pt, clip]
  {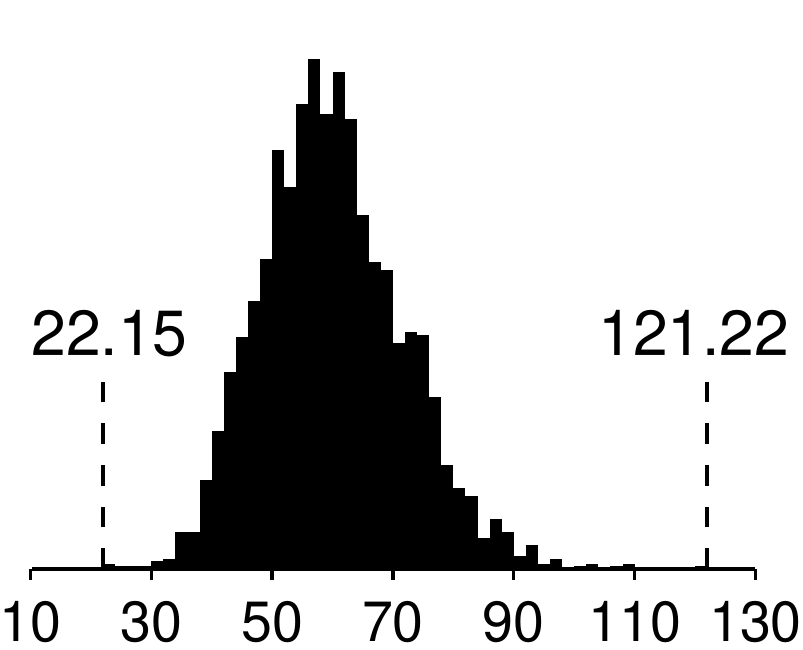}%
  \hspace{100pt}%
  \includegraphics[width=85pt, trim=198pt 294pt 184pt 313pt, clip]
  {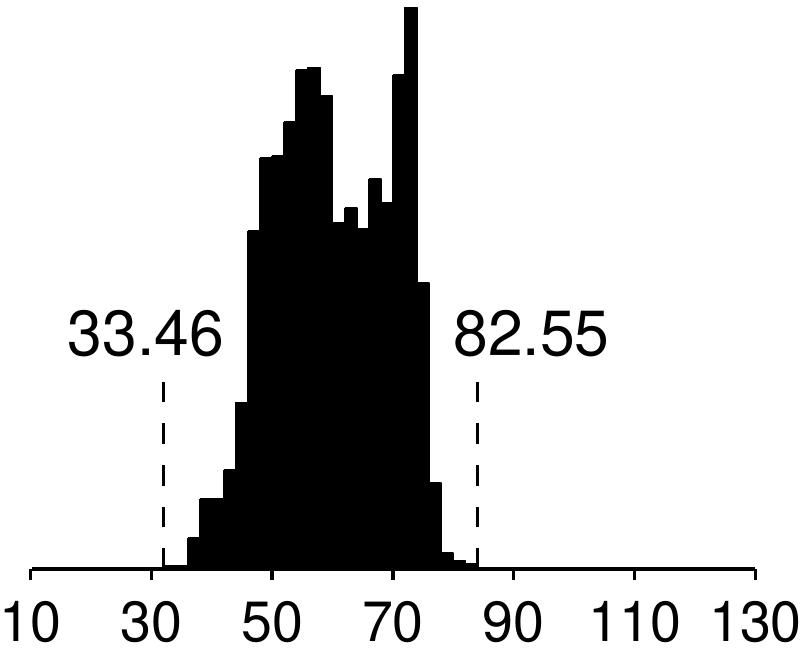}\\[-160pt]
  \includegraphics[width=125pt, trim=210pt 298pt 199pt 291pt, clip]
  {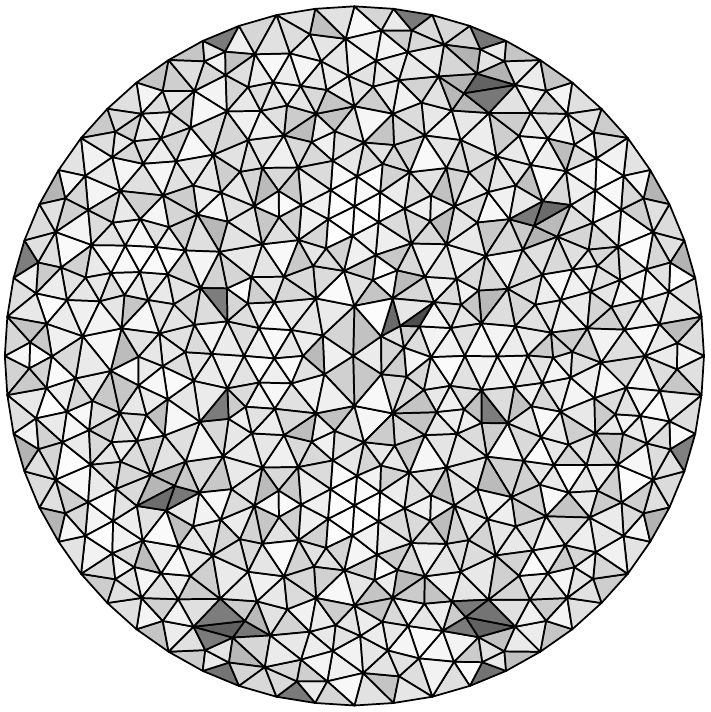}%
  \hspace{60pt}%
  \includegraphics[width=125pt, trim=210pt 298pt 199pt 291pt, clip]
  {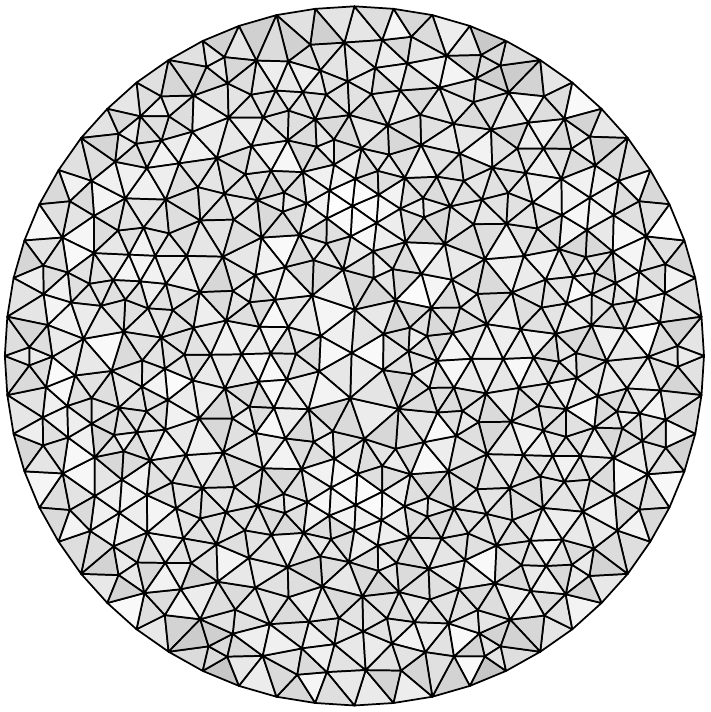}\\[5pt] 
  \begin{minipage}{125pt}
  \centering
  {\small $p = 3.10\%$\\
  $n = 27$\\
  $\sigma = 11.763$}
  \end{minipage}%
  \hspace{60pt}%
  \begin{minipage}{125pt}
  \centering
  {\small $p = 0.00\%$\\
  $n = 0$\\
  $\sigma = 9.702$}
  \end{minipage} 
  \caption{From the initial mesh shown at left, with
    $3.10\%$ of its triangles nonacute, minimizing $E_{4}$
    produces the $2$-well-centered mesh shown at right
    in $30$ iterations.  Histograms
    of the angles in the mesh are included, with the minimum and
    maximum angles marked on each histogram.  The optimization took
    $1.61$ seconds.}
  \label{fig:disk}
\end{figure}

The mesh of the disk in Fig.~\ref{fig:disk} is small enough that the
results of an experiment on the mesh can be visually inspected.  Many
of the triangles are already acute in the initial mesh, but some are
not.  Based on the shading scheme, we see visually that the result
mesh has no nonacute triangles.  The histograms of the angles in the
mesh confirm this, showing that the maximum angle was reduced from
$121.22$\textdegree\ to $82.55$\textdegree, and the minimum angle has
increased from $22.15$\textdegree\ to $33.46$\textdegree.  The
optimization took $1.61$ seconds.

\subsection{A Larger Mesh}

\begin{figure}
  \centering
  \includegraphics[width=370pt, trim = 1.4in 5.0in 0.9in 4.8in, clip]%
  {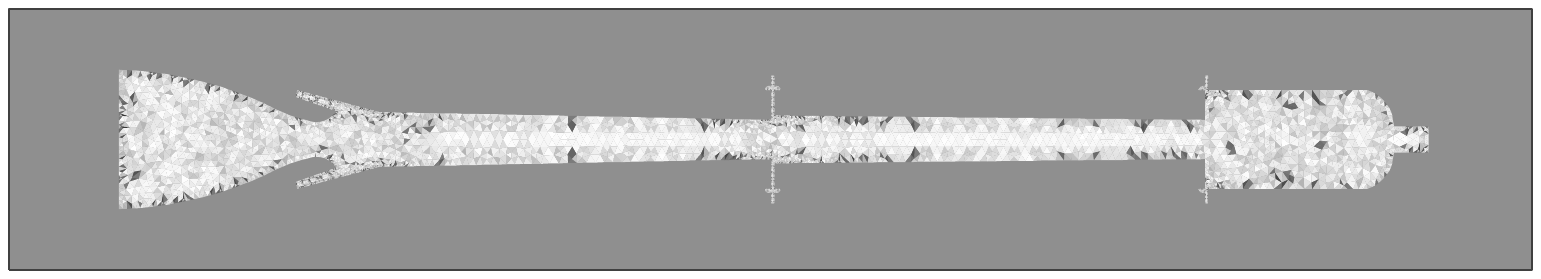}\\
  \includegraphics[width=370pt, trim = 1.4in 5.0in 0.9in 4.8in, clip]%
  {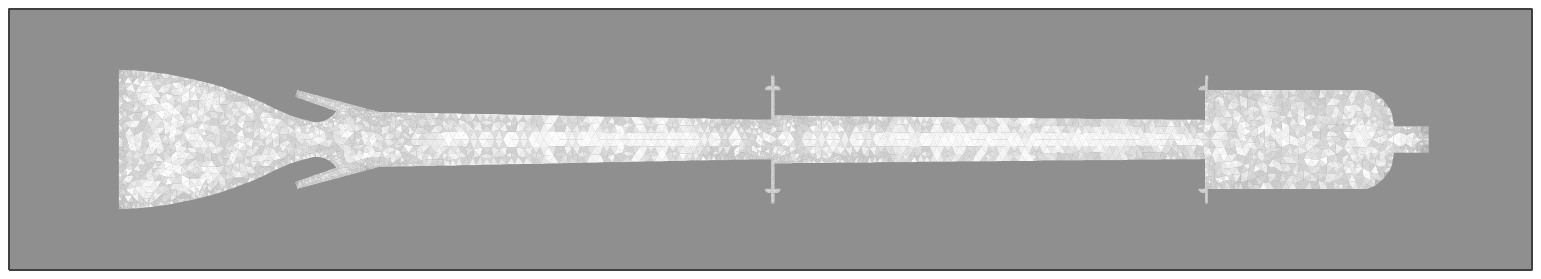}\\[7pt]
  \begin{minipage}{80pt}
  \centering
  \includegraphics[width=80pt, trim=205pt 294pt 196pt 311pt, clip]%
  {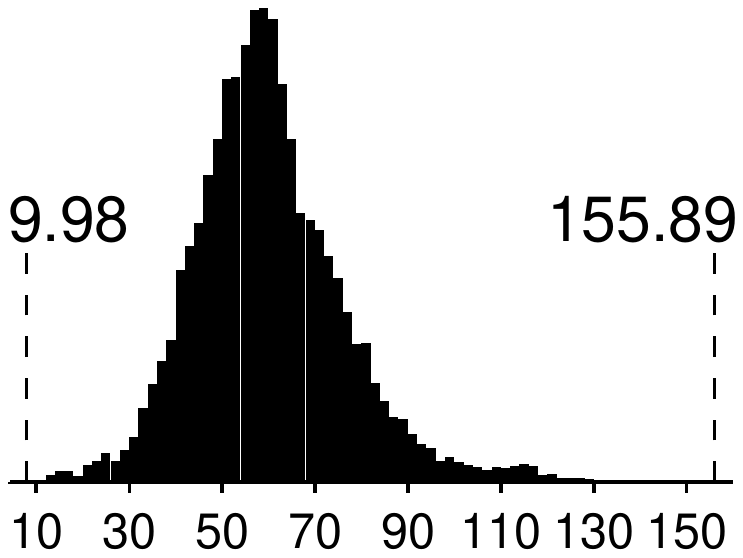}\\
  {\small $p = 13.25\%$\\
  $n = 1188$\\
  $\sigma = 16.329$}
  \end{minipage}%
  \hspace{210pt}%
  \begin{minipage}{80pt}
  \centering
  \includegraphics[width=80pt, trim=205pt 294pt 196pt 311pt, clip]%
  {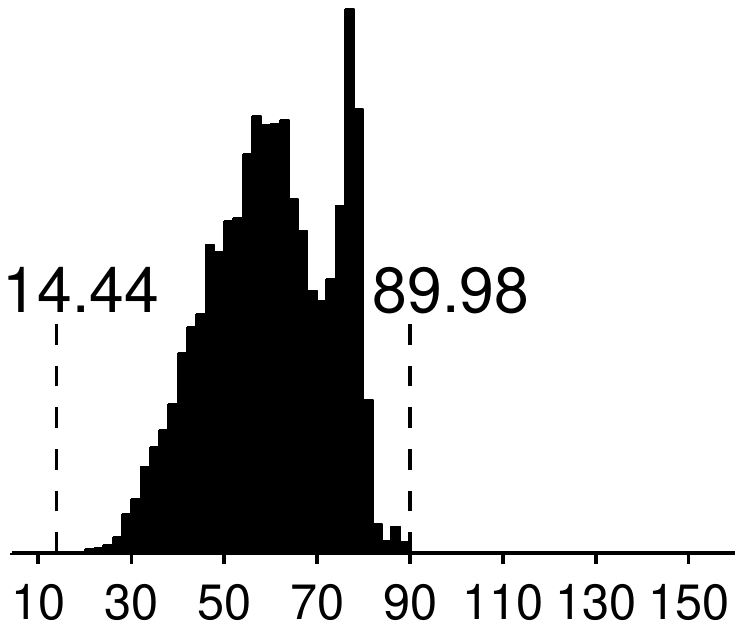}\\
  {\small $p = 0.00\%$\\
  $n = 0$\\
  $\sigma = 13.317$}
  \end{minipage}\\[-106pt]
  \includegraphics[width=2.2in, trim = 2.7in 2.9in 2.4in 2.8in, clip]%
  {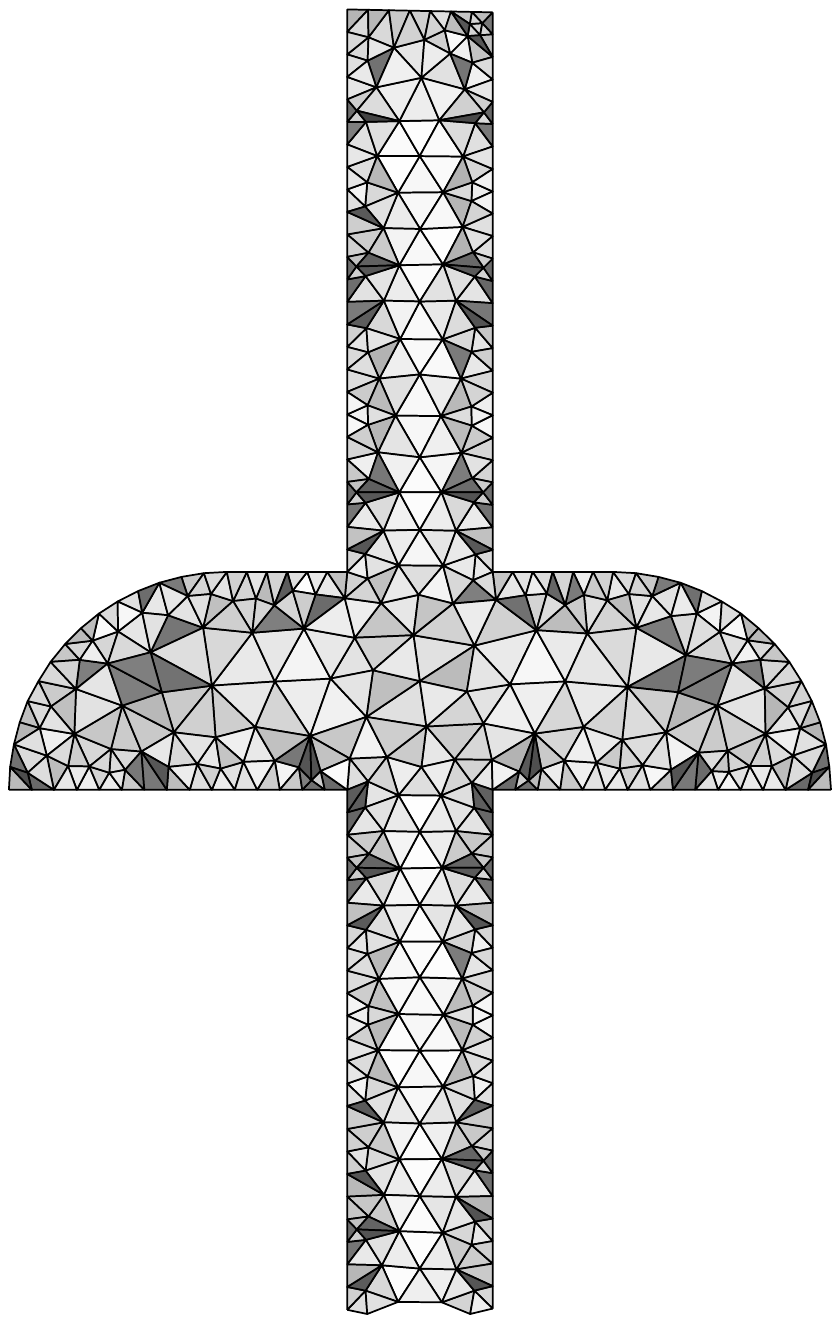}
  \hspace{7pt}%
  \includegraphics[width=2.2in, trim = 2.7in 2.9in 2.4in 2.8in, clip]%
  {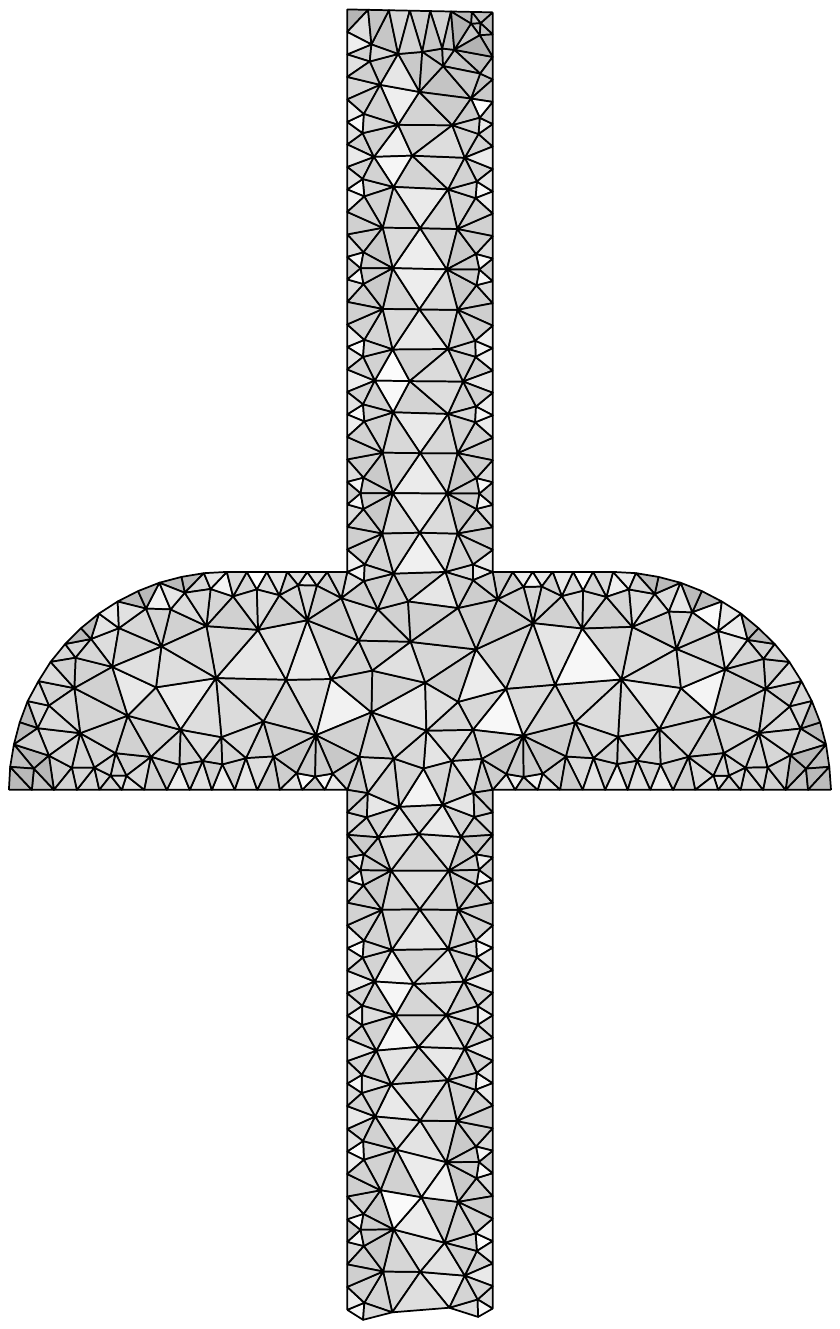}%
  \caption{Results of an experiment with a mesh of a 2-dimensional
    slice of the combustion chamber inside the Titan IV rocket.  The
    initial mesh is displayed at the top.  Below it is the result
    mesh, which was obtained by 1000 iterations minimizing $E_{10}$ on
    the mesh.  Histograms show the distribution of angles in the
    initial and final meshes.  The zoomed in views of the joint slot
    (at the top center of the full mesh) show the level of mesh
    refinement in the regions of higher detail.  For the histograms
    and the zoomed views, the original mesh is on the left, and the
    result mesh is on the right.  The optimization took $805.35$
    seconds.}
\label{fig:titan2D}
\end{figure}

In Fig.~\ref{fig:titan2D} we show results for a larger mesh, a
mesh of a two-dimensional slice of the combustion chamber inside the
Titan IV rocket.  This mesh, which is based on a mesh that the third
author produced from his work for the Center for Simulation of
Advanced Rockets, has 8966 triangles.  At the top of
Fig.~\ref{fig:titan2D} we show an overview of the entire mesh, with
the initial mesh at the very top and the result (after optimizing
$E_{10}$ for 1000 iterations) just below it.  These meshes are drawn
without showing element edges, because even the thinnest possible
edges would entirely obscure some parts of the mesh.
The background color helps define the boundary of the mesh by
providing contrast with the light gray elements.

Below the mesh overview is a zoomed view of the top center portion of
the mesh, which represents a portion of a joint slot of the titan IV
rocket.  Figure~\ref{fig:titan2D} also includes histograms of the
angle distribution of the full mesh before and after the optimization.
The angle histogram and zoomed portion for the initial mesh are shown
on the left, and for the optimized mesh are shown on the right.

In the initial mesh there are 1188 nonacute triangles
($\approx13.25\%$ of the triangles), with a maximum angle around
$155.89$\textdegree.  The result mesh has a maximum angle of
$89.98$\textdegree, and all but 143 triangles ($\approx1.59\%$) have
maximum angle below $85$\textdegree.  Of the 143 triangles that have
angles above $85$\textdegree, 14 have all three vertices on the
boundary and are thus completely specified by the boundary.  One
example of this is in the upper left corner of the zoomed view, where
there is a triangle that looks much like an isosceles right triangle.
Another 60 triangles are forced to have triangles larger than
$85$\textdegree{} because they are part of a pair of triangles along a
part of the boundary with small but nonzero curvature.
There are four such pairs along each
curved boundary in the zoomed view in Fig.~\ref{fig:titan2D}.  In
fact, all but 4 of the 143 ``worst'' triangles have at least one
boundary vertex, and the remaining 4 triangles each have a vertex that
is distance one from the boundary.

\subsection{Some More Difficult Tests} \label{subsec:twoholes}

The next mesh is a mesh of a circular domain with two circular
holes.  The initial mesh is far from being $2$-well-centered,
with $61.04\%$ of its triangles nonacute, and a standard deviation
$\sigma \approx 31.238$ for the angle distribution.  An initial
attempt to make the mesh well-centered was unsuccessful, but
two slightly different strategies, described later, do produce
a well-centered mesh.  The initial mesh and its angle histogram
are shown in Fig.~\ref{fig:twoHolesFloat_1232} (left)
along with the result of minimizing $E_{4}$ on the mesh
for $500$ iterations (right). In this case, the optimization
took $88.70$ seconds.  Comparing the optimized mesh to
the initial mesh we see that the quality has improved; the
percentage of nonacute triangles is reduced, the standard
deviation has improved, and many of the largest angles
have been reduced.

\begin{figure}
  \centering
  \begin{picture}(370, 172)
    \put(100,0){%
      \includegraphics[width=90pt, trim=201pt 297pt 193pt 311pt, clip]
      {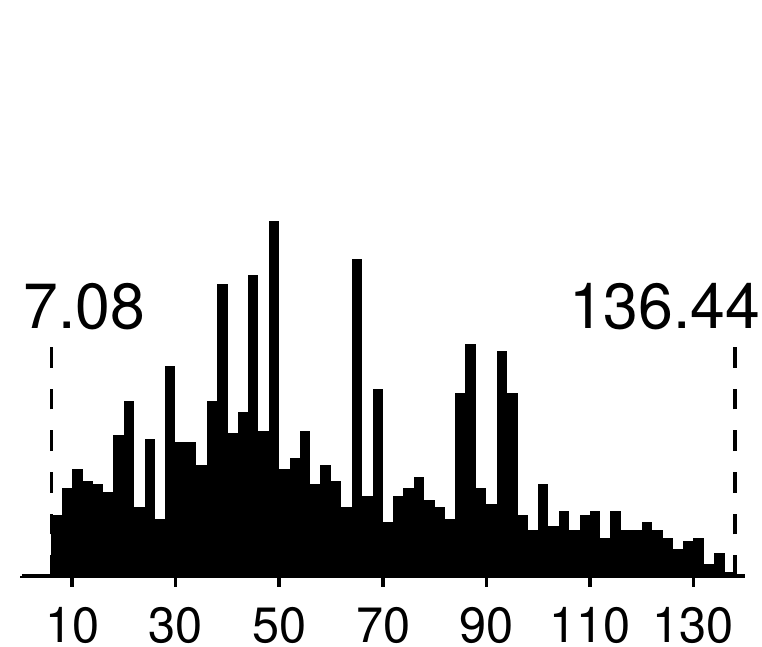}}
    \put(280,0){%
      \includegraphics[width=90pt, trim=201pt 297pt 193pt 311pt, clip]
      {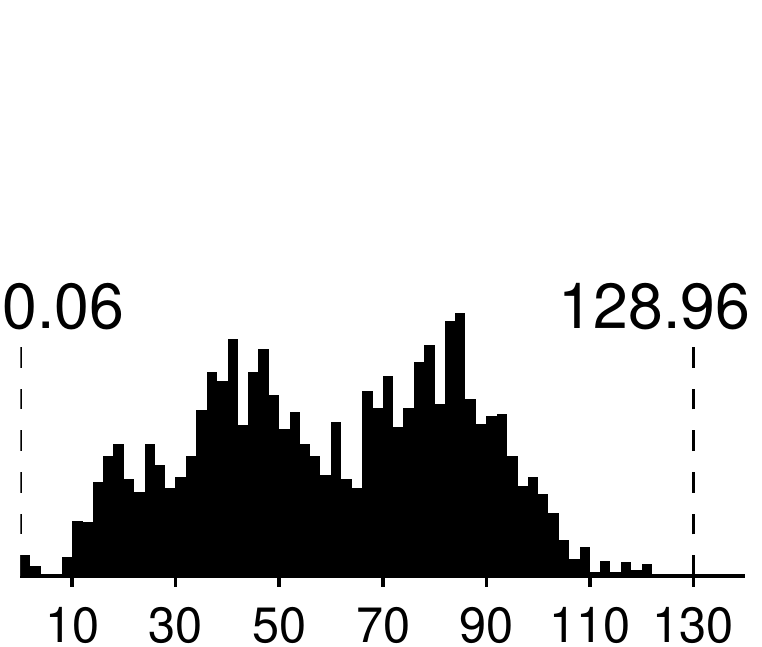}}
    \put(0, 37){%
      \includegraphics[width=135pt, trim=172pt 261pt 156pt 249pt, clip]%
      {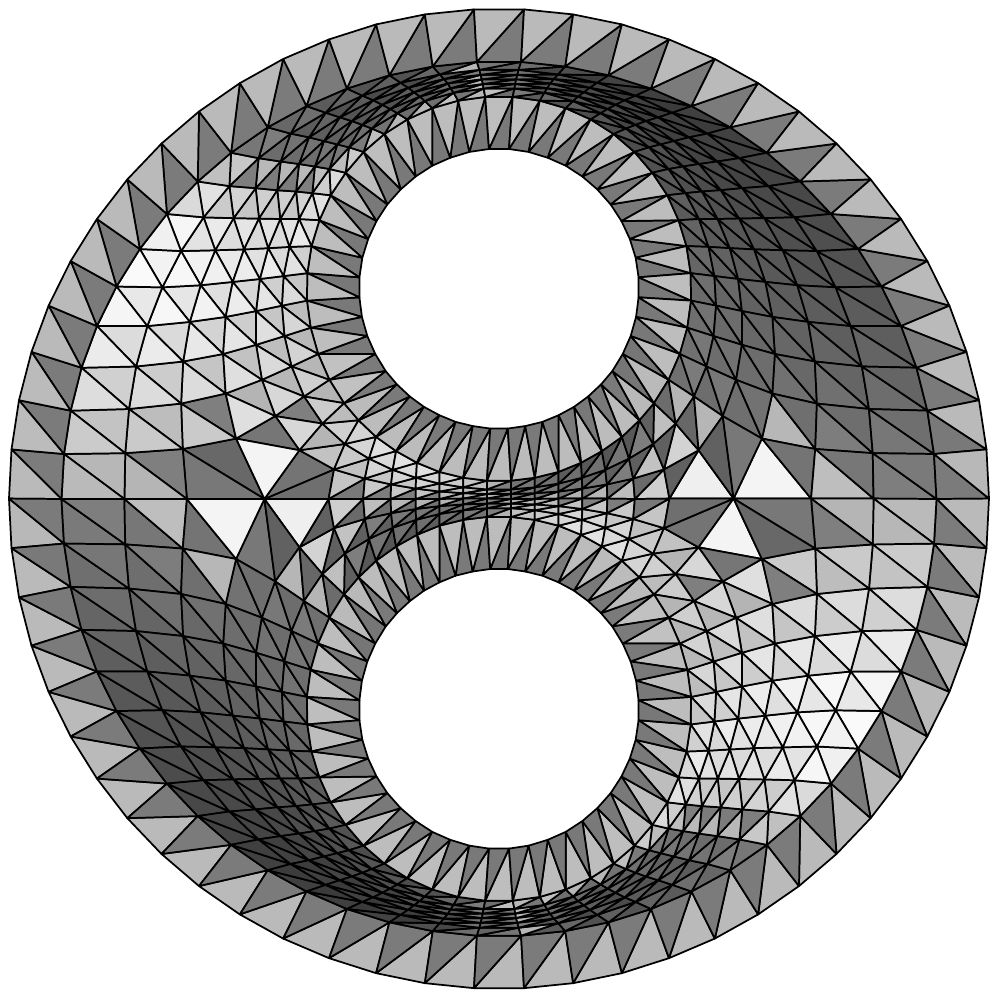}}
    \put(180, 37){%
      \includegraphics[width=135pt, trim=172pt 261pt 156pt 249pt, clip]%
      {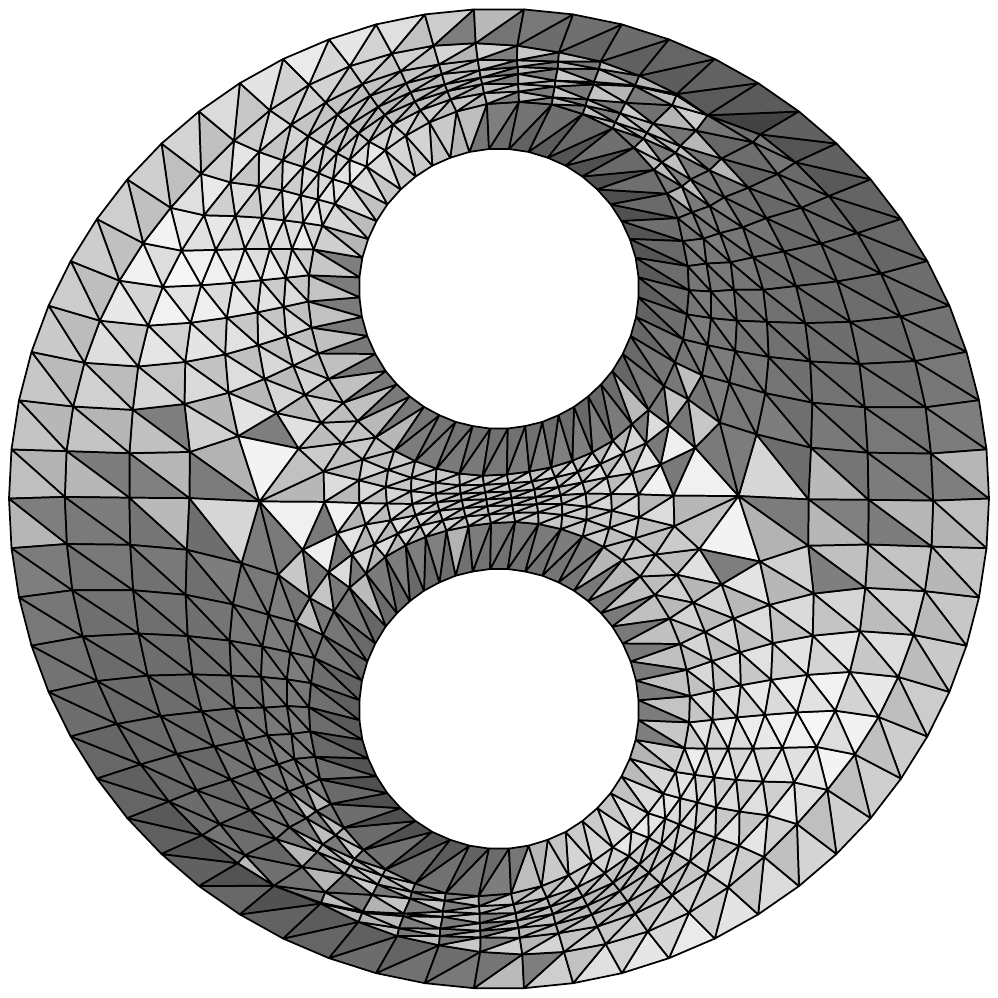}}
    \put(0, 12){%
      \begin{minipage}{135pt}
        \centering
        {\small $p = 61.04\%$\\
        $n = 752$\\
        $\sigma = 31.238$}
      \end{minipage}}
    \put(180, 12){%
      \begin{minipage}{135pt}
        \centering
        {\small $p = 38.88\%$\\
        $n = 479$\\
        $\sigma = 25.685$}
      \end{minipage}}
  \end{picture}%
  \caption{A first attempy at energy minimization applied to the two
    holes mesh on the left does not yield a well-centered mesh. Result
    after 500 iterations of $E_4$ minimization is shown on the
    right. The optimization took $88.70$ seconds.  The result mesh has
    some inverted triangles which are too thin to be seen.  In
    subsequent figures we show several strategies for producing a
    well-centered configuration.}
  \label{fig:twoHolesFloat_1232}
\end{figure}

Unfortunately, some of the
smallest angles of the initial mesh have also gotten smaller in
the optimized mesh.  In fact, four angles got so small
that their triangles became inverted in the optimized mesh.
The inverted triangles are too thin to actually see, but
there is one pair near the top right of the mesh and one pair
near the bottom left.  The energy value required to invert a
triangle is fairly large, but for large meshes or meshes
with a high percentage of bad triangles, improvements at
other locations in the mesh may be significant enough to
overcome the cost of triangle inversion for a small number
of the triangles in the mesh, and using the basic
energy $E_{p}$ can lead to inverted triangles.  Triangle
inversion can be prevented by including an inversion
barrier in the cost function.

{\bf{Energy combined with inversion barrier.}}  Modifying the energy
by introducing a term that has a barrier against inversion, i.e., a
term for which the energy value goes to infinity as a triangle moves
towards becoming degenerate, is probably the best way to handle the
problem of triangles that would become inverted with the basic $E_{p}$.
The {\texttt{IdealWeightInverseMeanRatio}} {\texttt{QualityMetric}}
provided by Mesquite is a cost function that has an implicit barrier
against inversion \cite{Munson2007}.  Let $E_{\text{imr}}$ represent
the cost function associated with the
{\texttt{IdealWeightInverseMeanRatio}}.  One can take a linear
combination of the energy $E_{p}$ with $E_{\text{imr}}$ to
create a new energy that has a barrier against inversion and,
depending on the coefficients, is still very much like $E_{p}$.
We have found that the energy
$\widetilde{E}_{p} := 100E_{p} + E_{\text{imr}}$ is often
effective in cases where the basic $E_{p}$ leads to inverted
triangles.  For this problem, for example, using $\widetilde{E}_{p}$
gives a well-centered result with no inverted triangles.  
Starting from the initial mesh and applying $500$ iterations
of~$\widetilde{E}_{4}$ followed by $500$ iterations
of~$\widetilde{E}_{6}$ and $500$ iterations of~$\widetilde{E}_{10}$
produced the $2$-well-centered mesh of the original domain displayed
in Fig.~\ref{fig:twoHolesFloat_1232_wct}.  The optimization took
$37.37 + 36.79 + 41.21 = 115.37$ seconds.

\begin{figure}
  \centering
  \begin{minipage}[c]{135pt}
    \includegraphics[width=135pt, trim=172pt 261pt 156pt 249pt, clip]%
    {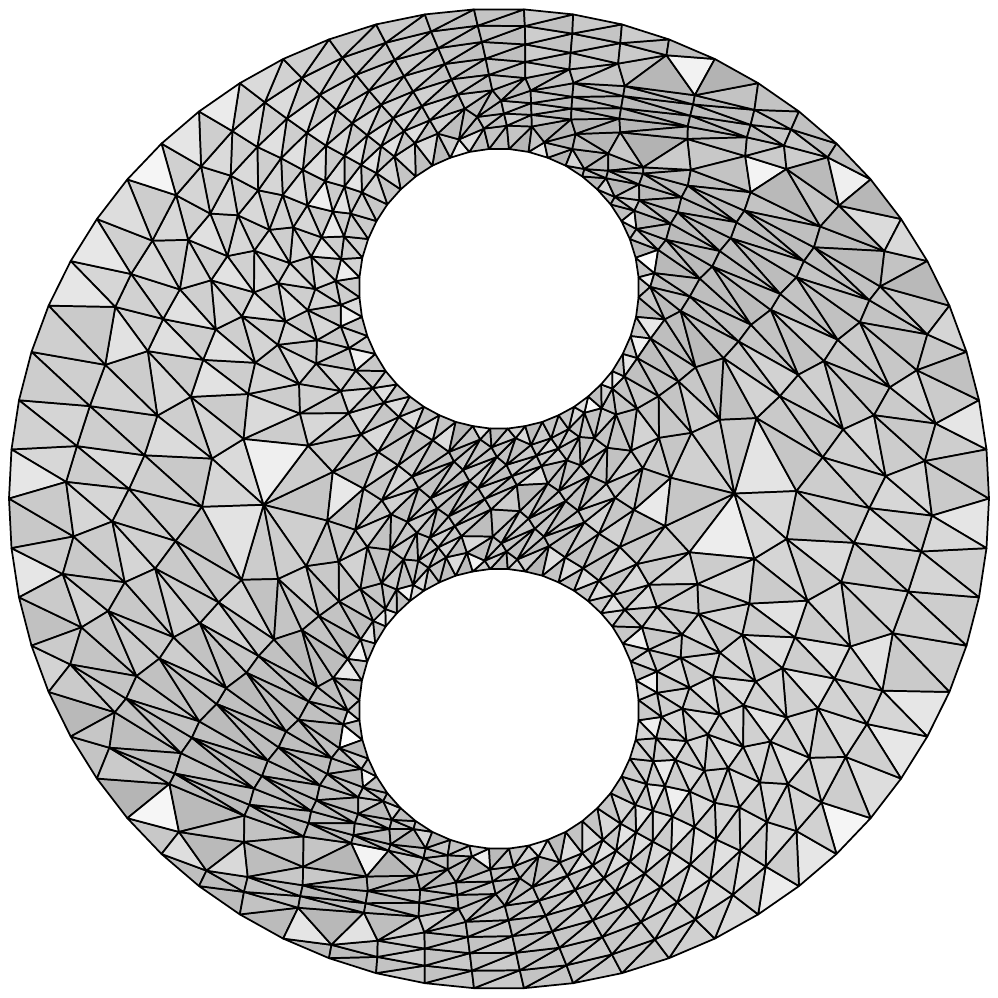}%
  \end{minipage}%
  \hspace{20pt}%
  \begin{minipage}[c]{90pt}
    \includegraphics[width=90pt, trim=201pt 297pt 193pt 311pt, clip]%
    {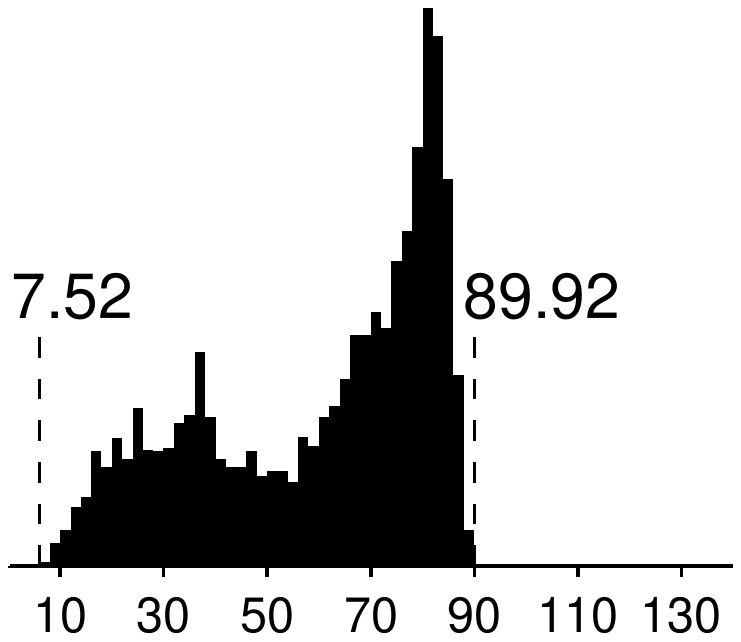}%
  \end{minipage}%
  \hspace{5pt}
  \begin{minipage}[c]{50pt}
    \centering
    {\small $p = 0.00\%$\\
    $n = 0$\\
    $\sigma = 22.450$}
  \end{minipage}
  \caption{A $2$-well-centered mesh of the two holes domain conforming
    to the mesh connectivity and boundary vertices of the original two
    holes mesh shown in Fig.~\ref{fig:twoHolesFloat_1232}.  The mesh
    was obtained using slightly modified cost
    functions~$\widetilde{E}_{p}$ that have a barrier against triangle
    inversion.  The optimization procedure was $500$ iterations
    of~$\widetilde{E}_{4}$ followed by $500$ iterations
    of~$\widetilde{E}_{6}$ followed by $500$ iterations
    of~$\widetilde{E}_{10}$.  Total optimization time was $115.37$
    seconds.}
  \label{fig:twoHolesFloat_1232_wct}
\end{figure}

{\bf{Improved boundary vertex locations.}}  Another way to get a
well-centered mesh from this initial mesh is to make the optimization
problem easier by changing the location of the boundary vertices.  The
mesh on the left in Fig.~\ref{fig:twoHolesOrth_1232} has the same mesh
connectivity as the initial two holes mesh from
Fig.~\ref{fig:twoHolesFloat_1232}, but the vertices along the boundary
have moved.  In the initial mesh the vertices along each boundary
were equally spaced, but in this case, the vertices on the
outer boundary are more dense at the north and south and less dense at
the east and west.  The vertices along the inner boundary curves
have also moved a bit.  For this mesh we use the basic energy $E_6$,
reaching a well-centered configuration by
$200$ iterations.  The result, obtained
in $18.03$ seconds, appears on the right in
Fig.~\ref{fig:twoHolesOrth_1232}.

\begin{figure}
  \centering
  \begin{picture}(370, 174)
    \put(102,0){%
      \includegraphics[width=88pt, trim=201pt 297pt 185pt 311pt, clip]
      {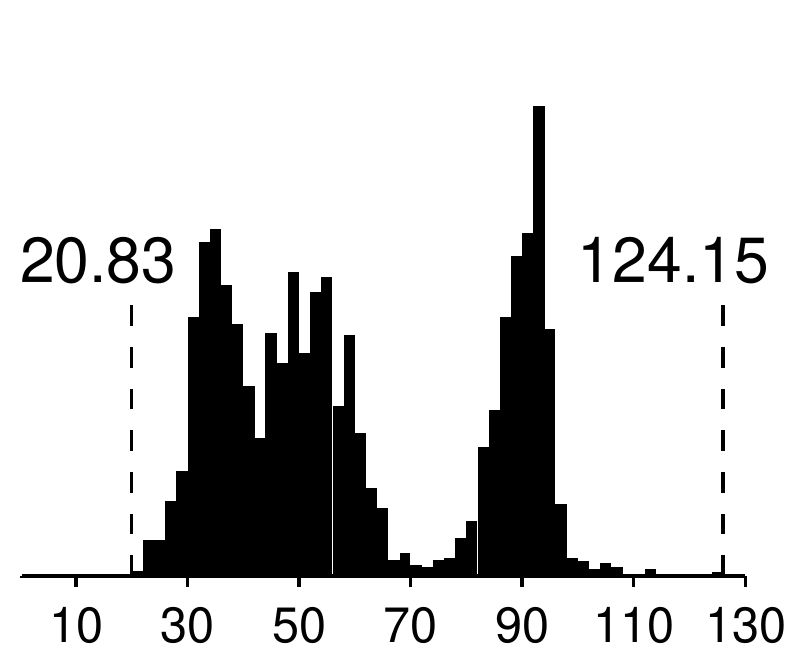}}
    \put(282,0){%
      \includegraphics[width=88pt, trim=201pt 297pt 185pt 311pt, clip]
      {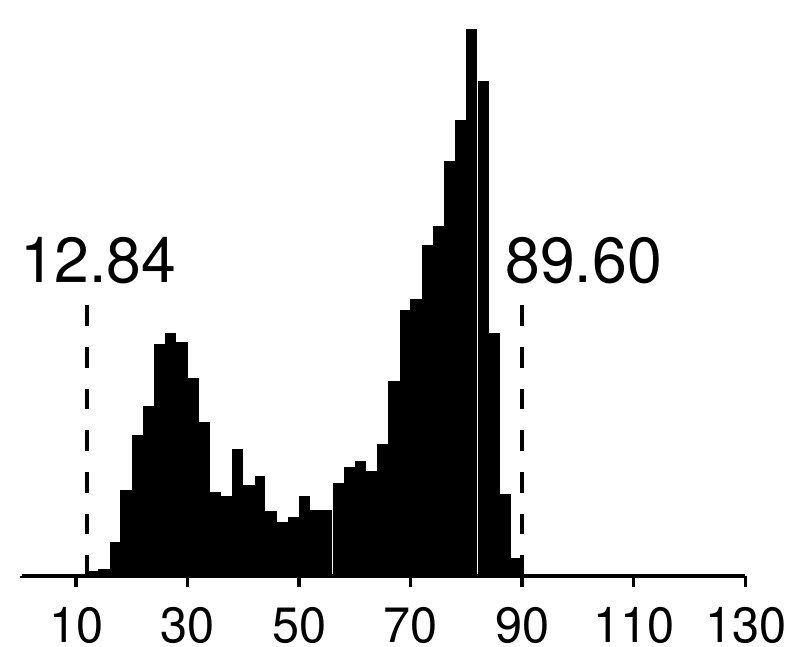}}
    \put(0, 39){%
      \includegraphics[width=135pt, trim=172pt 261pt 156pt 249pt, clip]%
      {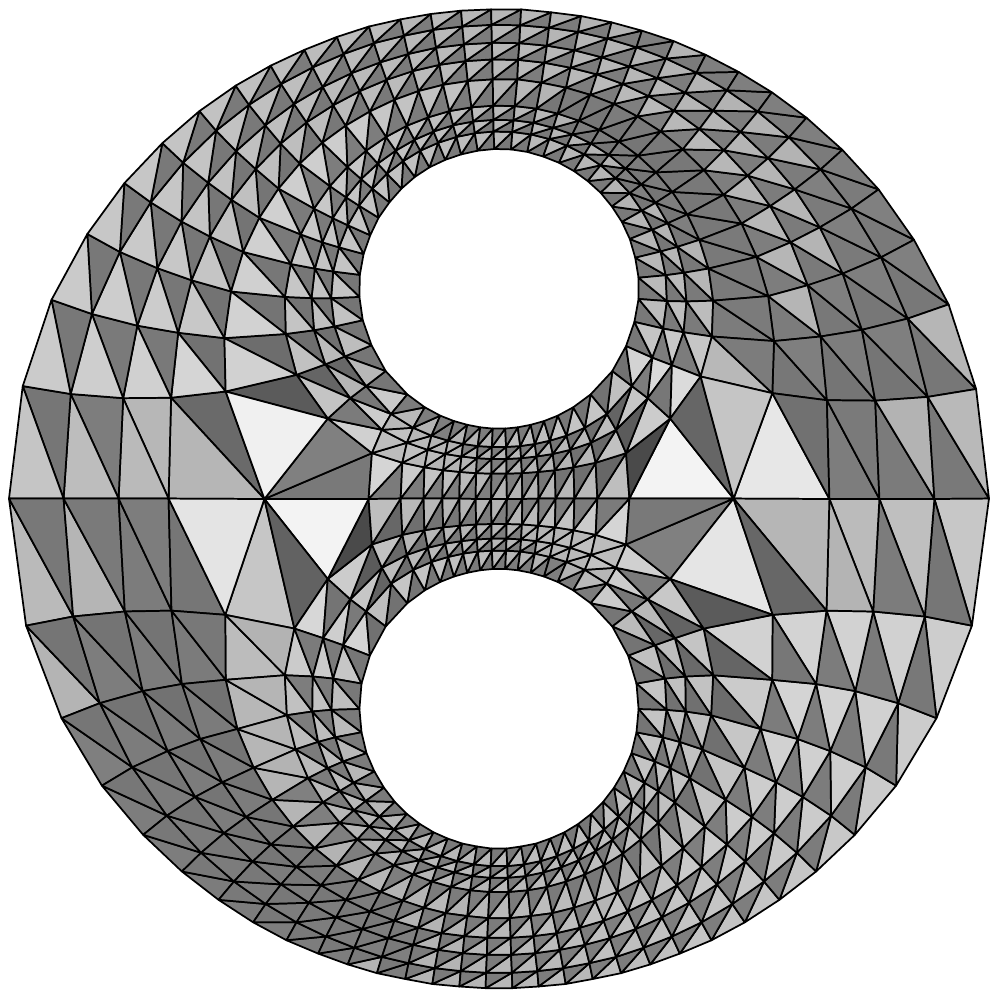}}
    \put(180, 39){%
      \includegraphics[width=135pt, trim=172pt 261pt 156pt 249pt, clip]%
      {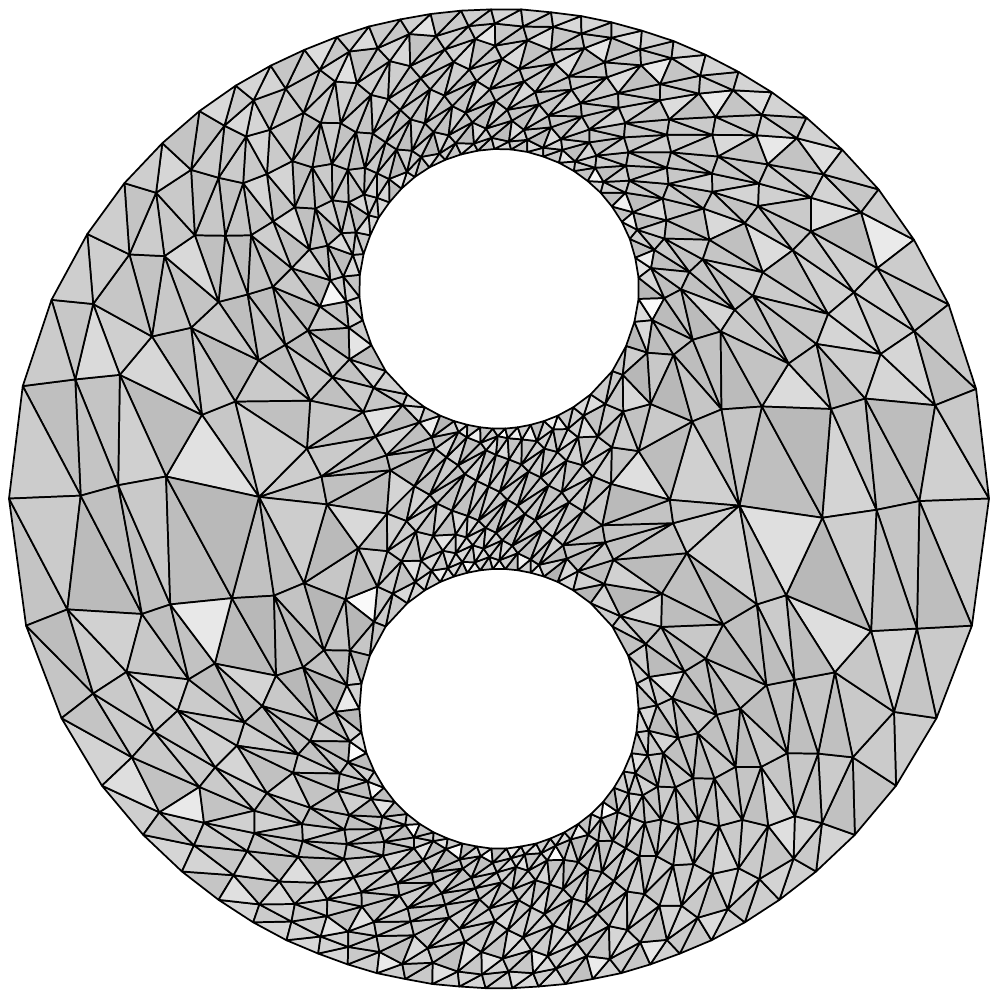}}
    \put(0, 14){%
      \begin{minipage}{135pt}
        \centering
        {\small $p = 54.63\%$\\
        $n = 673$\\
        $\sigma = 23.261$}
      \end{minipage}}
    \put(180, 14){%
      \begin{minipage}{135pt}
        \centering
        {\small $p = 0.00\%$\\
        $n = 0$\\
        $\sigma = 22.027$}
      \end{minipage}}
  \end{picture}%
  \caption{This mesh has the same mesh connectivity as the initial
    mesh in Fig.~\ref{fig:twoHolesFloat_1232}, but the vertices along
    the boundary (and in the interior) have been moved. The
    $2$-well-centered mesh on the right was obtained in $18.03$
    seconds with $200$ iterations of $E_6$ minimization.}
\label{fig:twoHolesOrth_1232}
\end{figure}

{\bf{Different mesh of the same domain.}}  The difficulty
of finding a $2$-well-centered mesh is primarily due to
the combined constraints of the mesh connectivity of the initial
mesh and the locations of the boundary vertices.  The
shape of the domain or the fact that the domain is not simply
connected are not inherently difficult for the problem of
$2$-well-centered triangulation.  When separated from the mesh
connectivity of the initial mesh, the location of the boundary
vertices are not a problem either.  We demonstrate this by
an experiment on the same domain with a completely
different mesh that has the same set of boundary vertices and
the same boundary vertex locations as the meshes
of Figs.~\ref{fig:twoHolesFloat_1232}
and~\ref{fig:twoHolesFloat_1232_wct}.  The experiment,
shown in Fig.~\ref{fig:twoholesdelref_1200}, produced a
mesh of the domain with maximum angle around
$79.50$\textdegree{} by optimizing $E_{8}$ for $100$
iterations.  The optimization took $7.44$ seconds.

\begin{figure}
  \centering
  \begin{picture}(370, 174)
    \put(102,0){%
      \includegraphics[width=88pt, trim=201pt 297pt 185pt 311pt, clip]
      {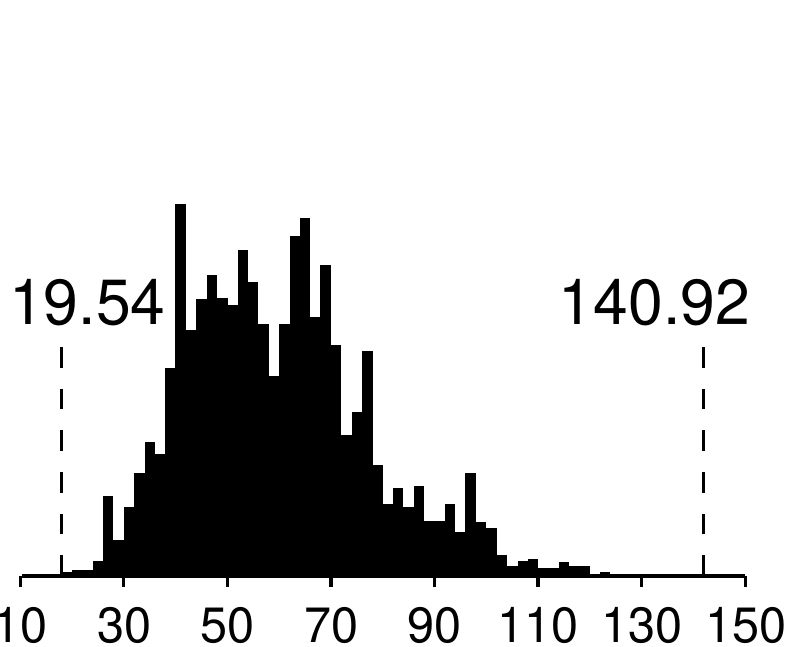}}
    \put(282,0){%
      \includegraphics[width=88pt, trim=201pt 297pt 185pt 311pt, clip]
      {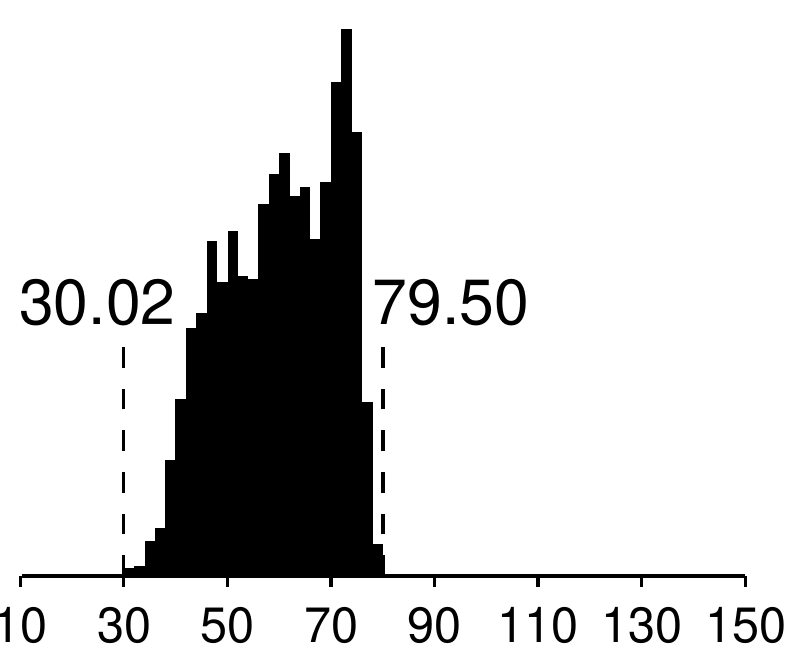}}
    \put(0, 39){%
      \includegraphics[width=135pt, trim=172pt 261pt 156pt 249pt, clip]%
      {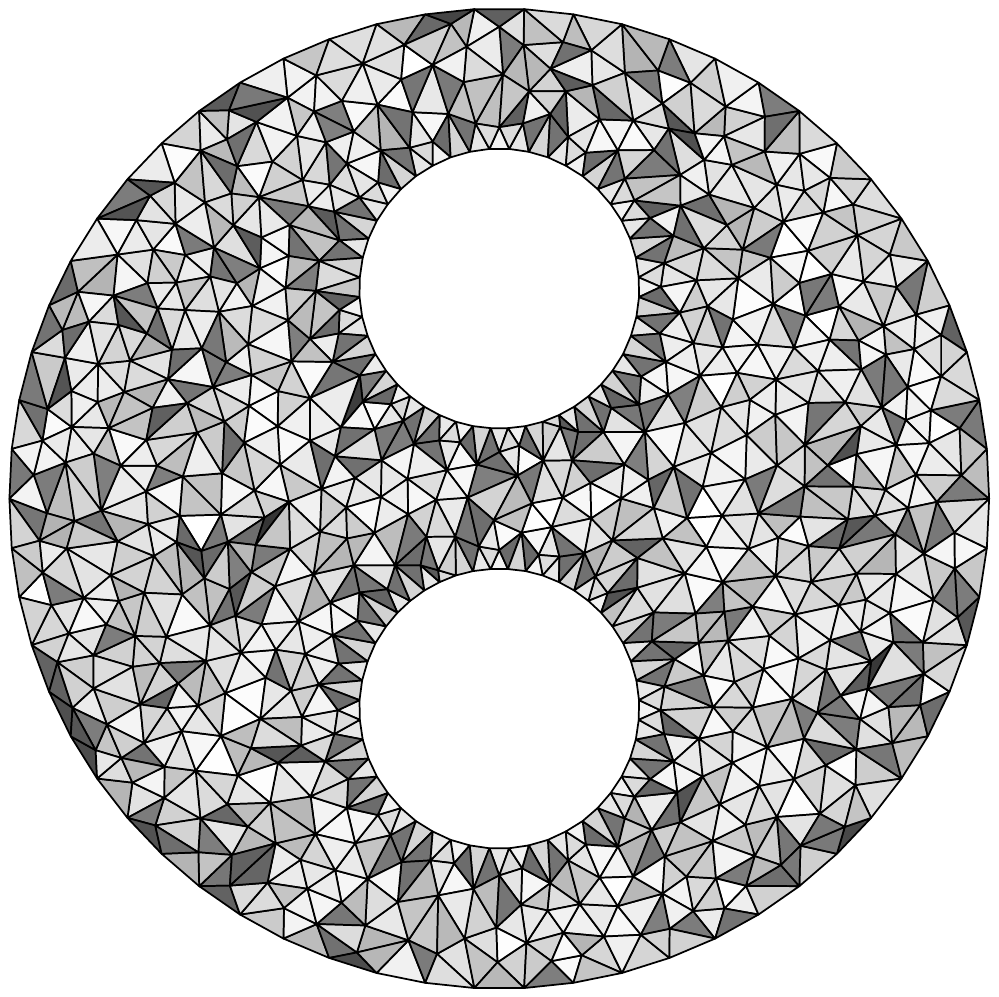}}
    \put(180, 39){%
      \includegraphics[width=135pt, trim=172pt 261pt 156pt 249pt, clip]%
      {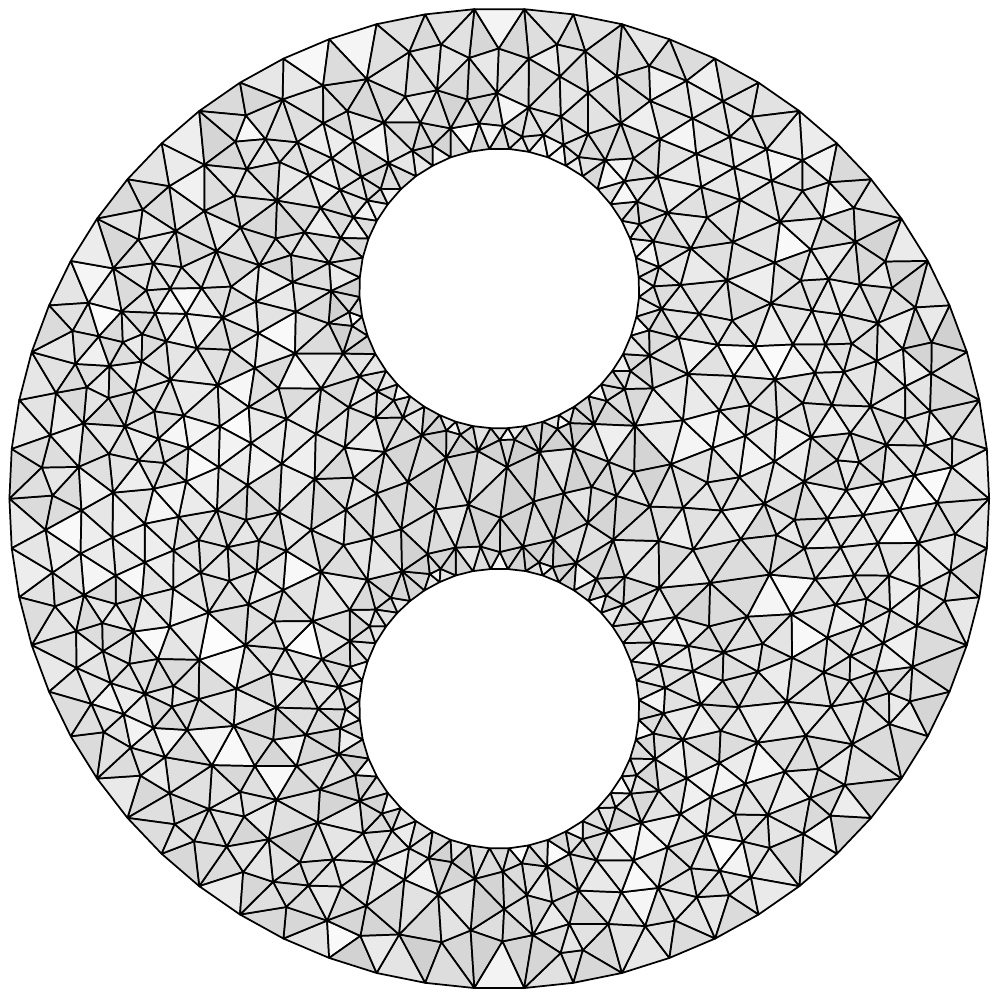}}
    \put(0, 14){%
      \begin{minipage}{135pt}
        \centering
        {\small $p = 21.92\%$\\
        $n = 263$\\
        $\sigma = 18.141$}
      \end{minipage}}
    \put(180, 14){%
      \begin{minipage}{135pt}
        \centering
        {\small $p = 0.00\%$\\
        $n = 0$\\
        $\sigma = 10.973$}
      \end{minipage}}
  \end{picture}%
  \caption{This is a mesh with the same domain and same boundary
    vertices as the mesh in Fig.~\ref{fig:twoHolesFloat_1232}.  The
    $2$-well-centered mesh on the right was obtained from the
    mesh on the left in $7.44$ seconds by minimizing $E_{8}$ for
    $100$ iterations.  The high-quality result shows that the
    difficulty of getting a $2$-well-centered mesh in
    Fig.~\ref{fig:twoHolesFloat_1232} is not due solely
    to the domain or the boundary vertices.  The initial
    mesh for this experiment was generated using the freely
    available software Triangle \cite{Shewchuk1996} and heuristics
    for improving the mesh connectivity.}
\label{fig:twoholesdelref_1200}
\end{figure}

\subsection{A Graded Mesh}
The two holes mesh of Fig.~\ref{fig:twoHolesFloat_1232} and the mesh
in Fig.~\ref{fig:titan2D} related to the titan rocket are both graded
meshes. However, the gradation of those meshes was controlled partly
by the size of elements on the boundary and by the geometry of the
mesh. In Fig.~\ref{fig:grdsquare_966} we show the results of applying
energy minimization to a mesh of the square with an artificially
induced gradation. The initial mesh and angle histogram appear at left
in Fig.~\ref{fig:grdsquare_966}.  The nearly converged result produced
by 30 iterations minimizing $E_4$ is displayed to its right. 

The initial size of the triangles of a mesh is not always preserved
well when optimizing the energy.  We expect, however, that
the energy will generally preserve the
grading of an input mesh if the initial mesh is
relatively high quality.  This hypothesis stems from the observation
that the energy is independent of triangle size, the idea that the
mesh connectivity combined with the property of $2$-well-centeredness
somehow controls the triangle size, and the supporting evidence of
this particular experiment. 

Thus optimizing graded meshes is a useful application of our
algorithm;  there are no known provably correct
algorithms for creation of graded acute-angled
triangulations of planar domains.  The recent algorithm of
\cite{ErUn2007} has produced graded acute triangulations in a
variety of experiments, but in all cases we have tried, we have
been able to improve the quality of
their triangulations (Section~\ref{subsec:lake}).  Moreover, their
algorithm is not known to generalize to higher dimensions.

\begin{figure}
  \centering
  \begin{picture}(300, 215)
    \put(60,0){%
      \includegraphics[width=80pt, trim=205pt 296pt 196pt 311pt, clip]%
      {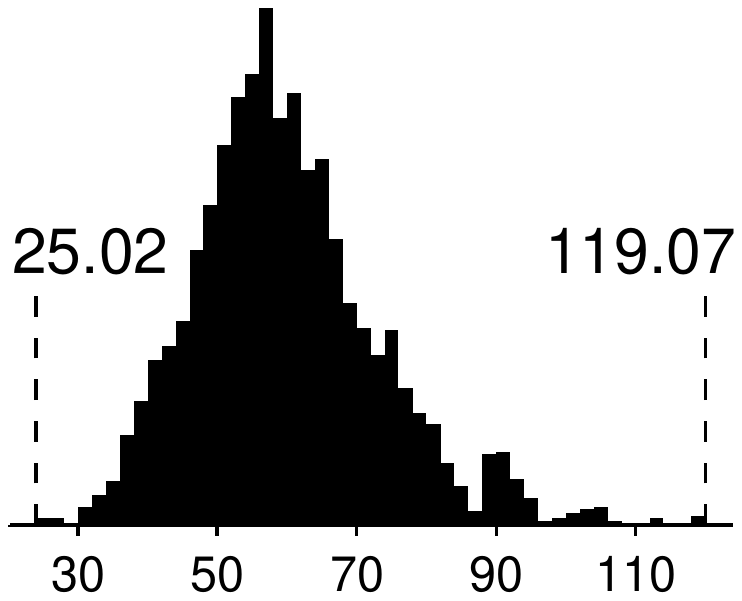}}
    \put(160,0){%
      \includegraphics[width=80pt, trim=205pt 296pt 196pt 311pt, clip]%
      {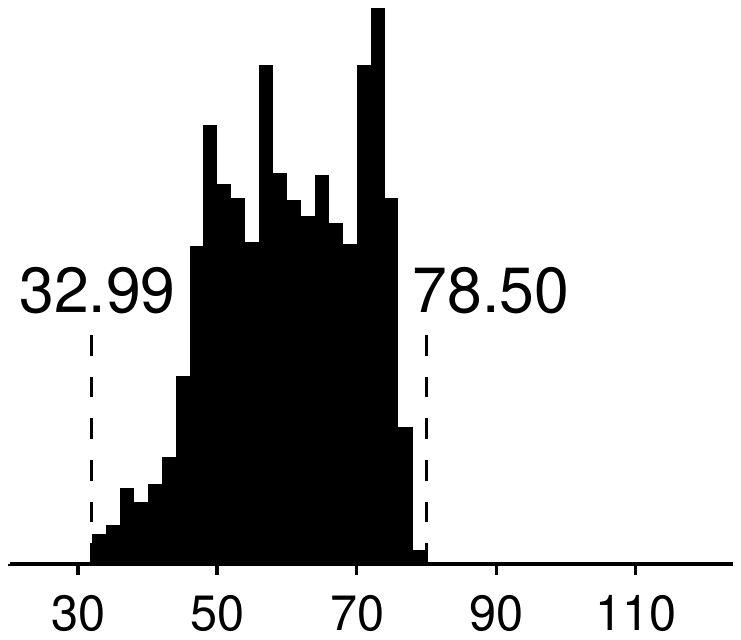}}
    \put(0, 30){%
      \begin{minipage}{50pt}
        \centering
        {\small $p = 10.66\%$\\
        $n = 103$\\
        $\sigma = 13.516$}
      \end{minipage}}
    \put(250, 30){%
      \begin{minipage}{50pt}
        \centering
        {\small $p = 0.00\%$\\
        $n = 0$\\
        $\sigma = 10.384$}
      \end{minipage}}
    \put(0,75){%
      \includegraphics[width=140pt, trim=132pt 220pt 113pt 206pt, clip]%
      {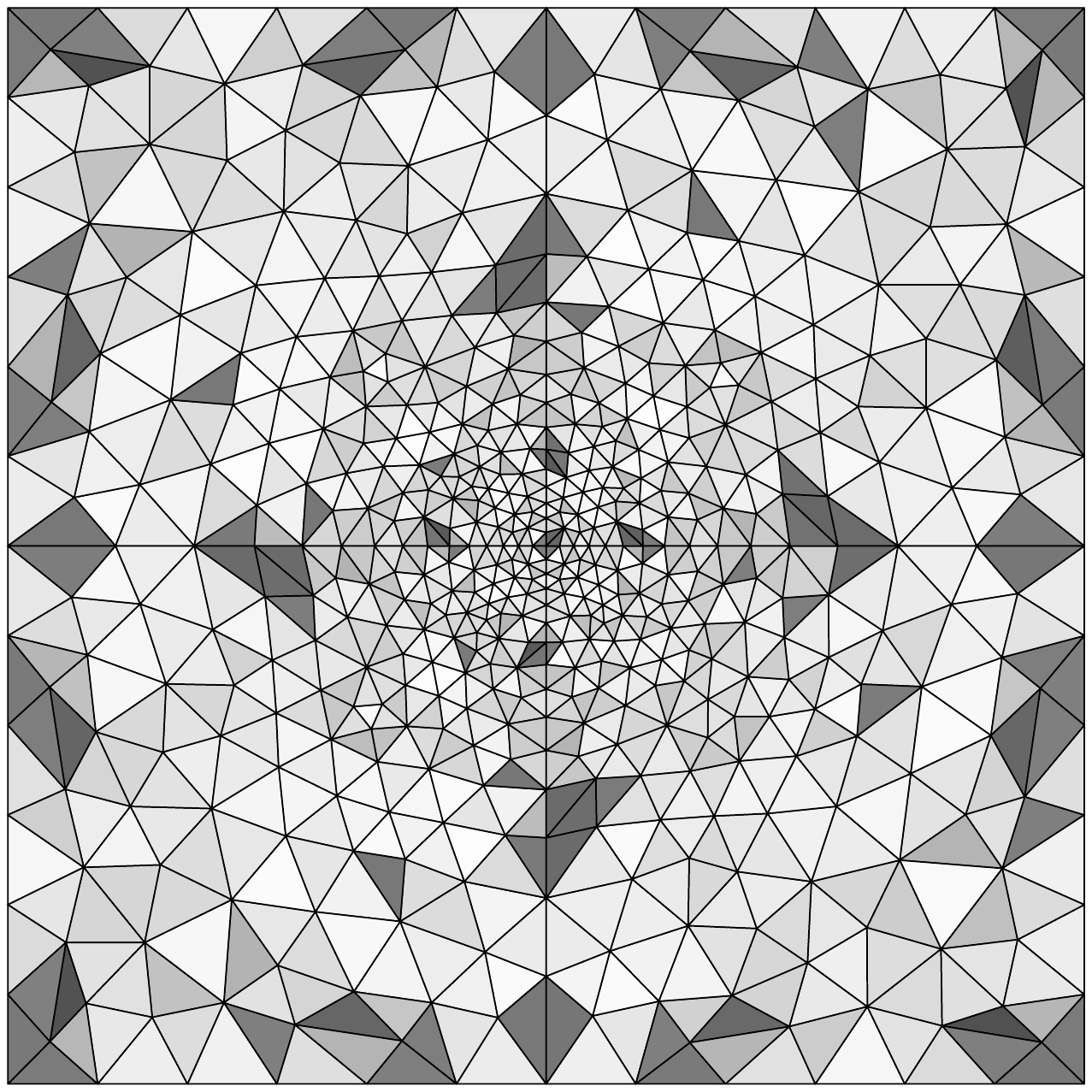}}
    \put(160,75){%
      \includegraphics[width=140pt, trim=132pt 220pt 113pt 206pt, clip]%
      {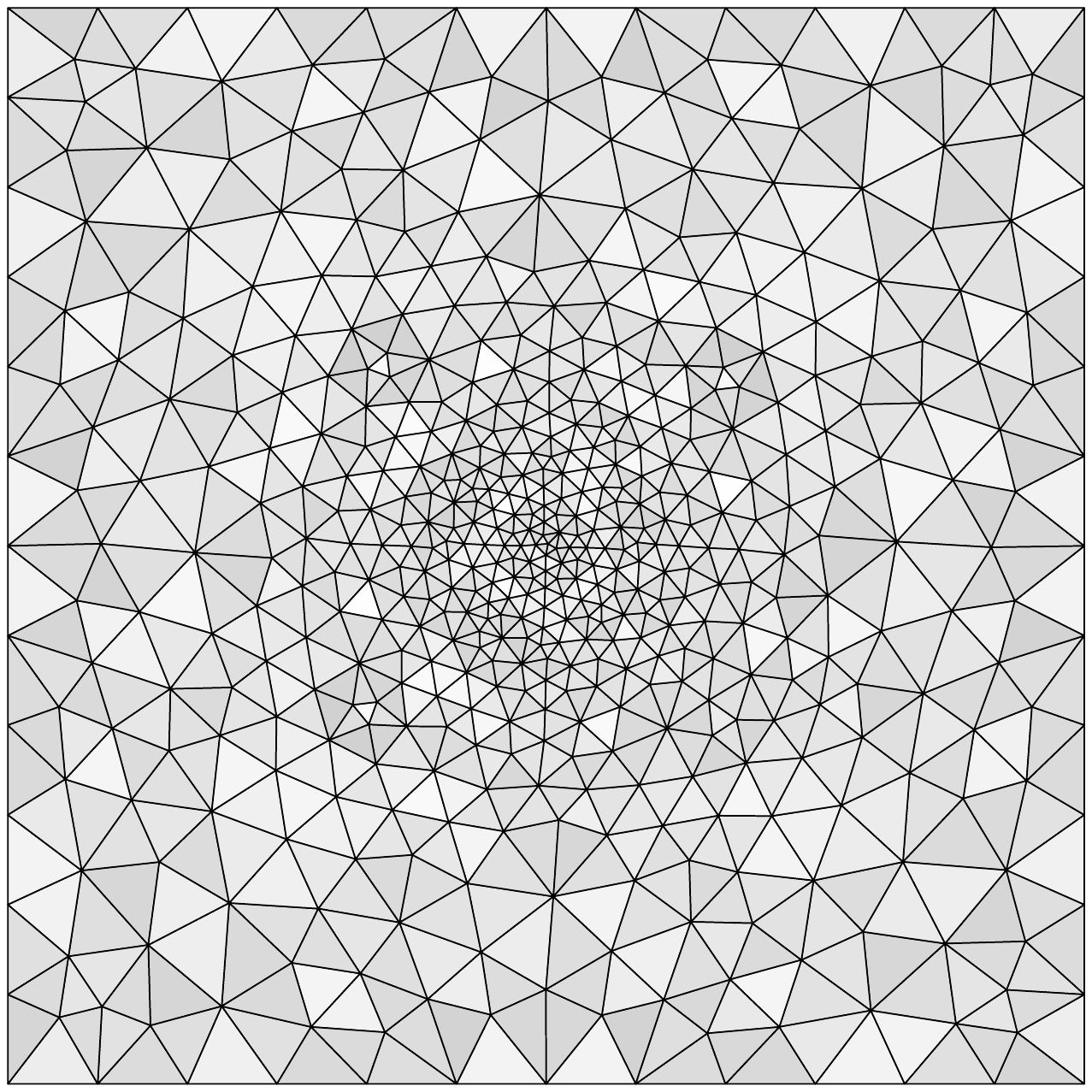}}
  \end{picture}
  \caption{For this graded mesh of the square, minimizing $E_{4}$ on
    the initial mesh (left) produces a $2$-well-centered mesh (right)
    that has grading similar to the initial mesh.  The optimization
    ran for $30$ iterations, completing in $2.16$ seconds.}
  \label{fig:grdsquare_966}
\end{figure}

\subsection{Mesh of Lake Superior} \label{subsec:lake}
The Lake Superior domain, with its complicated shape, has appeared
in many papers about quality meshing.  We include an example
optimizing a mesh of this well-known domain.  The initial mesh
is already $2$-well-centered in this experiment, but we show that
we can improve its quality with our optimization algorithm.  The
results are represented graphically in Fig.~\ref{fig:superior_1388}.

The initial acute-angled mesh is from the work of Erten and \"Ung\"or
\cite{ErUn2007} on generating acute $2$-D triangulations with a
variant of Delaunay refinement.  The initial mesh has a maximum angle
of $89.00$\textdegree\ with $174$ triangles having angles larger than
$88.00$\textdegree.  Directly optimizing $E_{10}$ on the initial mesh,
Mesquite finds a local minimum of $E_{10}$ after $6.63$ seconds ($21$
iterations).  The local minimum has exactly one nonacute triangle
(maximum angle $91.03$\textdegree) and only 40 triangles having angles
larger than $88.00$\textdegree.  The angle histogram for this result
is included in Fig.~\ref{fig:superior_1388} at top center.  The mesh
is visually very similar to the initial mesh and does not appear in
this paper.

If we start by optimizing $E_{4}$ and follow that by optimizing
$E_{10}$ we obtain a local (perhaps also global) minimum of $E_{10}$
with with much lower energy than the result obtained by directly
optimizing $E_{10}$.  The result of this optimization process is
shown on the right in Fig.~\ref{fig:superior_1388}.  The optimization
took $131.48$ seconds total; Mesquite spent $102.81$ seconds ($453$
iterations) finding a minimum of $E_{4}$ and $28.67$ seconds ($125$
iterations) finding a minimum of $E_{10}$.

Laplacian smoothing is a popular mesh optimization technique
that was first used for structured meshes with quadrilateral
elements and later generalized to triangle meshes~\cite{Winslow1964}.
A brief description of Laplacian smoothing is given
in~\cite{Field1988}.  We compare our mesh optimization
technique with Laplacian smoothing,
using the implementation of Laplacian smoothing provided by
the Mesquite library.  The result of Laplacian smoothing
on the Lake Superior mesh is shown in
Fig.~\ref{fig:superior_1388lplc}.  The optimization was
terminated after $100$ iterations, which is near convergence.
The run time was $1.31$ seconds.  
The maximum angle in the result is $109.27$\textdegree{}
and more than $4\%$ of the triangles are nonacute.

The result of optimizing the Lake Superior mesh with Laplacian
smoothing is typical of the results obtained
with Laplacian smoothing.  We performed
experiments with Laplacian smoothing on all of the
2-D meshes presented in this paper, and no mesh became
well-centered except for the mesh of the square
in Fig.~\ref{fig:grdsquare_966}, where Laplacian smoothing
produced a mesh with maximum angle $87.54$\textdegree{}
compared to the maximum angle of $78.50$\textdegree{} obtained
by our method.  In most cases the percentage of nonacute
triangles after Laplacian smoothing was between $1\%$ and
$5\%$, but for the meshes in
Figs.~\ref{fig:twoHolesFloat_1232},~\ref{fig:twoHolesFloat_1232_wct},
and~\ref{fig:twoHolesOrth_1232}, the percentage of nonacute
triangles was much higher, getting as high as
$48.70\%$ for the mesh in Fig.~\ref{fig:twoHolesOrth_1232}.
Clearly the traditional Laplacian smoothing is not an
appropriate tool for finding acute triangulations.

\begin{figure}
  \centering
  \begin{picture}(370,147)
    \put(0,88){%
      \includegraphics[width=70pt, trim=206pt 296pt 187pt 311pt, clip]%
      {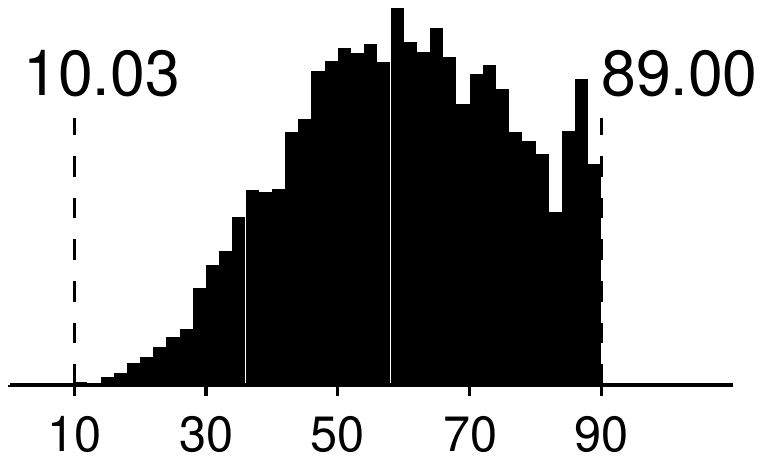}}
    \put(0,67){%
      \begin{minipage}{60pt}
        \centering
        {\small $p = 0.00\%$\\
        $n = 0$\\
        $\sigma = 16.729$}
      \end{minipage}}
    \put(152,88){%
      \includegraphics[width=70pt, trim=206pt 296pt 187pt 311pt, clip]%
      {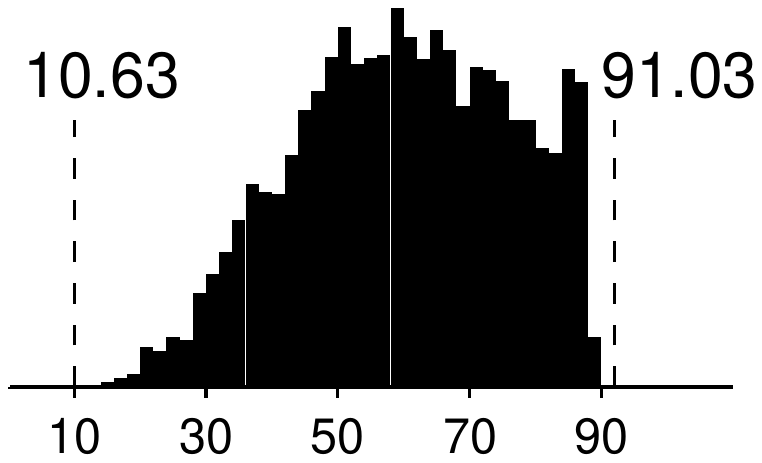}}
    \put(152,67){%
      \begin{minipage}{60pt}
        \centering
        {\small $p = 0.05\%$\\
        $n = 1$\\
        $\sigma = 16.440$}
      \end{minipage}}
    \put(300,88){%
      \includegraphics[width=70pt, trim=206pt 296pt 187pt 311pt, clip]%
      {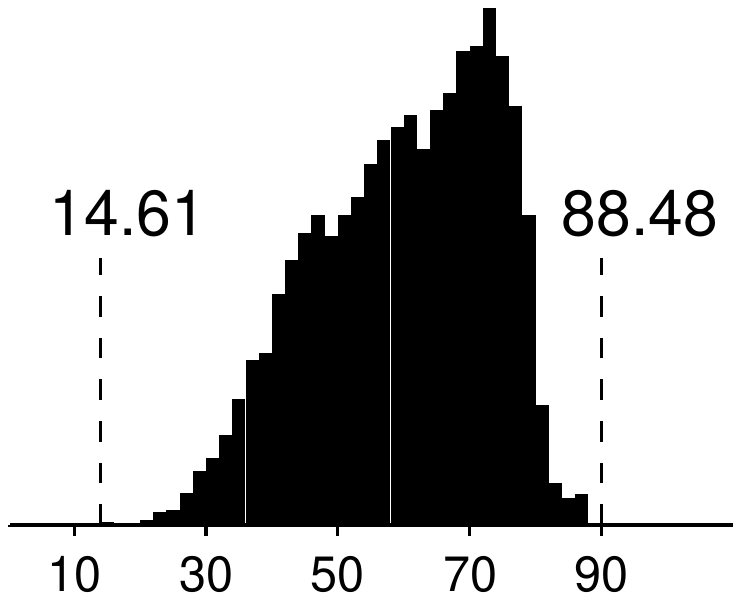}}
    \put(300,67){%
      \begin{minipage}{60pt}
        \centering
        {\small $p = 0.00\%$\\
        $n = 0$\\
        $\sigma = 13.504$}
      \end{minipage}}
    \put(0,0){%
      \includegraphics[width=170pt, trim=36pt 267pt 10pt 249pt, clip]%
      {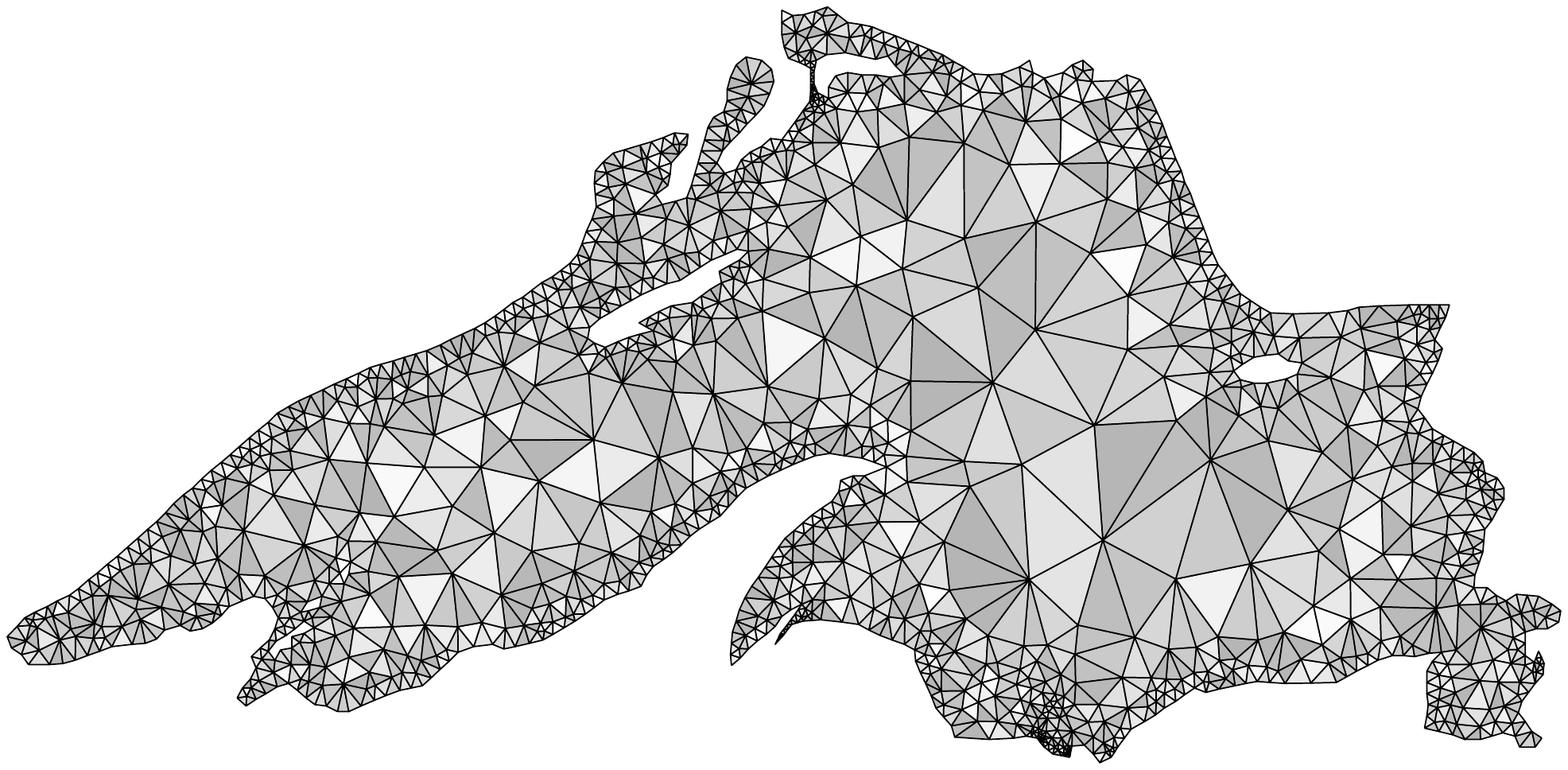}}
    \put(173,0){%
      \includegraphics[width=170pt, trim=36pt 267pt 10pt 249pt, clip]%
      {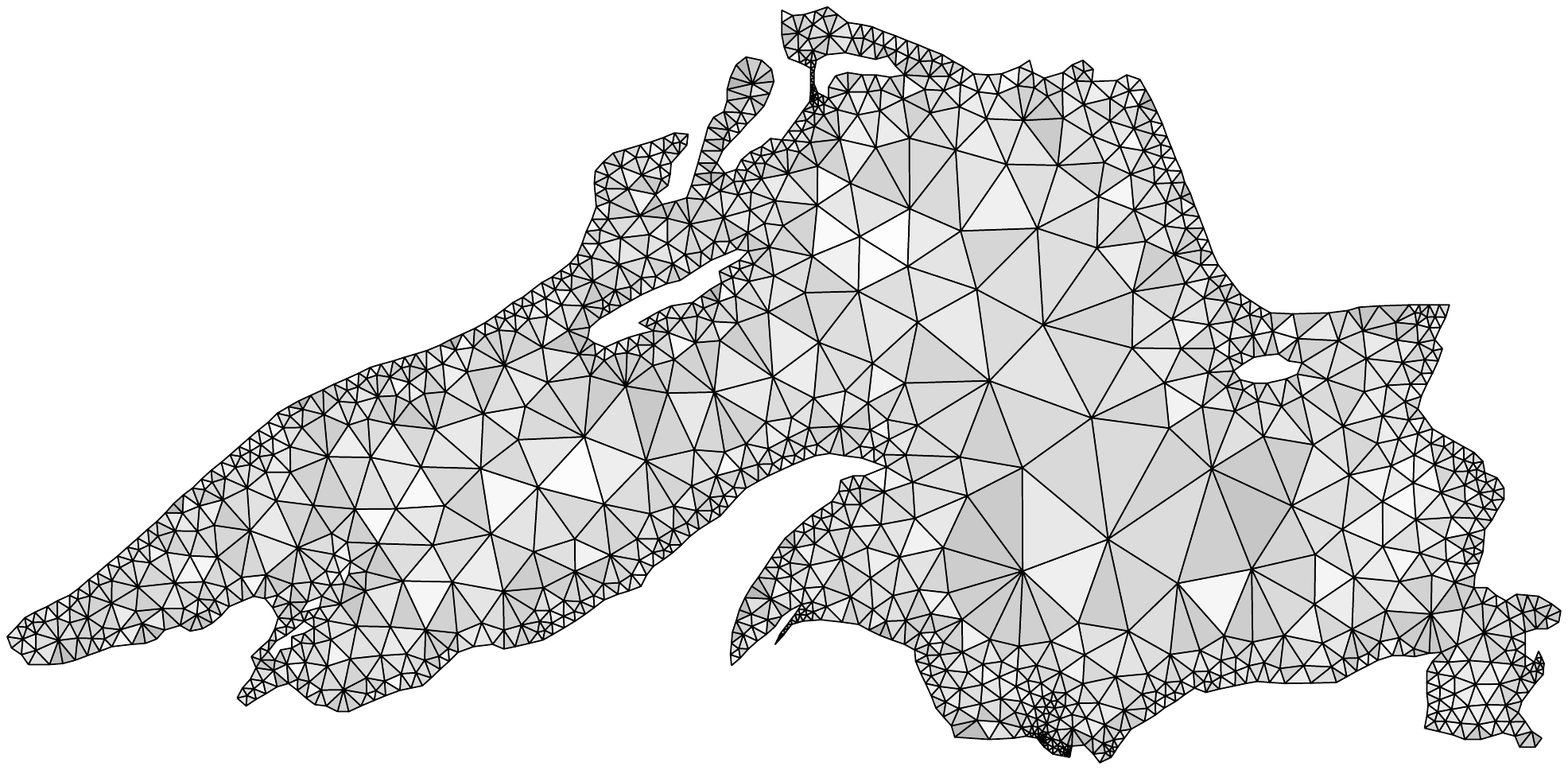}}
  \end{picture}
  \caption{Result for a mesh of Lake Superior.  The initial mesh shown
    on the left is a $2$-well-centered mesh from \cite{ErUn2007}.  The
    improved mesh shown on the right was obtained by first optimizing
    $E_{4}$ and then optimizing $E_{10}$.  The angle histogram at top
    center shows the result of optimizing $E_{10}$ directly on the
    initial mesh.  Many of the angles that were near $90$\textdegree\
    have dropped to below $80$\textdegree.}
  \label{fig:superior_1388}  
\end{figure}

\begin{figure}
  \centering
  \begin{picture}(250,92) 
    \put(180,37){%
      \includegraphics[width=70pt, trim=206pt 296pt 187pt 311pt, clip]%
      {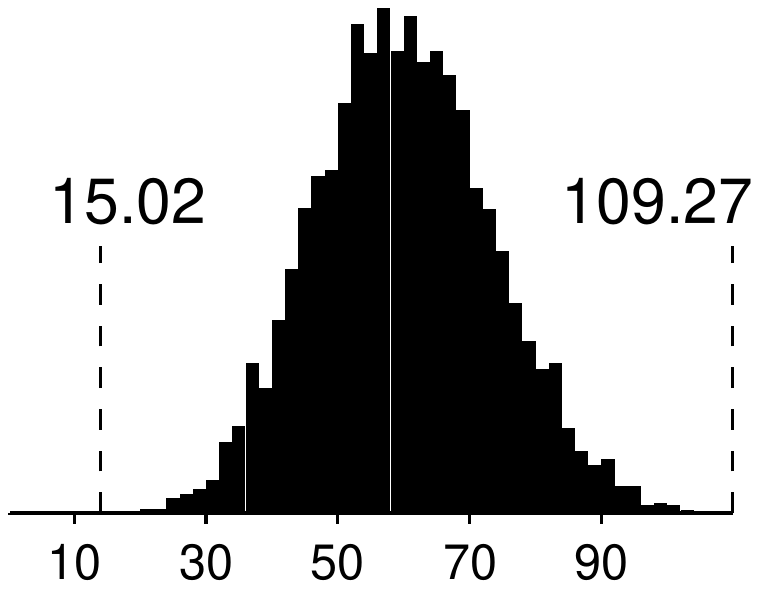}}
    \put(180,3){%
      \begin{minipage}[b]{60pt}
        \centering
        {\small $p = 4.88\%$\\
        $n = 108$\\
        $\sigma = 13.393$}
      \end{minipage}}
    \put(0,5){%
      \includegraphics[width=170pt, trim=36pt 267pt 10pt 249pt, clip]%
      {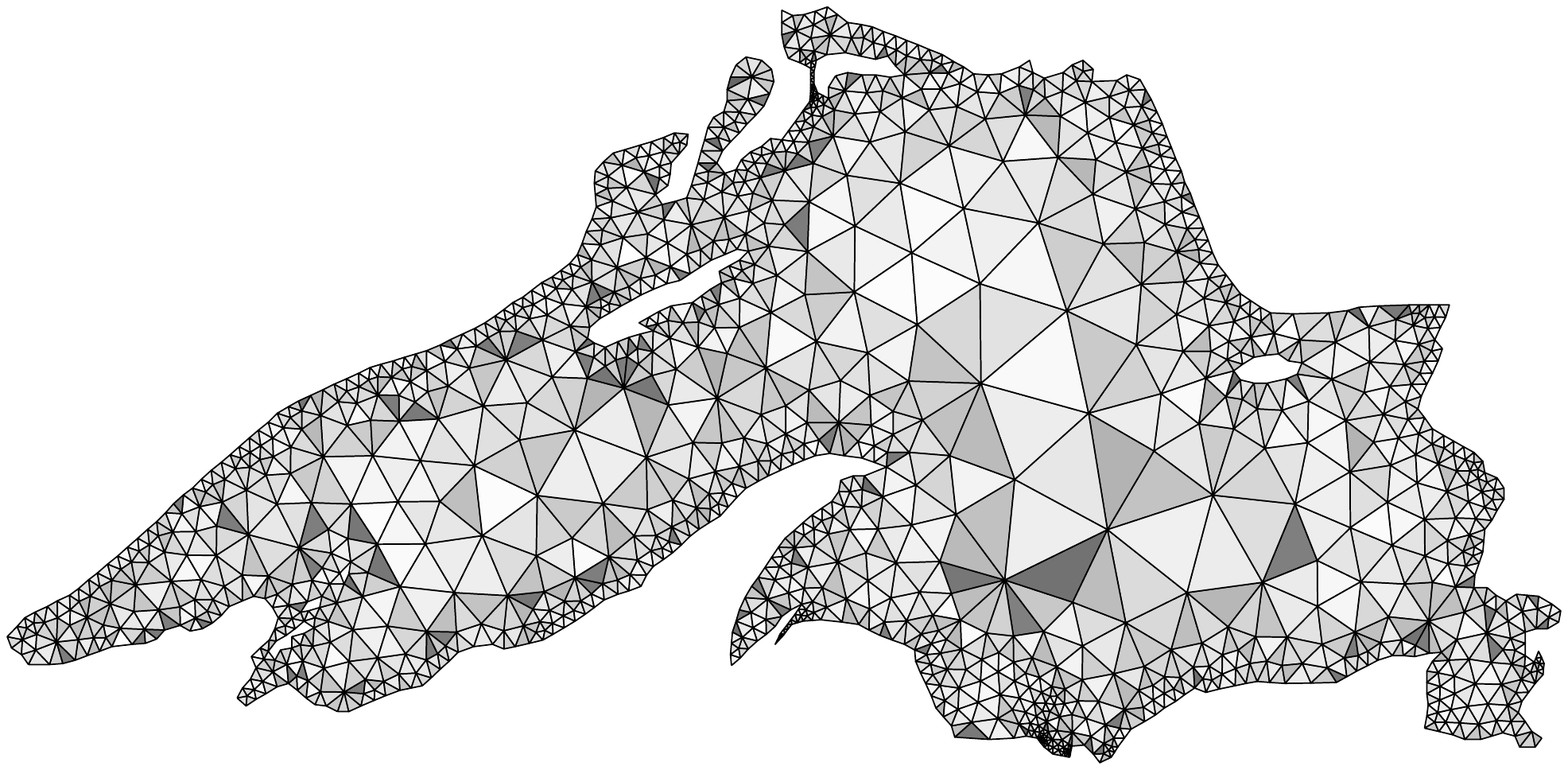}}
  \end{picture}
  \caption{Result of applying Laplacian smoothing to the initial
    acute mesh of Lake Superior (left side of
    Fig.~\ref{fig:superior_1388}).  More than $4\%$ of the triangles
    become nonacute, and the maximum angle increases
    to $109.27$\textdegree.}
  \label{fig:superior_1388lplc}  
\end{figure}

\subsection{Colombia, India, and Thailand} \label{subsec:geo}

We end our 2-D experimental results with a collection of
three large meshes of complicated geographical domains.
The experiments are summarized in \ref{fig:geomeshes}.
For each of these meshes the optimization started by minimizing
$\widetilde{E}_{8}$ for $500$ iterations and then proceeded
by minimizing $E_{8}$, running $500$ iterations at a time
until the mesh became well-centered.  After the mesh became
well-centered, we used one more round of $500$ iterations
minimizing $E_{8}$ to get some additional improvement in the
angle distribution.

\begin{figure}
  \centering
  \begin{picture}(370,381)

    \put(0, 316){%
      \includegraphics[height=64pt, trim=206pt 297pt 186pt 312pt, clip]%
      {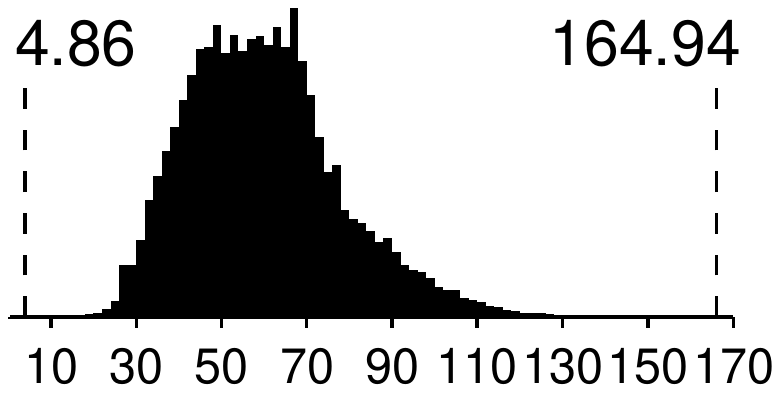}}
    \put(0,304){%
      \begin{minipage}[t]{76pt}
        \centering
        {\small $p = 19.96\%$\\
        $n = 7631$\\
        $\sigma = 18.172$}
      \end{minipage}}
    \put(77,260){%
      \begin{minipage}[b]{108pt}
        \centering
        \includegraphics[height=120pt, trim=607pt 179pt 535pt 126pt, clip]%
        {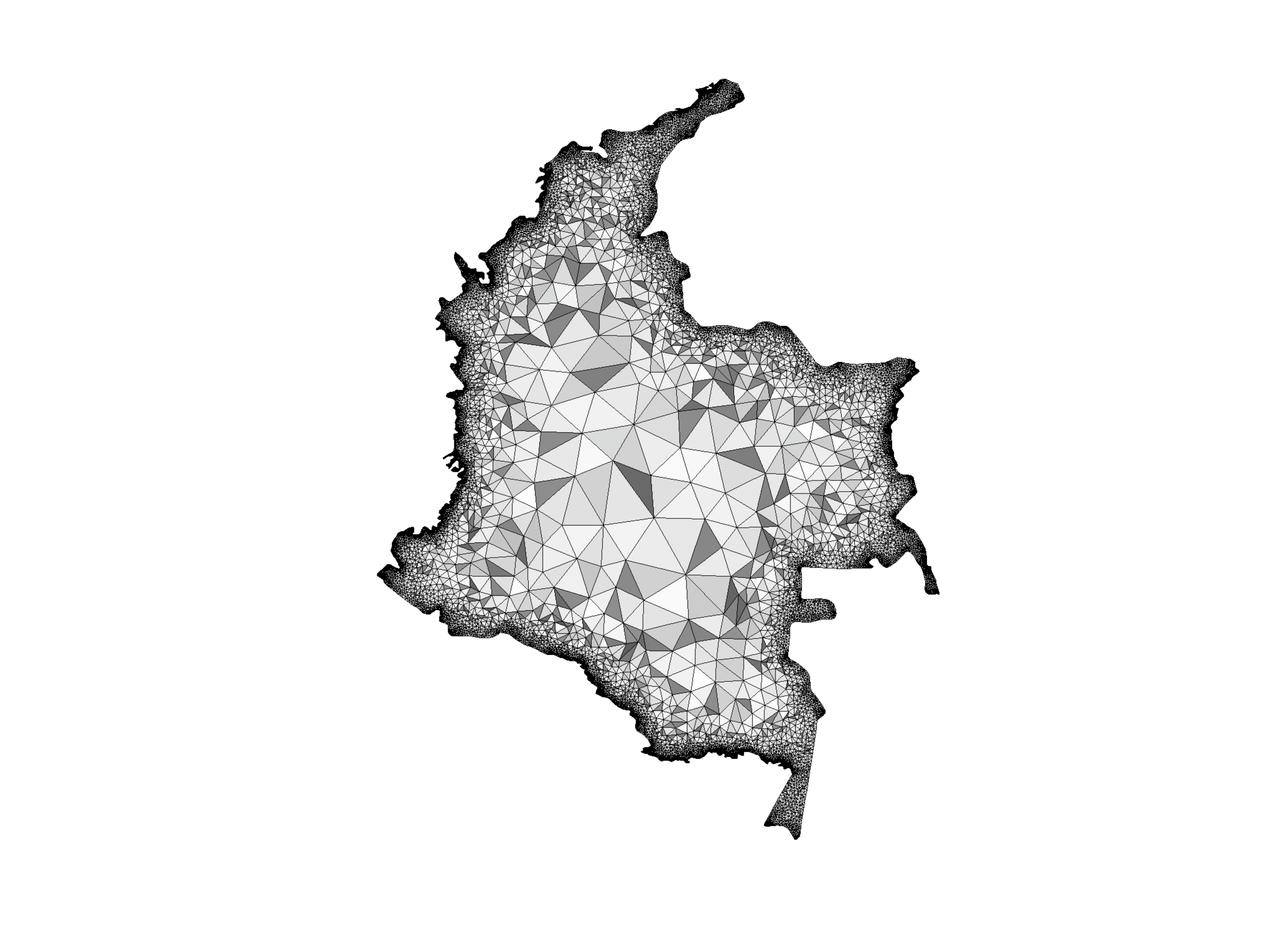}
      \end{minipage}}
    \put(185,260){%
      \begin{minipage}[b]{108pt}
        \centering
        \includegraphics[height=120pt, trim=607pt 179pt 535pt 126pt, clip]%
        {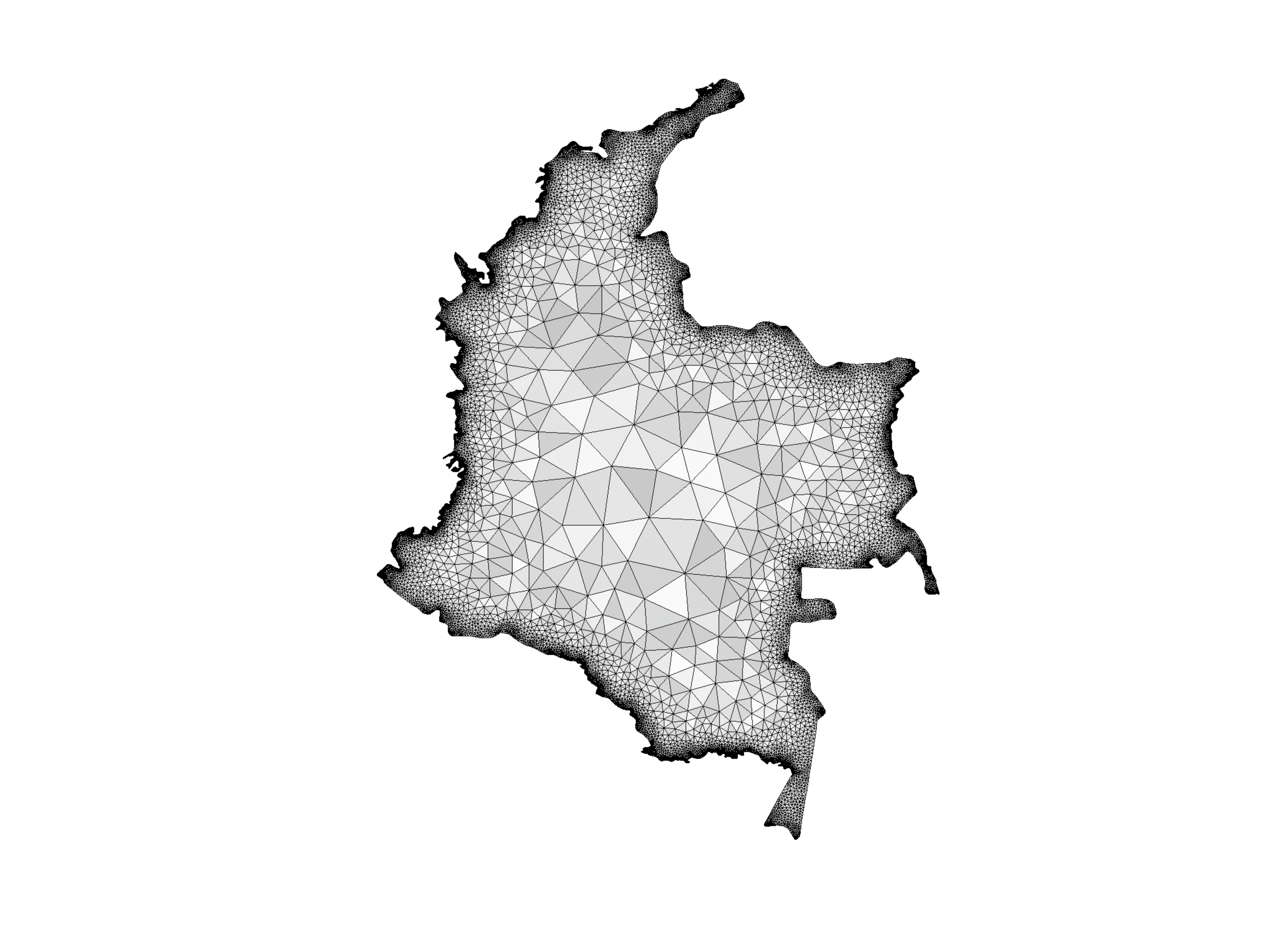}
      \end{minipage}}
    \put(293, 316){%
      \includegraphics[height=64pt, trim=206pt 297pt 186pt 312pt, clip]%
      {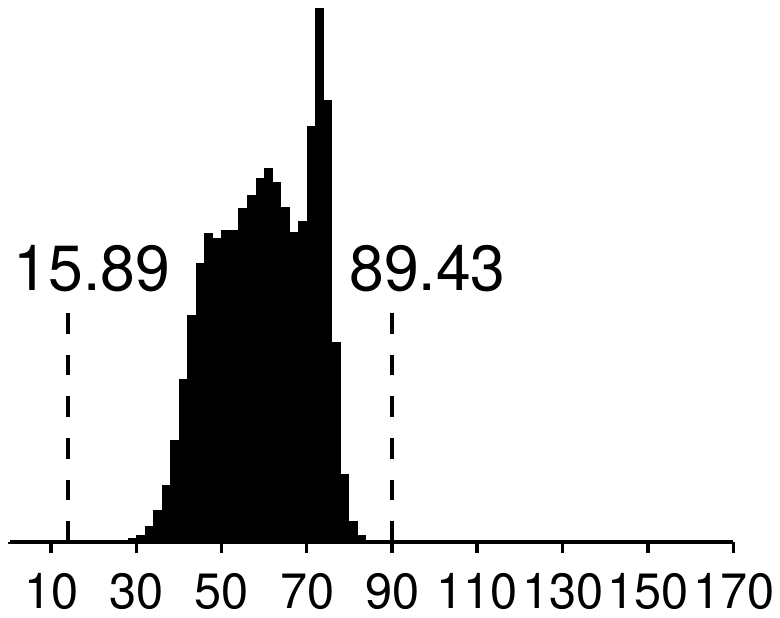}}
    \put(293,304){%
      \begin{minipage}[t]{76pt}
        \centering
        {\small $p = 0.00\%$\\
        $n = 0$\\
        $\sigma = 11.278$}
      \end{minipage}}
    \put(0,186){%
      \includegraphics[height=60pt, trim=206pt 297pt 186pt 312pt, clip]%
      {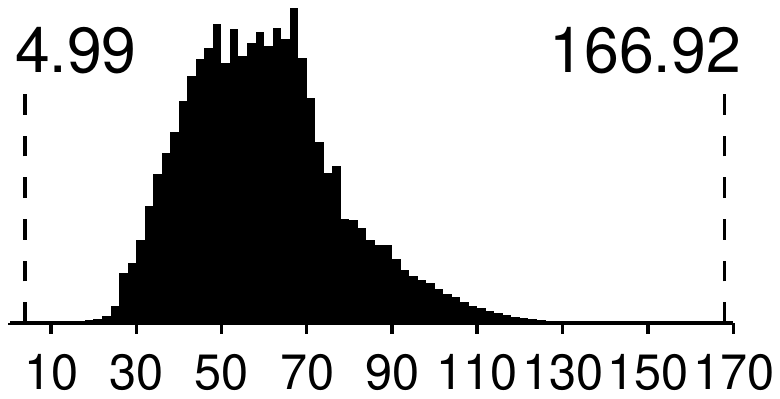}}
    \put(0,174){%
      \begin{minipage}[t]{76pt}
        \centering
        {\small $p = 20.04\%$\\
        $n = 12497$\\
        $\sigma = 18.180$}
      \end{minipage}}
    \put(77,130){%
      \begin{minipage}[b]{108pt}
        \centering
        \includegraphics[height=120pt, trim=512pt 179pt 439pt 126pt, clip]%
        {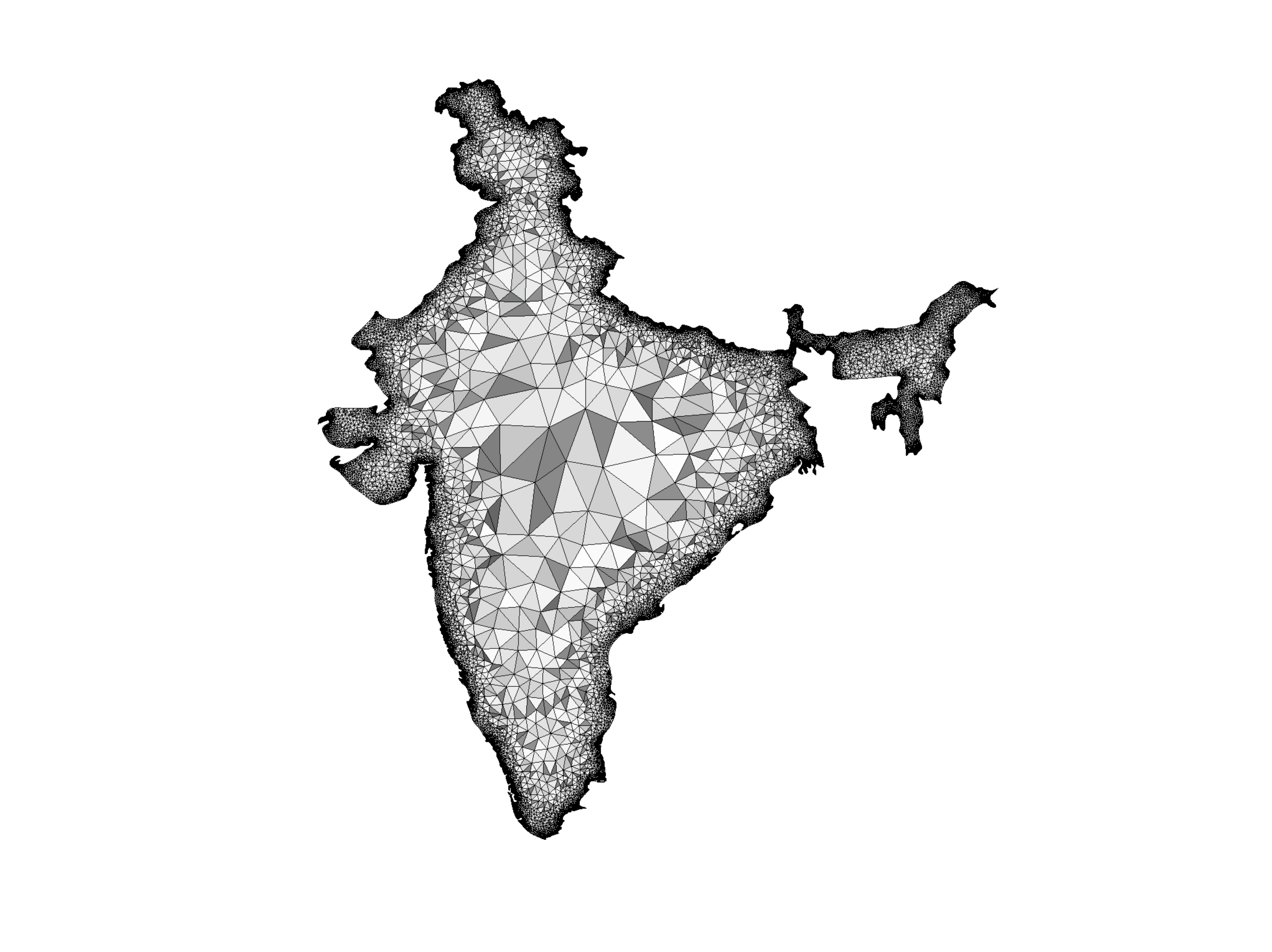}
      \end{minipage}}
    \put(185,130){%
      \begin{minipage}[b]{108pt}
        \centering
        \includegraphics[height=120pt, trim=512pt 179pt 439pt 126pt, clip]%
        {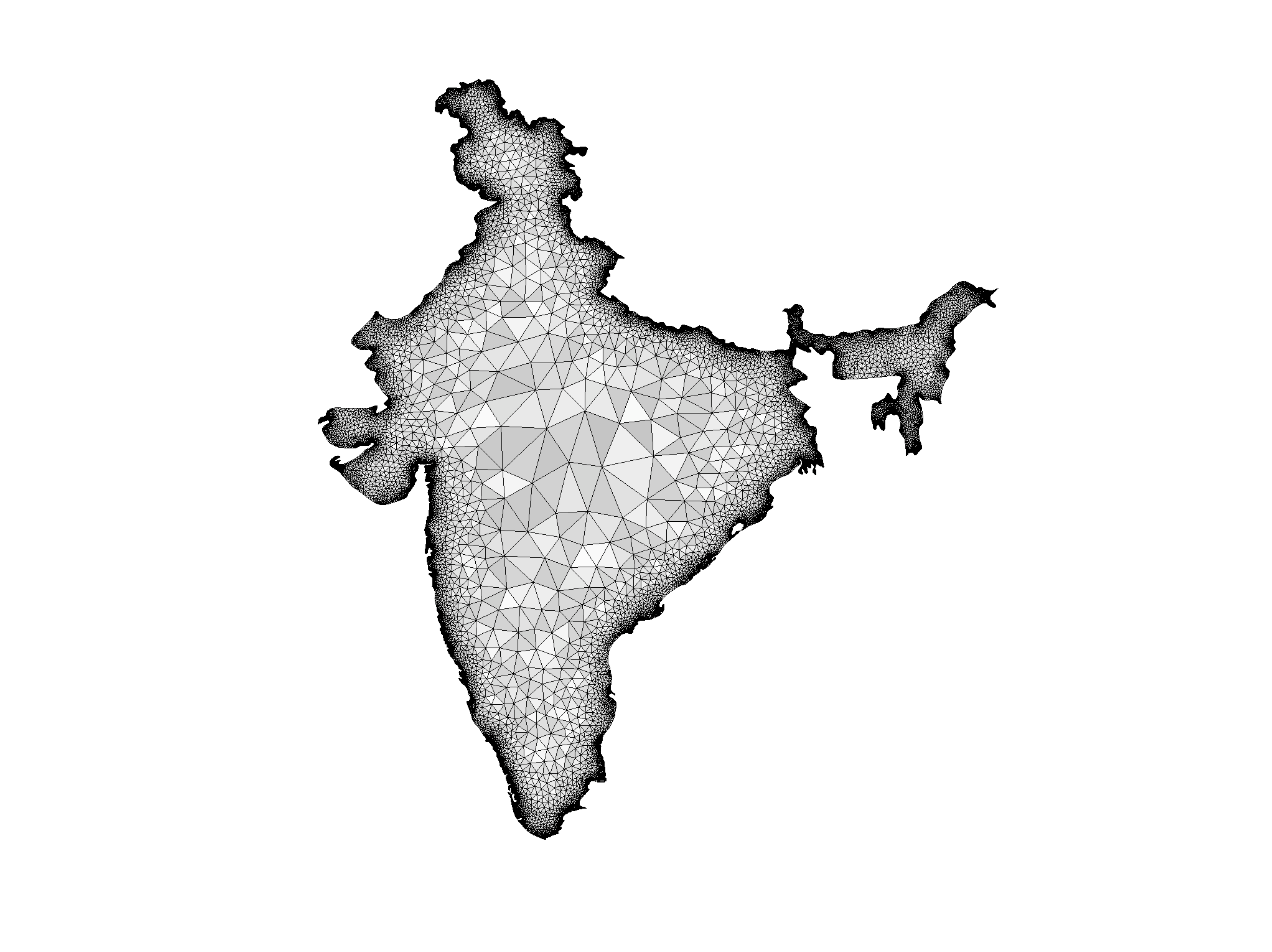}
      \end{minipage}}
    \put(293,186){%
      \includegraphics[height=60pt, trim=206pt 297pt 186pt 312pt, clip]%
      {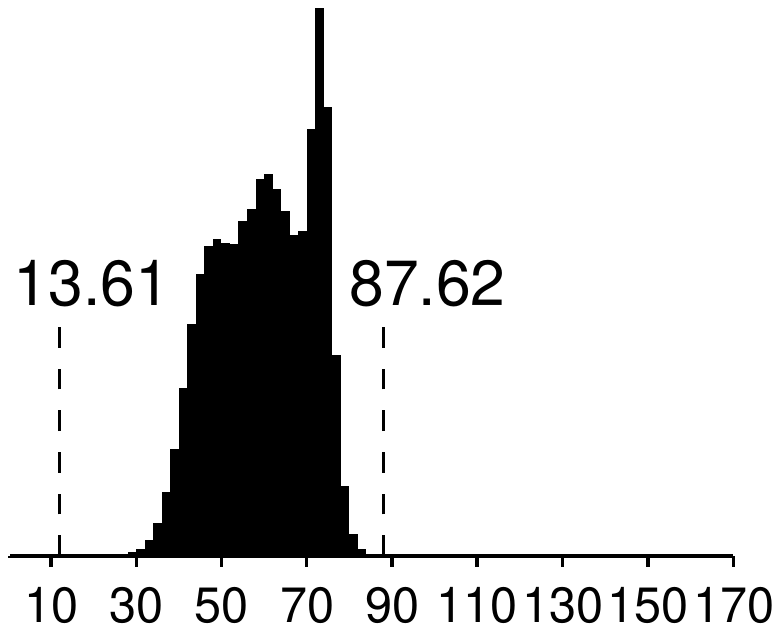}}
    \put(293,174){%
      \begin{minipage}[t]{76pt}
        \centering
        {\small $p = 0.00\%$\\
        $n = 0$\\
        $\sigma = 11.295$}
      \end{minipage}}
    \put(0,56){%
      \includegraphics[height=60pt, trim=206pt 297pt 186pt 312pt, clip]%
      {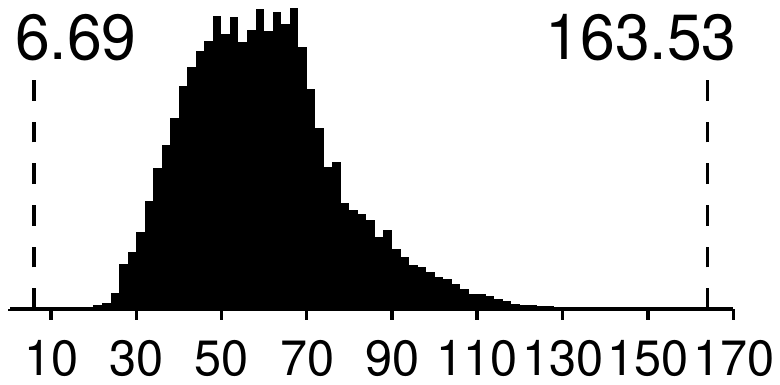}}
    \put(0,44){%
      \begin{minipage}[t]{76pt}
        \centering
        {\small $p = 19.43\%$\\
        $n = 6715$\\
        $\sigma = 17.999$}
      \end{minipage}}
    \put(77,0){%
      \begin{minipage}[b]{108pt}
        \centering
        \includegraphics[height=120pt, trim=733pt 179pt 660pt 126pt, clip]%
        {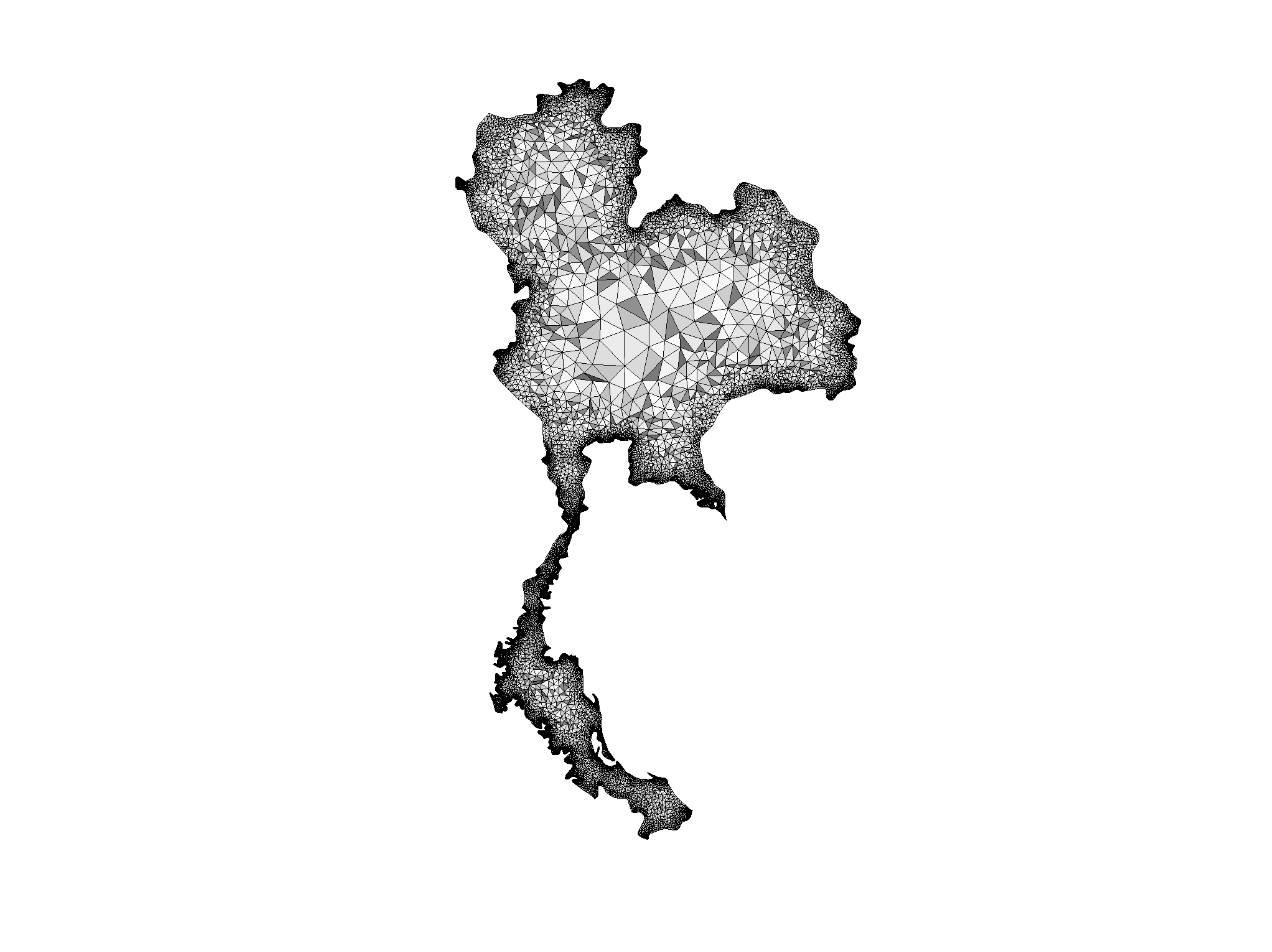}
      \end{minipage}}
    \put(185,0){%
      \begin{minipage}[b]{108pt}
        \centering
        \includegraphics[height=120pt, trim=733pt 179pt 660pt 126pt, clip]%
        {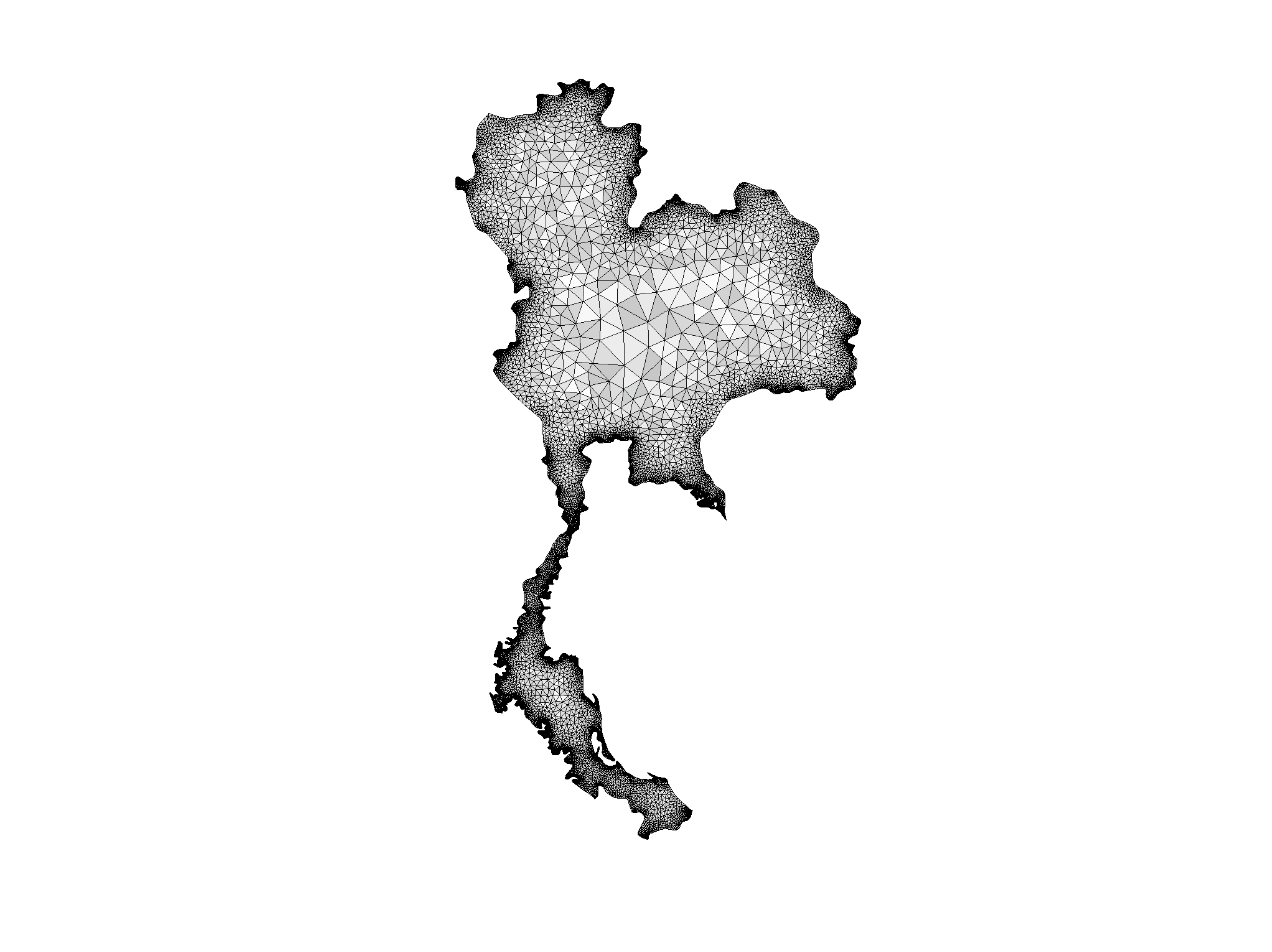}
      \end{minipage}}
    \put(293,56){%
      \includegraphics[height=60pt, trim=206pt 297pt 186pt 312pt, clip]%
      {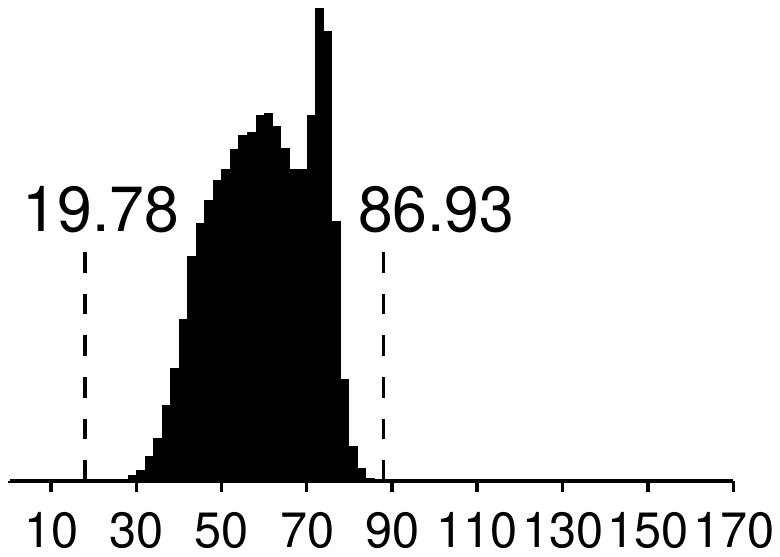}}
    \put(293,44){%
      \begin{minipage}[t]{76pt}
        \centering
        {\small $p = 0.00\%$\\
        $n = 0$\\
        $\sigma = 11.583$}
      \end{minipage}}
  \end{picture}
  \caption{These meshes of complicated geographical boundaries
    were optimized with an initial $500$ iterations of
    $\widetilde{E}_{8}$, followed by successive rounds of
    $500$ iterations minimizing the basic $E_{8}$.  The
    optimization produces well-centered meshes that preserve
    the grading of the input meshes.  (The dark regions
    near the boundaries of the meshes come from the agglomeration
    of the edges of triangles that are
    too small to be seen.)
    The data for the
    geographical boundaries was produced using the
    {\texttt{CountryData[]}} command of Mathematica.
    Initial meshes were constructed from the
    input polygons using Triangle
    \cite{Shewchuk1996} and heuristics
    for improving the mesh connectivity.}
  \label{fig:geomeshes}  
\end{figure}

The total number of iterations for the meshes was
$2000$ iterations for Colombia, $3500$ iterations for
India, and $3000$ iterations for Thailand, with total
optimization times of $5284.01$ seconds, $16162.20$ seconds,
and $8263.82$ seconds.  The meshes are quite large, with
$38233$ triangles, $62370$ triangles, and $34562$ triangles
respectively.  In each case, more than $19\%$ of the triangles
are nonacute in the initial mesh, and the maximum angle is
larger than $160$\textdegree, yet the optimization finds a
well-centered result.  It is also clear that the optimization
preserves the gradual change in element size from the tiny triangles
needed to resolve the boundaries to the much larger triangles
in the interiors of the meshes.


\subsection{3D Meshes} \label{subsec:3dresult}

For tetrahedral meshes, the question of whether the mesh connectivity
permits a well-centered mesh is more difficult than its
two-dimensional analogue \cite{VaHiGuRaZh2008}. In part because we do
not yet have an effective preprocessing algorithm for tetrahedral
meshes, many of our optimization experiments in three dimensions have
been limited to meshes with carefully designed mesh connectivity.
The mesh shown in Fig.~\ref{fig:cube_430} is one of these meshes.
The shading of the tetrahedral elements in Fig.~\ref{fig:cube_430}
represents the shadows that would result from viewing
the faceted object under a light source; it has nothing
to do with the quality of the elements of the mesh.  The full
mesh is a mesh of the three-dimensional cube with $430$ tetrahedra.
Figure~\ref{fig:cube_430} uses a cutaway view to display some
of the elements in the interior of the mesh.

Although the initial mesh was carefully designed to have good
mesh connectivity (e.g., each vertex has at least $10$ incident
edges) and a high-quality surface mesh, it was not $3$-well-centered.
In fact, $22.33\%$ of the tetrahedra are not $3$-well-centered.
Optimizing $E_{16}$ for $3.92$ seconds ($20$ iterations) produced
a $3$-well-centered mesh.  Even though the initial mesh was carefully
designed, the optimization result is nontrivial.  We compared
optimization of $E_{16}$ to the Mesquite implementation of Laplacian
smoothing, applying Laplacian smoothing to the initial mesh and
running it until it converged after $60$ iterations
($0.14$ seconds).  The result of Laplacian
smoothing is a mesh in which $22.33\%$ of the tetrahedra are not
$3$-well-centered.  Figure~\ref{fig:cube_430} includes
the $h(v, \sigma)/R(\sigma)$ distributions for the initial mesh,
the mesh after optimizing $E_{16}$, and the mesh resulting from
Laplacian smoothing.  Near each histogram we show the percentage $p$
and number $n$ of tetrahedra (not $h/R$ values) that are not
$3$-well-centered, and we report the mean $\mu$ and standard
deviation $\sigma$ of the distribution of $h/R$ values.

It is worth noting that, because of its
difficulty, obtaining well-centered
triangulations and/or acute triangulations of $3$-dimensional
objects is significant no matter how they are obtained.
In our other work we have made use of the optimization
techniques developed in this paper to
construct well-centered triangulations of several
simple three-dimensional shapes \cite{VaHiGu2008} and to
constructively prove the existence of an acute triangulation
of the $3$-dimensional cube \cite{VaHiZhGu2009}, solving
an open problem mentioned in \cite{EpSuUn2004} and
\cite{BrKoKrSo2009}.

\begin{figure}
  \centering
  \begin{picture}(346, 154)
    \put(76, 39){\begin{minipage}[c]{90pt}
        \includegraphics%
        [width = 90pt, trim = 179pt 297pt 195pt 311pt, clip]%
        {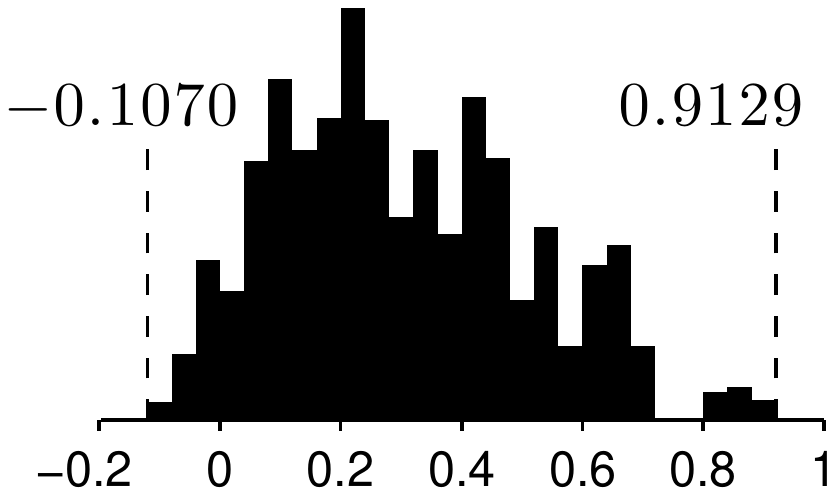}
      \end{minipage}}
    \put(0, 36){\begin{minipage}[c]{76pt}
        \centering
        {\small Laplacian Smoothing\\
        $p = 22.33\%$\\
        $n = 96$\\
        $\mu = 0.30223$\\
        $\sigma = 0.20807$}
      \end{minipage}}
    \put(76, 116){\begin{minipage}[c]{90pt}
        \includegraphics%
        [width = 90pt, trim = 179pt 297pt 195pt 311pt, clip]%
        {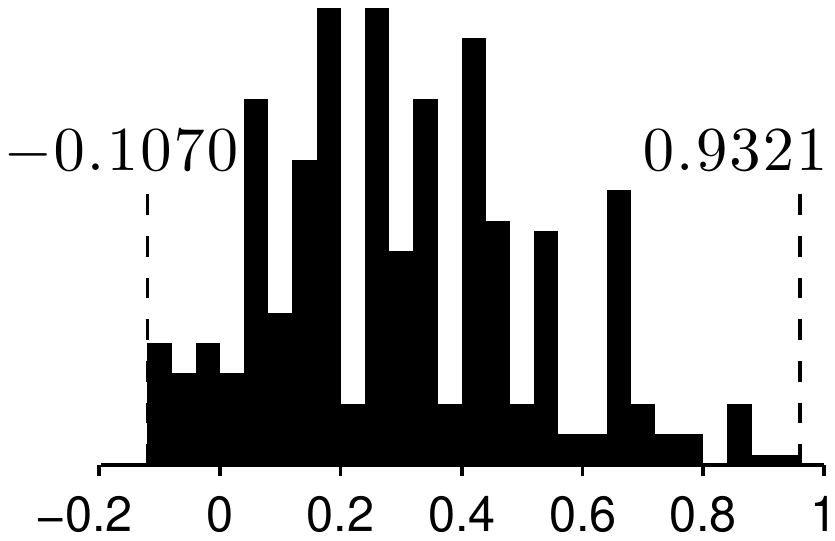}
      \end{minipage}}
    \put(0, 117){\begin{minipage}[c]{76pt}
        \centering
        {\small Initial Mesh\\
        $p = 22.33\%$\\
        $n = 96$\\
        $\mu = 0.30052$\\
        $\sigma = 0.21966$}
      \end{minipage}}
    \put(210, 0){%
      \includegraphics[width= 90pt, trim= 417pt 246pt 344pt 120pt, clip]%
      {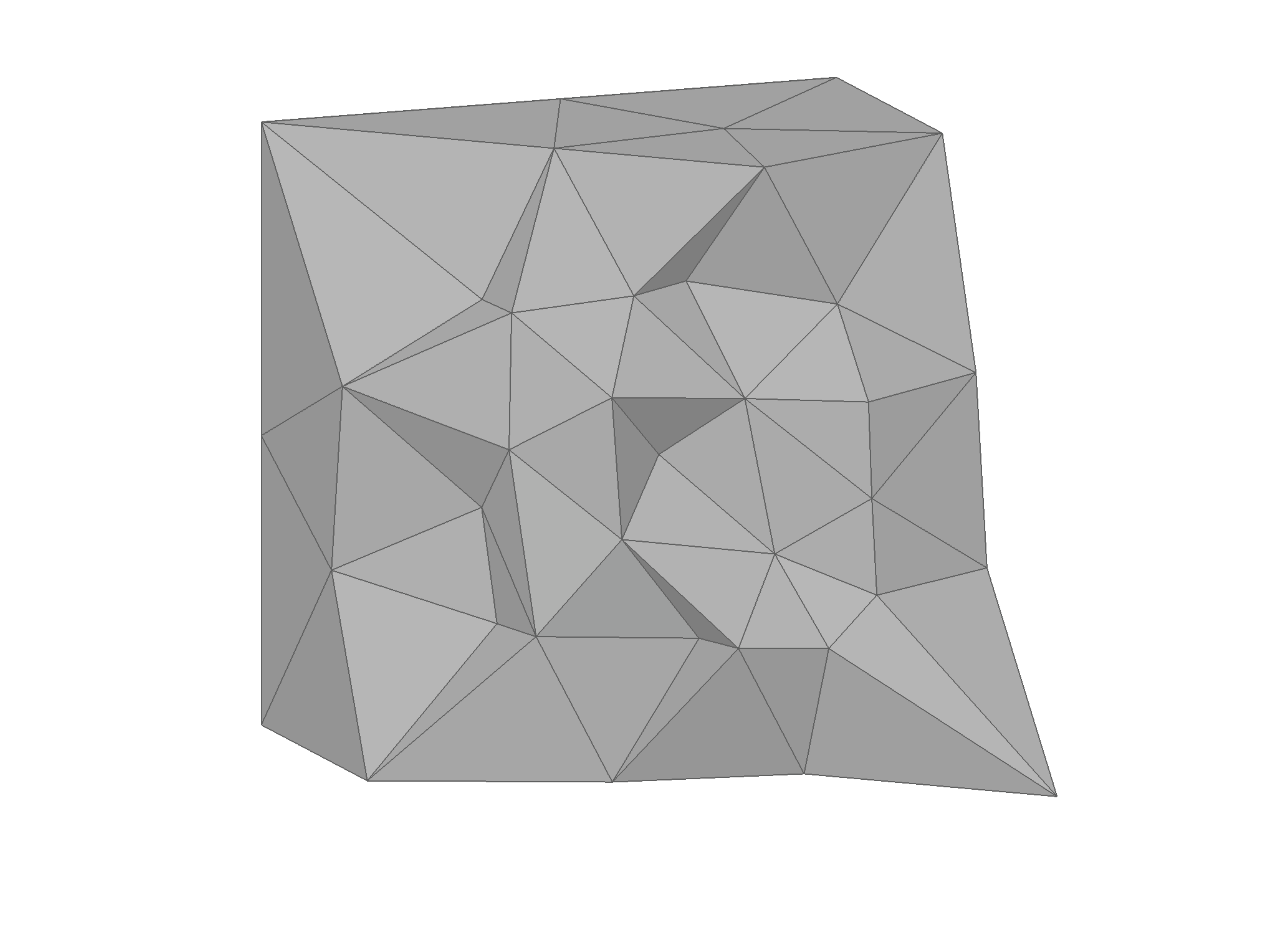}}
    \put(170, 116){\begin{minipage}[c]{90pt}
        \includegraphics%
        [width = 90pt, trim = 179pt 297pt 195pt 311pt, clip]%
        {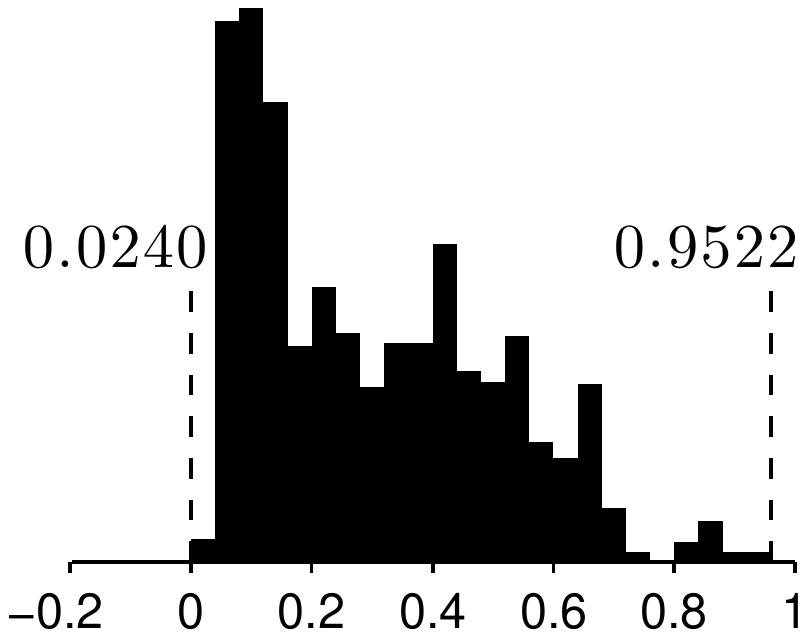}
      \end{minipage}}
    \put(260, 117){\begin{minipage}[c]{76pt}
        \centering
        {\small $E_{16}$ Optimization\\
        $p = 0.00\%$\\
        $n = 0$\\
        $\mu = 0.30196$\\
        $\sigma = 0.20442$}
      \end{minipage}}
  \end{picture}
  \caption{A cutout view showing the interior of a $3$-well-centered
    mesh of the cube.  The mesh is the result of $3.92$ seconds 
    ($20$ iterations) of optimizing $E_{16}$ on an initial mesh
    for $22.33\%$ of the tetrahedra were not $3$-well-centered.
    Recall that a tetrahedron
    $\sigma^{3}$ is $3$-well-centered if and only if
    $h(v, \sigma^{3})/R(\sigma^{3}) > 0$ for each
    vertex $v$ of $\sigma^{3}$.
    For a regular tetrahedron, $h/R = 1/3$.
    The $h/R$ distributions for the initial mesh, the result of
    optimizing $E_{16}$, and the result of Laplacian smoothing show
    the superiority of our method for finding
    $3$-well-centered meshes.}
  \label{fig:cube_430}
\end{figure}

\section{Conclusions and Research Questions}
\label{sec:conclusions}



This paper shows that an $n$-well-centered simplex can be
characterized in terms of the equatorial balls of its facets and uses
this alternate characterization to prove that an $n$-well-centered
mesh in $\Real^{n}$ is a Delaunay mesh.  The paper introduces the
related cost functions $E_{\infty}$ and $E_{p}$ that quantify the
well-centeredness of triangulations in any dimension, extending the
function introduced in \cite{VaHiGuRa2007}.  Some properties of the
cost function are discussed, and it is shown that a cost function
quantifying well-centeredness must be nonconvex.

After introducing the cost function, the paper shows that the minmax
angle triangulation is the optimal triangulation with respect to the
$E_{\infty}$ energy and discusses why our algorithm uses the local
preprocessing algorithm of \cite{VaHiGuRa2007} instead of computing
the maxmin triangulation after each step of optimization.  The
discussion raises the interesting research question of how to
efficiently compute (and recompute)
a triangulation that minimizes the
maximum angle among triangulations with no lonely vertices.

The task
of developing a local preprocessing algorithm that works in
dimensions higher than $2$ is another important research objective.
A simple and complete characterization of the mesh connectivity
requirements for a vertex and its one-ring in a tetrahedral mesh in
$\Real^3$ to be $3$-well-centered would be helpful. We have made a
start for such a characterization in~\cite{VaHiGuRaZh2008}, where we
have discovered some beautiful connections to the triangulation of the
spherical link of the one ring.

The experiments of Section~\ref{sec:results} show that the proposed cost
function can be effective in finding a well-centered triangulation for
meshes that permit such triangulations.  The optimization problem in
the context of our nonconvex cost functions $E_{p}$ is a difficult
problem, though, and Mesquite does not always find a global minimum of
the energy. While it is easy to show that our gradient descent type
algorithm converges to a local stationary point, it would be nice to
have an optimization method guaranteed to find a global minimum of the
energy. This however is a very hard problem and typical of the
difficulties faced by other iterative algorithms for mesh
optimization. For example, for the vastly popular iterative algorithms
for centroidal Voronoi tessellations \cite{DuFaGu1999} and their
variations \cite{DuGuJu2003, DuWa2005}, restricted convergence results
have only recently started appearing
\cite{DuEmJu2006,EmJuRa2008}. Similarly, a convergence proof for
variational tetrahedral meshing \cite{AlCoYvDe2005} is known for only
one rings, although the algorithm is very useful in practice.

It would also be worthwhile to improve the efficiency of
our optimization.  In particular, it would be interesting
to study methods for localizing the energy and applying
optimization in only those specific areas where
it is needed. Besides possibly making the optimization
more efficient, localizing the energy would make it
easier to parallelize the algorithm.  The experiment
in Section~\ref{subsec:twoholes} that made the optimization easier
by repositioning the boundary vertices suggests that
using a constrained optimization with boundary
vertices free to move along the boundary could make the
optimization more effective.

It is also possible that the cost function could be improved.  Using a
linear combination of $E_{\text{imr}}$ with $E_{p}$ was effective for
the two holes mesh of Section~\ref{subsec:twoholes} and
the geographical meshes in Section~\ref{subsec:geo},
but the coefficients of the linear combination
were chosen quite arbitrarily, and there may be other, better ways to
prevent element inversion.  There were also some experiments which
needed to use $E_{p}$ with more than one parameter $p$ in order to
find a nice result.  Taking a linear combination of $E_{p}$ for
different powers of $p$ might be effective for those situations and
perhaps more generally.

Since the original submission of this
manuscript, the authors have become aware
that Sazanov et al. generated a $3$-well-centered mesh
of a spherical layer by repeating the near-boundary triangulation
of their mesh stitching approach without stitching to an ideal
mesh \cite{SaHaMoWe2007}.  Generalizing their construction
to more complicated 3-D domains is another interesting direction
for research.

To summarize the paper briefly,
our generalized characterization of well-cen\-tered\-ness
offers, for the first time, a direction in which planar acute
triangulations may be generalized. More complex three dimensional
experiments will have to await a better preprocessing and better
mathematical understanding of the topological obstructions to
well-centeredness.

We believe we have shown enough evidence in this and related
publications that one can produce simple three dimensional
well-centered tetrahedral meshes. In planar domains, it is already
possible to produce well-centered triangulations with or without holes
and gradations, for complex domains. It is
also possible to improve triangulations
that are already acute. Like many other
successful mesh optimization algorithms, a convergence theory for
well-centered meshing will be discovered eventually, we hope, either
by us or by other researchers. For further developments, we felt the
need to make available the evidence that well-centered meshes are now
possible for experiments, and that there is a useful characterization
theory for such meshes.

\section*{Acknowledgment}
We would like to thank Vadim Zharnitsky for useful discussions.
We also thank Hale Erten and Alper \"Ung\"or for providing us
with the Lake Superior mesh.  Lastly, we thank
the reviewers; their comments led to improvement of the paper.

\bibliographystyle{acmurldoi}
\bibliography{wct}

\end{document}